\documentclass{easychair}

\usepackage{booktabs}
\usepackage{subcaption}
\usepackage{latexsym}
\usepackage{setspace}
\usepackage{cancel}
\usepackage{listings}
\usepackage{graphicx}
\usepackage{appendix}
\usepackage{amssymb}
\usepackage{stmaryrd}
\usepackage{amsmath}
\usepackage{leftidx}
\usepackage{mathtools}
\usepackage{paralist}
\usepackage{color}
\usepackage{mathrsfs}
\usepackage{tikz}
\usepackage[draft]{minted}
\usetikzlibrary{shapes,automata,positioning}
\usepackage[linesnumbered,ruled,noend]{algorithm2e}
\usepackage{wrapfig}

\newcommand{\set}[1]{\{ #1 \}}

\newcommand{\Nat}{\ensuremath{\mathbb{N}}}

\newcommand{\eqdef}{\stackrel{\mbox{\begin{tiny}def\end{tiny}}}{=}} 

\newcommand {\emptyword}{\ensuremath{\epsilon}}

\newcommand {\pspace} {\textsc{pspace}}

\newcommand {\expspace} {\textsc{expspace}}
\newcommand {\exptime} {\textsc{exptime}}

\newcommand {\nonelementary} {\textsc{non-elementary}}

\newcommand {\elementary} {\textsc{elementary}}

\newcommand{\cut}[1]{}

\mathchardef\mhyphen="2D 

\newcommand{\hide}[1]{}

\newcommand\cA{\mathcal{A}}
\newcommand\cB{\mathcal{B}}

\newcommand\cE{\mathcal{E}}
\newcommand\cG{\mathcal{G}}
\newcommand\Ll{\mathcal{L}}

\newcommand\replaceall{\mathsf{replaceAll}}
\newcommand\sreplaceall{\mathsf{sreplaceAll}}
\newcommand\reverse{\mathsf{reverse}}
\newcommand\indexof{\mathsf{IndexOf}}

\newcommand\revsym{\pi}

\newcommand\strline{\mathsf{SL}}

\newcommand\vars{\mathsf{Vars}}

\newcommand{\ASSERT}[1]{\textbf{assert}(#1)}

\newcommand{\arity}{r}

\newcommand{\FA}{FA}
\newcommand{\FFA}{2FA}

\newcommand{\FT}{FT}
\newcommand{\FFT}{2FT}
\newcommand{\FunFT}{FFT}

\newcommand{\PT}{PT}

\newcommand{\ialphabet}{\Sigma}

\newcommand{\EndLeft}{\ensuremath{\vartriangleright}}
\newcommand{\EndRight}{\ensuremath{\vartriangleleft}}

\newcommand{\Lang}{\mathcal{L}}
\newcommand{\Tran}{\mathcal{T}}

\newcommand{\Aut}{\ensuremath{\mathcal{A}}}

\newcommand{\Transducer}{\ensuremath{T}}
\newcommand{\controls}{\ensuremath{Q}}
\newcommand{\finals}{\ensuremath{F}}
\newcommand{\transrel}{\ensuremath{\delta}}

\newcommand{\Left}{\ensuremath{-1}}
\newcommand{\Right}{\ensuremath{1}}

\newcommand{\defn}[1]{\emph{#1}}

\newcommand{\conacc}{\Omega}

\newcommand{\reginvrel}{\textbf{RegInvRel}}

\newcommand{\prerec}{\reginvrel}

\newcommand{\regmondec}{\textbf{RegMonDec}}

\newcommand\rcdim{\mathsf{art}}

\newcommand\rcdep{\mathsf{asgn}}

\newcommand\rcasrt{\mathsf{asrt}}

\newcommand\rcphi{\mathsf{fnsize}}
\newcommand\rcpsi{\mathsf{fasize}}

\newcommand{\Pre}{\textsf{Pre}}

\newcommand\bigO{\mathcal{O}}

\newcommand\tiles{\Theta}
\newcommand\hrel{H}
\newcommand\vrel{V}
\newcommand\tile{t}
\newcommand\inittile{t_I}
\newcommand\fintile{t_F}
\newcommand\tileheight{h}
\newcommand\tilewidth{\ell}

\newcommand\expheight{n}
\newcommand\linlen{m}
\newcommand\tilesnum[1]{\Theta_{#1}}
\newcommand\hrelnum[1]{H_{#1}}
\newcommand\vrelnum[1]{V_{#1}}
\newcommand\inittilenum[1]{\inittile^{#1}}
\newcommand\fintilenum[1]{\fintile^{#1}}

\newcommand\goodnums[2]{\ap{S_{#1}}{#2}}

\newcommand\tenc[2]{[#2]_{#1}}

\newcommand\nmax[1]{\text{MAX}_{#1}}

\newcommand\numeq{\bowtie}
\newcommand\numplus{\oplus}
\newcommand\numsep{\#}
\newcommand\passsep{\natural}

\newcommand\tilerow{r}

\newcommand{\tup}[1]{\left( #1 \right)}

\newcommand\ap[2]{{#1}\mathord{\brac{#2}}}

\newcommand{\opset}{\mathscr{O}}

\newcommand{\strlineconcat}{$\strline$[$\concat$, $\sreplaceall$, $\reverse$, \FT]}

\newcommand{\strlinefft}{$\strline$[$\concat$, $\replaceall$, $\reverse$, \FunFT]}

\newcommand\brac[1]{\left(#1\right)}

\newcommand\setcomp[2]{\left\{{#1}\ \middle|\ {#2}\right\}}

\newcommand\lang[1]{\mathcal{L}\mathord{\brac{#1}}}

\newcommand\rowdelim{\#}

\newcommand\resetchar{!}
\newcommand\bit{b}

\newcommand\nbit[2]{{#2}_{#1}}

\newcommand\repl[3]{\$^{#3}_{#1, #2}}
\newcommand\replall[1]{\$_{#1}}
\newcommand\simplerepl[2]{\$^{#2}_{#1}}

\newcommand\caleybox[1]{\llparenthesis #1 \rrparenthesis}
\newcommand\internalchar{\flat}

\newcommand\transducerbench{\textsc{Transducer}}
\newcommand\slogbench{\textsc{SLOG}}
\newcommand\slogbenchr{\textsc{SLOG (replace)}}
\newcommand\slogbenchra{\textsc{SLOG (replaceall)}}
\newcommand\kaluzabench{\textsc{Kaluza}}

\newtheorem{fact}{Fact}

\newcommand{\OMIT}[1]{}

\newcommand\shortlong[2]{#2}

\newif\ifdraft\draftfalse
\ifdraft
\newcommand{\anthony}[1]{\color{red} {AL: #1 :LA} \color{black}}
\newcommand{\zhilin}[1]{\color{brown} {ZL: #1 :LZ} \color{black}}
\newcommand{\tl}[1]{\color{blue} {TL: #1 :LT} \color{black}}
\newcommand{\mat}[1]{\color{cyan} {MH: #1 :HM} \color{black}}

\else
\newcommand{\anthony}[1]{}
\newcommand{\zhilin}[1]{}
\newcommand{\tl}[1]{}
\newcommand{\mat}[1]{}
\fi

\newcommand{\concat} {\circ}
\newcommand{\replace} {{\sf replace}}
\newcommand{\str} {{\sf Str}}
\newcommand{\intnum} {{\sf Int}}
\newcommand{\regexp} {{\sf RegExp}}

\newtheorem{theorem}{Theorem}[section]
\newtheorem{remark}[theorem]{Remark}
\newtheorem{definition}[theorem]{Definition}
\newtheorem{example}[theorem]{Example}
\newtheorem{proposition}[theorem]{Proposition}
\newtheorem{lemma}[theorem]{Lemma}

\begin{document}

\title{Decision Procedures for Path Feasibility of String-Manipulating Programs with Complex Operations}
\titlerunning{String-Manipulating Programs with Complex Operations}

\author{
    Taolue Chen \and
    Matthew Hague \and
    Anthony W. Lin \and
    Philipp R\"ummer \and
    Zhilin Wu
}
\authorrunning{
    T. Chen \and
    M. Hague \and
    A. W. Lin \and
    P. R\"ummer \and
    Z. Wu
}

\institute{
  Birkbeck, University of London
  \and
  Royal Holloway, University of London
  \and
  University of Oxford
  \and
  Uppsala University
  \and
  Institute of Software, Chinese Academy of Sciences
}

\maketitle


\begin{abstract}
    The design and implementation of decision procedures for checking path 
    feasibility in string-manipulating programs is an important problem, with 
    such applications as symbolic execution 
    of programs with strings and automated detection of cross-site scripting
    (XSS) vulnerabilities in web applications. A (symbolic) path is given as a
    finite sequence of assignments and assertions (i.e. without loops), and 
    checking 
    its feasibility amounts to determining the existence of inputs that yield
    a successful execution.
    Modern 
    programming languages
    (e.g. JavaScript, PHP, and Python) support many complex string
    operations, and strings are also often implicitly modified during
    a computation in some intricate fashion (e.g. by some autoescaping 
    mechanisms). 
    
    In this paper we provide two general semantic conditions which
    together ensure the decidability of path feasibility: (1) each
    assertion admits regular monadic decomposition (i.e.\ is an effectively recognisable relation), and (2)
    each assignment uses a (possibly nondeterministic) function whose inverse 
    relation preserves regularity. We show that the semantic conditions are 
    \emph{expressive}
    since they are satisfied by a multitude of string operations including
    concatenation, one-way and two-way finite-state transducers, $\replaceall$
    functions (where the replacement string could contain variables), 
    string-reverse functions, regular-expression matching, and 
    some (restricted) forms of letter-counting/length functions. 
    The semantic conditions also strictly subsume existing decidable
    string theories (e.g. straight-line fragments, and acyclic logics), and 
    most existing benchmarks (e.g. most of Kaluza's, and all of SLOG's, 
    Stranger's, and SLOTH's benchmarks). Our semantic conditions also yield
    a conceptually \emph{simple} decision procedure, as well as an 
    \emph{extensible} architecture of a string solver
    in that a user may easily incorporate his/her own string functions into the
    solver by simply providing code for the pre-image computation without 
    worrying about other parts of the solver.
    Despite these, the semantic conditions are unfortunately too general to 
    provide a fast and complete decision procedure. We provide strong 
    theoretical evidence for this in the form of complexity results.
    To rectify this problem, we propose two solutions. Our main
    solution is to allow only
    partial string functions (i.e., prohibit nondeterminism) in condition (2).
    This restriction is satisfied in many cases in practice, and
    yields decision procedures that are effective in both theory and
    practice. Whenever nondeterministic functions are still needed (e.g. the
    string function split), our second solution is to provide a syntactic
    fragment that provides a support of nondeterministic functions, and 
    operations like one-way transducers, $\replaceall$ (with constant replacement
    string), the string-reverse function, concatenation, and regular-expression
    matching. We show that this fragment can be reduced to an existing 
    solver SLOTH that exploits fast model checking algorithms like IC3.

    We provide an efficient implementation of our decision procedure 
    (assuming our first solution above, i.e., deterministic partial string 
    functions)  in a new string solver OSTRICH.
    Our implementation provides built-in support for concatenation, reverse, functional transducers (\FunFT{}), and $\replaceall$ and provides a framework for extensibility to support further string functions.
    We demonstrate the 
    efficacy of our new solver against other competitive solvers.
\end{abstract}



\section{Introduction}
\label{sec:intro}

Strings are a fundamental data type in virtually all programming languages.
Their generic nature can, however, lead to many subtle programming
bugs, some with security consequences, e.g., cross-site scripting
(XSS), which is among the OWASP Top 10 Application Security Risks
\cite{owasp17}. One effective
automatic testing method for identifying subtle programming errors
is based on \emph{symbolic execution}
\cite{king76} and combinations with dynamic analysis
called \emph{dynamic symbolic execution} \cite{jalangi,DART,EXE,CUTE,KLEE}.
See \cite{symbex-survey} for an excellent survey. Unlike purely random testing,
which runs only \emph{concrete} program executions on different
inputs, the techniques of symbolic execution analyse \emph{static} paths
(also called symbolic executions) through the software system under test.
Such a path can be viewed as a constraint $\varphi$ (over
appropriate data domains) and the hope is that a fast
solver is available for checking the satisfiability of $\varphi$ (i.e. to check
the \emph{feasibility} of the static path), which can be used for generating
inputs that lead to certain parts of the program or an erroneous behaviour.

\OMIT{
A symbolic execution analysis relies on constraint solvers at its core. When
the system
under test is a program that uses string data type (including most programs
written in scripting languages like JavaScript/PHP), these techniques
crucially need constraint solvers over the string domain (a.k.a. \emph{string
solvers}). This in fact has been one main reason behind the development of
the theory, implementation, and symbolic execution applications of string
solving in the past decade, e.g., see
\cite{BTV09,Berkeley-JavaScript,HAMPI,Stranger,Vijay-length,YABI14,Abdulla14,LB16,fang-yu-circuits,Abdulla17,CCHLW18,HJLRV18,S3,TCJ16,Z3-str,Z3-str2,Z3-str3,cvc4,Saner,RVG12,jalangi,expose}.
}

\OMIT{
\smallskip
\noindent
\textbf{Constraints from Symbolic Executions. }
}
Constraints from symbolic
execution on string-manipulating programs can be understood in terms of the
problem of path
feasibility over a bounded program $S$ with neither loops nor
branching (e.g. see \cite{BTV09}). 
That is, $S$ is a sequence of assignments and conditionals/assertions, i.e.,
generated by the grammar
%
\begin{equation}
    S ::= \qquad y := f(x_1,\ldots,x_\arity) \ |\
    \text{\ASSERT{$g(x_1,\ldots,x_\arity)$}}\ |\
            S; S\
            \label{eq:symbex}
\end{equation}
where $f: (\Sigma^*)^\arity \to \Sigma^*$ is  a partial string function and $g\subseteq (\Sigma^*)^\arity$ is a string relation.
The following is a simple example of a symbolic execution $S$
which uses string variables ($x$, $y$, and $z$'s) and string constants
(letters \texttt{a} and \texttt{b}), and the concatenation operator ($\concat$):
\begin{equation}
        z_1 := x \concat \text{\texttt{ba}} \concat y;\quad
        z_2 := y \concat \text{\texttt{ab}} \concat x;\quad
        \ASSERT{z_1 == z_2}
        \label{eq:intro_ex}
\end{equation}
The problem of \defn{path feasibility/satisfiability}\footnote{
    It is equivalent to
satisfiability of string constraints in the SMT framework
\cite{SMT-CACM,SMT-chapter,KS08}.
Simply convert a symbolic execution $S$
into a \emph{Static Single Assignment} (SSA) form (i.e. use a new
variable
on l.h.s. of each assignment) and treat assignments as equality,
e.g., formula for the above example is
        $z_1 = x + \text{\texttt{ba}} + y \wedge
        z_2 = y + \text{\texttt{ab}} + x\ \wedge
        z_1 = z_2$, where $+$ denotes the string concatenation operation.
}
asks whether, for a given program $S$, there exist \emph{input} strings (e.g.
$x$ and $y$ in (\ref{eq:intro_ex}))
that can successfully take
$S$ to the end of the program while satisfying all the assertions.
This path can be satisfied by
assigning $y$ (resp.~$x$) to $\texttt{b}$ (resp.~the empty string).
In this paper, we will also allow nondeterministic
functions $f: (\Sigma^*)^\arity \to 2^{\Sigma^*}$ since nondeterminism can be a
useful modelling construct. For example, consider the code in
Figure \ref{fig:pair}. It ensures that each element in \verb+s1+ (construed as a list delimited by
\verb+-+) is longer than each element in \verb+s2+.  If $f: \Sigma^* \to
2^{\Sigma^*}$ is a function that nondeterministically outputs a substring
delimited by \verb+-+, our symbolic execution analysis can be reduced to
feasibility of the path:
\[
x := f(s_1);\quad
y := f(s_2);\quad
\ASSERT{\texttt{len}(x) \leq \texttt{len}(y)}
\]

In the last few decades much research on the satisfiability problem of string constraints
suggests that it takes very little for a string constraint language to
become undecidable.
For example,
although the existential theory of
concatenation and regular constraints (i.e. an atomic expression is either
$E = E'$, where $E$ and $E'$ are concatenations of string constants and variables, or $x \in L$, where $L$ is 
\begin{wrapfigure}{r}{0.5\textwidth}
\begin{flushleft}
\begin{minipage}{0.4\textwidth}
\begin{minted}{python}
# s1, s2: strings with delimiter '-'
for x in s1.split('-')
  for y in s2.split('-')
    assert(len(x) > len(y))
\end{minted}
\end{minipage}
\end{flushleft}
    \caption{A Python code snippet \label{fig:pair}}
\end{wrapfigure}
%
a regular language) is decidable and
in fact \pspace-complete \cite{Plandowski,diekert,J16}, the theory becomes 
undecidable when enriched with letter-counting
\cite{buchi}, i.e., expressions of the form $|x|_a = |y|_b$, where
$|\cdot|_a$ is a function mapping a word to the number of occurrences of the
the letter $a$ in the word. Similarly, although
finite-state transductions \cite{LB16,DV13,BEK} are crucial for expressing many
functions used in string-manipulating programs --- including autoescaping
mechanisms (e.g. backslash escape, and HTML escape in JavaScript), and the
$\replaceall$ function with a constant replacement pattern ---
checking a simple formula of the form $\exists x R(x,x)$, for a given rational
transduction\footnote{A rational transduction is a transduction defined by a rational transducer, namely, a finite automaton over the alphabet $(\Sigma \cup \{\varepsilon\})^2$, where $\varepsilon$ denotes the empty string.} $R$, can easily encode the Post Correspondence Problem
\cite{Morvan00}, and therefore is undecidable.

Despite the undecidability of allowing various operations in string constraints,
in practice it is common for a string-manipulating program to contain multiple
operations (e.g. concatenation and finite-state transductions), and so a path
feasibility solver nonetheless needs to be able to handle them.
This is one reason why some string solving practitioners
opted to support more string operations and settle with incomplete solvers
(e.g. with no guarantee of
termination) that could still solve some constraints that arise in practice,
e.g., see
\cite{S3,TCJ16,Z3-str,Z3-str2,Z3-str3,Berkeley-JavaScript,YABI14,Abdulla17,Stranger,HAMPI,cvc4,trau18}.
For example, the tool S3 \cite{S3,TCJ16} supports general recursively-defined
predicates and uses a number of incomplete heuristics to detect unsatisfiable
constraints. As another example, the tool Stranger \cite{Stranger,YABI14}
supports concatenation, $\replaceall$ (but with both pattern and replacement
strings being constants), and regular constraints, and performs widening (i.e.
an overapproximation) when a concatenation operator is seen in the analysis.
Despite the excellent performance of some of these solvers on several existing
benchmarks,
there are good reasons for designing decision procedures with stronger
theoretical guarantees, e.g., in the form of decidability (perhaps accompanied
by a complexity analysis). One such reason is that string constraint solving is
a research area in its infancy with an insufficient range of benchmarking
examples to convince us that if a string solver works well on existing
benchmarks, it will also work well on future benchmarks.
\OMIT{
To exacerbate the situation, since solvers are typically used only as part of a
program analysis engine (which already deals with an undecidable problem) and
have to be queried \emph{multiple} times in an analysis of a single program,
the lack of a
}
A theoretical
result provides a kind of robustness guarantee upon which a practical
solver could further improve and optimise.
\OMIT{
\anthony{I'm not sure about the following ...}
Robustness is an important property for a solver
Another reason is that
solvers are typically used only as part of a program
analysis engine (which already deals with an undecidable problem) and have to
be queried \emph{multiple} times in an analysis of a single program.
The presence of a stronger theoretical guarantee could help minimise
unpredictability
}

\OMIT{
Having a rich
constraint language accompanied with a complete solver

,
the lack of a stronger theoretical guarantee --- for example, in the form of
decidability --- is not an entirely satisfactory state of affairs.
}

\OMIT{
This is not an entirely satisfactory state of affairs
Since solvers are often queried multiple times in a symbolic execution analysis,
we believe
}

Fortunately, recent years have seen the possibility of recovering some
decidability of string constraint languages with multiple string operations,
while retaining applicability for constraints that arise in practical
symbolic execution applications. This is done by imposing syntactic restrictions
including acyclicity \cite{BFL13,Abdulla14}, solved form \cite{Vijay-length},
and straight-line \cite{LB16,HJLRV18,CCHLW18}. These restrictions are known
to be satisfied by many existing string constraint benchmarks, e.g., Kaluza
\cite{Berkeley-JavaScript}, Stranger \cite{Stranger}, SLOG
\cite{fang-yu-circuits,HJLRV18}, and mutation XSS benchmarks of \cite{LB16}.
However, these results are
unfortunately rather fragmented, and it is difficult to extend the
comparatively limited number of supported string operations. In the following,
we will elaborate this point more precisely.
The acyclic logic of \cite{BFL13} permits only rational
transductions, in which the $\replaceall$ function with constant
pattern/replacement strings and regular constraints (but not concatenation)
can be expressed. On the other hand, the acyclic logic of
\cite{Abdulla14} permits concatenation, regular constraints, and the length
function, but neither the  $\replaceall$ function nor transductions. This
logic is in fact quite related to the solved-form logic proposed earlier by
\cite{Vijay-length}. The straight-line logic of
\cite{LB16} unified the earlier logics by allowing concatenation,
regular constraints, rational transductions, and length and letter-counting
functions. It was pointed out by  \cite{CCHLW18}
that this logic cannot express the $\replaceall$ function with the replacement string
provided as a variable, which was never studied in the context of
verification and program analysis. Chen \emph{et al.} proceeded by showing that
a new straight-line logic with the more general $\replaceall$ function and
concatenation is decidable, but becomes undecidable when the length function is
permitted.

Although the aforementioned results have been rather successful in capturing
many
string constraints that arise in applications (e.g. see the benchmarking
results of \cite{Vijay-length} and \cite{LB16,HJLRV18}), many natural problems
remain unaddressed. \emph{To what extent can one combine these operations without
sacrificing decidability?} For example, can a useful decidable logic
permit
the more general  $\replaceall$, rational transductions, and concatenation at the
same time?
\emph{To what extent can one introduce new string operations without sacrificing
decidability?} For example, can we allow the string-reverse function (a
standard library function, e.g., in Python), or more generally functions given
by two-way transducers (i.e. the input head can also move to the left)? Last
but not least, since there are a plethora of
complex string operations, it is impossible for a solver designer to incorporate
all the string operations that will be useful in all application domains. Thus,
\emph{can (and, if so, how do) we design an effective string solver that
can easily be extended with user-defined string functions, while providing a
strong completeness/termination guarantee?}
Our goal is to provide theoretically-sound and practically implementable
solutions to these problems.

\smallskip
\noindent
\textbf{Contributions.}
        We provide two general semantic conditions (see Section
        \ref{sec-dec}) which together ensure
decidability of path feasibility for string-manipulating programs:
        \begin{enumerate}
            \item[(1)] the conditional $R \subseteq (\ialphabet^*)^k$ in each
                assertion admits a regular monadic decomposition, and
            \item[(2)] each assignment uses a function $f: (\ialphabet^*)^k \to
                2^{\ialphabet^*}$ whose inverse relation preserves
                ``regularity''.
        \end{enumerate}
        Before describing these conditions in more detail, we comment on the
        four main features (4Es) of our decidability result:
        (a) \emph{Expressive}: the two conditions are satisfied by most string
        constraint benchmarks (existing and new ones including those of
        \cite{Berkeley-JavaScript,Stranger,fang-yu-circuits,HJLRV18,LB16}) and
        strictly generalise
        several expressive and decidable constraint languages (e.g. those of
        \cite{LB16,CCHLW18}),
        (b) \emph{Easy}: it leads to a decision procedure that is conceptually
        simple (in particular, substantially simpler than many existing ones),
        (c) \emph{Extensible}: it provides an extensible architecture of
        a string solver that allows users to easily incorporate their own
        user-defined functions to the solver, and (d) \emph{Efficient}:
        it provides a sound basis of our new fast string solver OSTRICH that is
        highly competitive on string constraint benchmarks.
        We elaborate the details of the two aforementioned semantic
        conditions, and our contributions below.

        The first semantic condition simply means that $R$ can be effectively
        transformed into
a finite union $\bigcup_{i=1}^n (L_i^{(1)} \times \cdots \times L_i^{(k)})$ of Cartesian 
products of regular languages. (Note that this is \emph{not} the intersection/product of regular languages.)
A relation that is definable in this way is
often called a \defn{recognisable relation} \cite{CCG06}, which
        is one standard extension of the notion of regular languages (i.e. unary
        relations) to general $k$-ary relations. The
framework of recognisable relations can express interesting conditions that
might at a first glance seem beyond ``regularity'', e.g., $|x_1| + |x_2| \geq 3$
        as can be seen below in
        Example \ref{ex:length}. Furthermore, there are algorithms
        (i.e. called \emph{monadic decompositions} in \cite{monadic-decomposition})
        for deciding whether a given relation represented in highly expressive
        symbolic representations
        (e.g. a \defn{synchronised rational relation} or a
        \defn{deterministic rational relation}) is recognisable and, if so,
        output a symbolic representation of the recognisable relation
        \cite{CCG06}. On the other hand, the second condition means
        that the pre-image $f^{-1}(L)$ of a regular language $L$ under the
        function $f$ is a $k$-ary recognisable relation. This is an expressive
        condition (see Section \ref{sec-core}) satisfied by many string
        functions including concatenation,
        the string reverse function, one-way and two-way finite-state
        transducers, and the $\replaceall$ function where the replacement string
        can contain variables. Therefore, we obtain strict generalisations of
        the decidable string constraint languages in \cite{LB16} (concatenation,
        one-way transducers, and regular constraints) and in \cite{CCHLW18}
        (concatenation, the $\replaceall$ function, and regular constraints).
        In addition, many string solving benchmarks (both existing and
        new ones) derived from practical applications satisfy our two semantics
        conditions including the benchmarks of SLOG \cite{fang-yu-circuits}
        with replace and $\replaceall$, the benchmarks of Stranger
        \cite{Stranger}, ${\sim}80\%$ of Kaluza benchmarks
        \cite{Berkeley-JavaScript}, and the transducer benchmarks of
        \cite{LB16,HJLRV18}.
        \OMIT{
        , and finally new benchmarking examples derived from
        auto-sanitisation web-templating that we introduce in this paper
        }
        We provide a simple and clean decision procedure (see Section
        \ref{sec-dec})
        which propagates the regular language constraints in a \emph{backward}
        manner via the regularity-preserving pre-image computation.
        Our semantic conditions also naturally lead to
        extensible architecture of a
        string solver: a user can easily extend our solver with one's own
        string functions by simply providing one's code for computing
        the pre-image $f^{-1}(L)$ for an input regular language $L$ without
        worrying about other parts of the solver.

        Having talked about the Expressive, Easy, and Extensible features of
        our decidability result (first three of the four Es), our decidability result does not immediately lead to an Efficient decision procedure and a fast string solver. A substantial proportion
        of the remaining
        paper is dedicated to analysing the cause of the problem and proposing
        ways of addressing it which are effective from both theoretical and
        practical standpoints.

        Our hypothesis is that allowing general string relations
        $f: (\ialphabet^*)^k \to 2^{\ialphabet^*}$ (instead of just partial
        functions $f: \ialphabet^* \to \ialphabet^*$), although broadening the
        applicability of the resulting theory (e.g. see Figure \ref{fig:pair}),
        makes the constraint solving problem considerably more difficult.
        One reason is that propagating $n$ regular constraints
        $L_1,\ldots,L_n$ backwards through a string relation $f:
        (\ialphabet^*)^k \to 2^{\ialphabet^*}$ seems to require performing a
        product automata construction for $\bigcap_{i=1}^n L_i$ before
        computing a recognisable relation for $f^{-1}(\bigcap_{i=1}^n L_i)$.
        To make things worse, this product construction has to be done
        \emph{for practically every variable in the constraint}, each of which
        causes an exponential blowup. We illustrate this with a concrete
        example in Example~\ref{ex:nondistr}. We provide a strong piece of theoretical evidence that 
        unfortunately
        this is unavoidable in the worst case. More
        precisely, we show (see Section \ref{sec-core}) that the complexity of
        the path feasibility problem
        with binary relations represented by one-way finite transducers
        (a.k.a. \emph{binary rational relations}) and
        the $\replaceall$ function (allowing a variable in the replacement
        string) has a \nonelementary{} complexity (i.e., time/space complexity
        cannot be bounded by a fixed tower of exponentials) with a single
        level of exponentials
        caused by a product automata construction for each variable in the
        constraint.
        This is especially surprising since allowing
        either binary rational relations or the aforementioned
        $\replaceall$ function results in a constraint language whose complexity
        is at most double exponential time and single exponential space
        (i.e. \expspace{}); see \cite{LB16,CCHLW18}.
        To provide further evidence
        of our hypothesis, we
        accompany this with another lower bound (also see Section
        \ref{sec-core}) that the path feasibility
        problem has a \nonelementary{} complexity for relations that are
        represented by two-way finite transducers (without the $\replaceall$
        function), which are possibly one of the
        most natural and well-studied classes of models of string relations $f:
        \ialphabet^* \to 2^{\ialphabet^*}$ (e.g. see \cite{AD11,EH01,FGRS13}
        for the model).

        We propose two remedies to the problem. The first one is
        to allow only string functions in our constraint language. This allows
        one to avoid the computationally expensive product automata construction
        for each variable in the constraint. In fact, we show (see
        Section \ref{sec:implemented-alg}) that the
        \nonelementary{} complexity for the case of binary rational relations and
        the $\replaceall$ function can be substantially brought down to
        double exponential time and single exponential space (in fact,
        \expspace{}-complete) if the binary rational relations are restricted
        to partial functions. In fact, we prove that this complexity still
        holds if we additionally allow the string-reverse function and the
        concatenation operator. The \expspace{} complexity might still sound prohibitive, but
        the highly competitive performance of our new solver OSTRICH (see below) shows that this is not the case.

        Our second solution (see Section \ref{sec:strlineconact}) is
        to still allow string relations, but find an appropriate syntactic
        fragment of our semantic conditions that yield better computational
        complexity. Our proposal for such a fragment is to \emph{restrict the
        use of $\replaceall$ to constant replacement strings}, but allow
        the string-reverse function and binary rational relations. The
        complexity of this fragment is shown to be \expspace-complete, building
        on the result of \cite{LB16}. There are at least two advantages of
        the second solution. While string relations are supported, our algorithm
        reduces the problem to constraints which can be handled by the
        existing solver SLOTH \cite{HJLRV18} that has a reasonable performance.
        Secondly, the fully-fledged length constraints (e.g. $|x| = |y|$ and
        more generally linear arithmetic expressions on the lengths of string
        variables) can be incorporated into this syntactic fragment without
        sacrificing decidability or increasing the \expspace{} complexity.
        Our experimentation and the comparison of our tool with SLOTH (see below)
        suggest that \emph{our first proposed solution is to be strongly
        preferred when string relations are not used in the constraints}.

        We have implemented our
        first proposed decision procedure
        in a new fast string solver OSTRICH\footnote{As an aside, in contrast
        to an emu, an ostrich is known
        to be able to walk backwards, and hence the name of our solver,
        which propagates regular constraints in a backward direction.}
        (\emph{Optimistic STRIng Constraint
        Handler}). 
        Our solver provides built-in support for concatenation, reverse, functional transducers (\FunFT{}), and $\replaceall$.
        Moreover, it is designed to be extensible and adding support for new string functions is a straight-forward task.
        We compare OSTRICH with several state-of-the-art string solving tools 
        --- including SLOTH \cite{HJLRV18}, CVC4 \cite{cvc4}, and Z3
        \cite{Z3-str3} --- on a wide range of challenging benchmarks
        --- including SLOG's replace/replaceall 
        \cite{fang-yu-circuits}, Stranger's \cite{Stranger}, mutation XSS
        \cite{LB16,HJLRV18}, and the benchmarks of Kaluza that satisfy our semantic conditions (i.e.\ ${\sim}80\%$ of them) \cite{Berkeley-JavaScript}.
        It is the only tool that was able to return an answer on all of the benchmarks we used.
        Moreover, it significantly outperforms SLOTH, the only tool comparable 
        with OSTRICH in terms of theoretical guarantees and closest in 
        terms of expressibility.
        It also competes well with CVC4 --- a fast, but incomplete solver --- 
        on the benchmarks for which CVC4 was able to return a conclusive 
        response. 
        We report details of OSTRICH and empirical results in 
        Section \ref{sec:impl}.

\section{Preliminaries}\label{sec-prel}

\noindent
\textbf{General Notation.} 
Let $\mathbb{Z}$ and $\Nat$ denote the set of integers and natural numbers respectively. For $k \in \Nat$, let $[k] = \{1,\ldots, k\}$. For a vector $\vec{x}=(x_1,\ldots, x_n)$, let $|\vec{x}|$ denote the length of $\vec{x}$ (i.e., $n$) and  $\vec{x}[i]$ denote $x_i$ for each $i \in [n]$.  Given a function $f: A \to B$ and $X \subseteq B$, we use $f^{-1}(X)$ to define the pre-image of $X$ under $f$, i.e., $\{ a \in A: f(a) \in X \}$.

\smallskip
\noindent
\textbf{Regular Languages.}
Fix a finite \emph{alphabet} $\Sigma$. Elements in $\Sigma^*$ are called \emph{strings}. Let $\varepsilon$ denote the empty string and  $\Sigma^+ = \Sigma^* \setminus \{\varepsilon\}$. We will use $a,b,\ldots$ to denote letters from $\Sigma$ and $u, v, w, \ldots$ to denote strings from $\Sigma^*$. For a string $u \in \Sigma^*$, let $|u|$ denote the \emph{length} of $u$ (in particular, $|\varepsilon|=0$), moreover, for $a \in \Sigma$, let $|u|_a$ denote the number of occurrences of $a$ in $u$. A \emph{position} of a nonempty string $u$ of length $n$ is a number $i \in [n]$ (Note that the first position is $1$, instead of  0). In addition, for $i \in [|u|]$, let $u[i]$ denote the $i$-th letter of $u$. For a string $u \in \Sigma^*$, we use $u^R$ to denote the reverse of $u$, that is, if $u = a_1 \cdots a_n$, then $u^R= a_n \cdots a_1$.
For two strings $u_1, u_2$, we use $u_1 \cdot u_2$ to denote the \emph{concatenation} of $u_1$ and $u_2$, that is, the string $v$ such that $|v|= |u_1| + |u_2|$ and for each $i \in [|u_1|]$, $v[i]= u_1[i]$, and for each $i \in |u_2|$, $v[|u_1|+i]=u_2[i]$. Let $u, v$ be two strings. If $v = u \cdot v'$ for some string $v'$, then $u$ is said to be a \emph{prefix} of $v$. In addition, if $u \neq v$, then $u$ is said to be a \emph{strict} prefix of $v$. If $u$ is a prefix of $v$, that is, $v = u \cdot v'$ for some string $v'$, then 
we use $u^{-1} v$ to denote $v'$. In particular, $\varepsilon^{-1} v = v$.

A \emph{language} over $\Sigma$ is a subset of $\Sigma^*$. We will use $L_1, L_2, \dots$ to denote languages. For two languages $L_1, L_2$, we use $L_1 \cup L_2$ to denote the union of $L_1$ and $L_2$, and $L_1 \cdot L_2$ to denote the concatenation of $L_1$ and $L_2$, that is, the language $\{u_1 \cdot u_2 \mid u_1 \in L_1, u_2 \in L_2\}$. For a language $L$ and $n \in \Nat$, we define $L^n$, the \emph{iteration} of $L$ for $n$ times, inductively as follows: $L^0=\{\varepsilon\}$ and $L^{n} =L \cdot L^{n-1}$ for $n > 0$. We also use $L^*$ to denote an arbitrary number of iterations of $L$, that is, $L^* = \bigcup \limits_{n \in \Nat} L^n$. Moreover, let $L^+ = \bigcup \limits_{n \in \Nat \setminus \{0\}} L^n$.

\begin{definition}[Regular expressions $\regexp$]
	\[e \eqdef \emptyset \mid \varepsilon \mid a \mid e + e \mid e \concat e \mid e^*, \mbox{ where } a \in \Sigma. \]
	Since $+$ is associative and commutative, we also write $(e_1 + e_2) + e_3$ as $e_1 + e_2 + e_3$ for brevity. We use the abbreviation $e^+ \equiv e \concat e^*$. Moreover, for $\Gamma = \{a_1, \ldots, a_n\}\subseteq \Sigma$, we use the abbreviations $\Gamma \equiv a_1 + \cdots + a_n$ and $\Gamma^\ast \equiv (a_1 + \cdots + a_n)^\ast$. 
\end{definition}
We define $\Ll(e)$ to be the language defined by $e$, that is, the set of strings that match $e$, inductively as follows: $\Ll(\emptyset) =\emptyset$,
$\Ll(\varepsilon) =\{\varepsilon\}$,
%
$\Ll(a)= \{a\}$,
%
$\Ll(e_1 + e_2) = \Ll(e_1) \cup \Ll(e_2)$,
%
$\Ll(e_1 \concat e_2) = \Ll(e_1) \cdot \Ll(e_2)$,
%
$\Ll(e_1^*)=(\Ll(e_1))^*$.
In addition, we use $|e|$ to denote the number of symbols occurring in $e$.


\smallskip
\noindent
\textbf{Automata models.} We review some background from automata theory;
for more, see \cite{Kozen-automata,HU79}. Let $\ialphabet$ be a finite set (called
\defn{alphabet}).

\begin{definition}[Finite-state automata] \label{def:nfa}
	A \emph{(nondeterministic) finite-state automaton}
	(\FA{}) over a finite alphabet $\ialphabet$ is a tuple $\Aut =
	(\ialphabet, \controls, q_0, \finals, \transrel)$ where 
	$\controls$ is a finite set of 
	states, $q_0\in \controls$ is
	the initial state, $\finals\subseteq \controls$ is a set of final states, and 
	$\transrel\subseteq \controls \times 
	\ialphabet \times  \controls$ is the
	transition relation. 
\end{definition}

For an input string $w = a_1 \dots a_n$, a \emph{run} of $\Aut$ on $w$
is a sequence of states $q_0,\ldots, q_n$ such that $(q_{j-1}, a_{j}, q_{j}) \in
\transrel$  for every $j \in [n]$.
The run is said to be \defn{accepting} if $q_n \in \finals$.
A string $w$ is \defn{accepted} by $\Aut$ if there is an accepting run of
$\Aut$ on $w$. In particular, the empty string $\varepsilon$ is accepted by $\Aut$ iff $q_0 \in F$. The set of strings accepted by $\Aut$ is denoted by $\Lang(\Aut)$,
a.k.a., the language \defn{recognised} by $\Aut$.
The \defn{size} $|\Aut|$ of $\Aut$ is defined to be $|\controls|$; we will
use this when we discuss computational complexity.

For convenience, we will also refer to an \FA{} without initial and final states, that is, a pair $(Q, \delta)$, as a \emph{transition graph}.

\smallskip
\noindent\textbf{Operations of \FA{}s.} For an FA $\Aut=(Q, q_0, F, \delta)$, $q \in Q$ and $P \subseteq Q$, we use $\Aut(q, P)$ to denote the FA $(Q, q, P, \delta)$, that is, the FA obtained from $\Aut$ by changing the initial state and the set of final states to $q$ and $P$ respectively. We use $q \xrightarrow[\Aut]{w} q'$ to denote that a string $w$ is accepted by $\Aut(q, \{q'\})$.

Given two \FA{}s $\Aut_1 = (Q_1, q_{0,1}, F_{1}, \delta_1)$ and $\Aut_2 = (Q_2, q_{0,2}, F_2, \delta_2)$, the \emph{product} of $\Aut_1$ and $\Aut_2$, denoted by $\Aut_1 \times \Aut_2$, is defined as $(Q_1 \times Q_2, (q_{0,1}, q_{0,2}), F_1 \times F_2, \delta_1 \times \delta_2)$, where $\delta_1 \times \delta_2$ is the set of tuples $((q_1,q_2), a, (q'_1, q'_2))$ such that $(q_1, a, q'_1) \in \delta_1$ and $(q_2, a, q'_2) \in \delta_2$. Evidently, we have $\Lang(\Aut_1 \times \Aut_2) = \Lang(\Aut_1) \cap \Lang(\Aut_2)$.

Moreover, let $\Aut=(Q, q_0, F, \delta)$, we define $\Aut^\revsym$ as $(Q, q_f, \{q_0\}, \delta')$, where $q_f$ is a newly introduced state not in $Q$ and $\delta'$ comprises the transitions $(q', a, q)$ such that $(q, a, q') \in \delta$ as well as the transitions $(q_f, a, q)$ such that $(q, a, q') \in \delta$ for some $q' \in F$. 
Intuitively, $\Aut^\revsym$ is obtained from $\Aut=(Q, q_0, F, \delta)$ by reversing
the direction of each transition of $\Aut$ and swapping initial and final states. The new state $q_f$ in $\Aut^\revsym$ is introduced to meet the unique initial state requirement in the definition of FA. Evidently, $\Aut^\revsym$  recognises the reverse language of $\Lang(\Aut)$, namely, the language $\{u^R \mid u \in \Lang(\Aut)\}$.

It is well-known (e.g. see \cite{HU79}) that regular expressions and \FA{}s are 
expressively equivalent, and generate precisely all \emph{regular languages}.
In particular, from a regular expression, an equivalent \FA{} can be constructed 
in linear time. Moreover, regular languages are closed under Boolean
operations, i.e., union, intersection, and complementation.
 
\begin{definition}[Finite-state transducers]
	Let $\ialphabet$ be an alphabet. A \emph{(nondeterministic) finite  transducer} (\FT{}) $\Transducer$  over $\ialphabet$ is a tuple $(\ialphabet,  \controls, q_0, \finals, \transrel)$, where $\transrel$ is  a finite subset of $\controls \times  \ialphabet \times 
	\controls \times \ialphabet^*$. 
\end{definition}
The notion of runs of \FT{}s on an input string can be seen as a generalisation 
of \FA{}s by adding outputs. More precisely, given a string $w = a_1 \dots a_n$, a \emph{run} of $\Transducer$ on $w$
is a sequence of pairs $(q_1, w'_1), \ldots, (q_n, w'_n) \in \controls \times \Sigma^*$ 
such that for every $j \in [n]$, $(q_{j-1}, a_j, q_{j}, w'_j) \in
\transrel$.
The run is said to be \defn{accepting} if  $q_n \in \finals$.
When a run is accepting, $w'_1 \ldots w'_n$ is said to be the \emph{output} of the
run. Note that some of these $w'_i$s could be empty strings.
A word $w'$ is said to be an output of $\Transducer$ on $w$ if there is an accepting run of
$\Transducer$ on $w$ with output $w'$. We use $\Tran(\Transducer)$ to denote the
\emph{transduction} defined by $\Transducer$, that is, the relation comprising
the pairs $(w, w')$ such that $w'$ is an output of $\Transducer$ on $w$. 

We remark that an \FT{} usually defines a \emph{relation}. We shall speak of \emph{functional
	transducers}, i.e., transducers that define functions instead of relations. (For instance, deterministic transducers are always functional.) We will use \FunFT{} to denote the class of functional transducers. 

To take into consideration the outputs of transitions, we define the \emph{size} $|\Transducer|$ of $\Transducer$ as the sum of the sizes of transitions in $\Transducer$, where the size of a transition $(q, a, q', w')$ is defined as $|w'|+1$. 

\begin{example}\label{exmp-ft}
	We give an example \FT{}    
for the function \textbf{escapeString}, which backslash-escapes every occurrence of \texttt{'} and 
	\texttt{"}. The \FT{} has a single state, i.e., $Q=\{q_0\}$ and 
	the transition relation $\delta$ comprises
	$(q_0, \ell,  q_0, \ell)$ for each $\ell \neq  \texttt{'}$  or $\texttt{"} $, 
	$(q_0,\texttt{'}, q_0, \textbackslash\texttt{'})$, 
	$(q_0,\texttt{"}, q_0, \textbackslash\texttt{"})$,  
	and the final state $F=\{q_0\}$. We remark that this \FT{} is functional. \qed
\end{example}

\smallskip
\noindent
\textbf{Computational Complexity.}
In this paper, we will use computational complexity theory to provide
evidence that certain (automata) operations in our generic decision procedure
are unavoidable.
In particular, we shall deal with the following computational complexity
classes (see \cite{HU79} for more details): \pspace{} (problems solvable in polynomial
space and thus in exponential time),  \expspace{} (problems solvable
in exponential space and thus in double exponential time), and \nonelementary{} (problems not a member of the class \elementary{}, where \elementary{} comprises elementary recursive functions, which is the union of the complexity classes \exptime{}, 2-\exptime, 3-\exptime, $\ldots$, or alternatively, the union of the complexity classes \expspace, 2-\expspace, 3-\expspace, $\ldots$). Verification
problems that have complexity \pspace{} or beyond (see \cite{BK08}
for a few examples) have substantially benefited from techniques
such as symbolic model checking \cite{McMillan}. 


\section{Semantic conditions and A generic decision procedure}\label{sec-dec}

Recall that we consider symbolic executions of string-manipulating programs defined by the rules
\begin{equation} \label{eq:mainsymexe}
    S ::= \qquad y := f(x_1,\ldots,x_\arity) \ |\
    \text{\ASSERT{$g(x_1,\ldots,x_\arity)$}}\ |\ 
            S; S\ 
\end{equation}
where $f: (\Sigma^*)^\arity \to 2^{\Sigma^*}$ is a \emph{nondeterministic} partial string function and $g \subseteq (\Sigma^*)^\arity$
is a string relation. 
%
Without loss of generality, we assume that symbolic executions are in Static Single Assignment (SSA) form.\footnote{Each symbolic execution can be turned into the SSA form by using a new variable on the left-hand-side of each assignment.}

In this section, we shall provide two general semantic conditions for symbolic executions. The main result is that, whenever the symbolic execution generated by \eqref{eq:mainsymexe} satisfies these two conditions, the path feasibility problem is decidable. 
We first define the concept of recognisable relations which, intuitively, are simply a finite union of 
Cartesian products of regular languages. 
\begin{definition}[Recognisable relations] \label{def:recrel}
	An $\arity$-ary relation $R\subseteq \Sigma^*\times \cdots\times \Sigma^*$ is \emph{recognisable}  if $R=\bigcup_{i=1}^n L^{(i)}_1\times \cdots\times L^{(i)}_\arity$ where $L^{(i)}_j$ is regular for each $j\in [\arity]$. A \emph{representation} of a recognisable relation $R=\bigcup_{i=1}^n L^{(i)}_1\times \cdots\times L^{(i)}_\arity$ is $(\Aut^{(i)}_1, \ldots, \Aut^{(i)}_\arity)_{1 \le i \le n}$ such that each $\Aut^{(i)}_j$ is an \FA{} with $\Lang(\Aut^{(i)}_j)=L^{(i)}_j$. The tuples $(\Aut^{(i)}_1, \ldots, \Aut^{(i)}_\arity)$ are called the \emph{disjuncts} of the representation and the \FA{}s $\Aut^{(i)}_j$ are called the \emph{atoms} of the representation.
	%
\end{definition}

We remark that the recognisable relation is more expressive than it appears to be. For instance, it can be used to encode some special length constraints, as demonstrated in Example~\ref{ex:length}. 

\begin{example} \label{ex:length}
Let us consider the relation $|x_1|+|x_2| \ge 3$ where $x_1$ and $x_2$ are strings  over the alphabet $\ialphabet$. Although syntactically $|x_1|+|x_2| \ge 3$ is a length constraint, it indeed defines a recognisable relation. To see this, $|x_1|+|x_2| \ge 3$ is equivalent to 
the disjunction of $|x_1| \ge 3$, $|x_1| \ge 2 \wedge |x_2| \ge 1$, $|x_1| \ge 1 \wedge |x_2| \ge 2$, and $|x_2| \ge 3$, 
where each disjunct describes a cartesian product of regular languages. For instance, in $|x_1| \ge 2 \wedge |x_2| \ge 1$, $|x_1| \ge 2$ requires that $x_1$ belongs to the regular language $\Sigma \cdot \Sigma^+$, while $|x_2| \ge 1$ requires that $x_2$ belongs to the regular language $\Sigma^+$.  \qed
\end{example}
The equality binary predicate $x_1 = x_2$ is a standard non-example of 
recognisable
relations; in fact, expressing $x_1 = x_2$ as a union $\bigcup_{i \in I} L_i 
\times H_i$ of products requires us to have $|L_i| = |H_i| = 1$, which in 
turn forces us to have an infinite index set $I$.

\OMIT{
On the other hand, length constraints in general go beyond recognisable relations, as shown in Example~\ref{ex:length-eq}.
\begin{example}\label{ex:length-eq}
Let us consider the relation $|x_1| = |x_2|$.  This relation is not a recognisable relation, which can be proved by contraction. To the contrary, suppose that $|x_1| = |x_2|$ is a recognisable relation, say, $\bigcup \limits_{1 \le i \le n} L^{(i)}_{1} \times L^{(i)}_{2}$, where $L^{(i)}_{1}, L^{(i)}_{2}$ are regular languages. Let $a \in \ialphabet$ and consider the string pairs $(a^j, a^j)$ for $j \in \Nat$. Since $n$ is finite, by the pigeon-hole principle, there are $i: 1 \le i \le n$  and $j_1, j_2 \in \Nat$ such that $j_1 \neq j_2$, $(a^{j_1}, a^{j_1}) \in L^{(i)}_{1} \times L^{(i)}_{2}$, and $(a^{j_2}, a^{j_2}) \in L^{(i)}_{1} \times L^{(i)}_{2}$. From this, we deduce that $(a^{j_1}, a^{j_2}) \in L^{(i)}_1 \times L^{(i)}_2$, thus $|a^{j_1}| = |a^{j_2}|$, a contradiction.
\end{example}
}


The first semantic condition, \emph{Regular Monadic Decomposition} is stated as follows.
\begin{center}
\framebox{\parbox{0.9\textwidth}{\regmondec: For each assertion $\ASSERT{g(x_1,\ldots, x_\arity)}$ in $S$,  $g$ is a recognisable relation, a representation of which, in terms of Definition~\ref{def:recrel}, is effectively computable. 
}}
\end{center}

When $\arity=1$, the \regmondec{} condition 
requires that $g(x_1)$ is regular and may be given by an \FA{} $\Aut$, in which case $x_1\in\Lang(\Aut)$.

\smallskip

The second semantic condition concerns the pre-images of string operations.
%
A string operation $f(x_1, \ldots, x_\arity)$ with $\arity$ parameters ($\arity \geq 1$) gives rise to a relation $R_f\subseteq (\Sigma^*)^\arity \times \Sigma^*$. Let $L \subseteq \ialphabet^*$. The \emph{pre-image} of $L$ under $f$, denoted by $\Pre_{R_f}(L)$, is 
\[\left\{(w_1,\ldots, w_\arity) \in (\Sigma^*)^\arity \mid \exists w.\ w\in f(w_1, \ldots, w_\arity)\text{ and } w\in L  \right\}.\]
For brevity, we use $\Pre_{R_f}(\Aut)$ to denote $\Pre_{R_f}(\Lang(\Aut))$ for an \FA{} $\Aut$. 
The second semantic condition, i.e. the inverse relation of $f$ preserves regularity, is formally stated as follows.
\begin{center}
\framebox{\parbox{0.9\textwidth}{\prerec{}: For each operation $f$ in $S$ and each \FA{} $\Aut$,  $\Pre_{R_f}(\Aut)$ is a recognisable relation, a representation of which (Definition~\ref{def:recrel}), can be effectively computed from $\Aut$ and $f$.
}}
\end{center}
When $\arity=1$, this \prerec{} condition would state that the pre-image
of a regular language under the operation $f$ is \emph{effectively regular}, i.e. an \FA{} can be computed to represent the pre-image of the regular language under $f$. 

\begin{example}
Let $\Sigma = \{a, b\}$. Consider the string function $f(x_1, x_2) = a^{|x_1|_a + |x_2|_a} b^{|x_1|_b + |x_2|_b}$. (Recall that $|x|_a$ denotes the number of occurrences of $a$ in $x$.) We can show that for each \FA{} $\Aut$, $\Pre_{R_f}(\Aut)$ is a recognisable relation. Let $\Aut$ be an \FA{}. W.l.o.g. we assume that $\Lang(\Aut) \subseteq a^* b^*$. It is easy to observe that $\Lang(\Aut)$ is a finite union of the languages  $\{a^{c_1 p + c_2} b^{c'_1 p'+c'_2} \mid p \in \Nat, p' \in \Nat\}$, where $c_1,c_2,c'_1, c'_2$ are natural number constants. Therefore, to show that $\Pre_{R_f}(\Aut)$ is a recognisable relation, it is sufficient to show that $\Pre_{R_f}(\{a^{c_1 p + c_2} b^{c'_1 p'+c'_2} \mid p \in \Nat, p' \in \Nat\})$ is a recognisable relation.
 

Let us consider the typical situation that $c_1 \neq 0$ and $c'_1 \neq 0$. Then $\Pre_{R_f}(\{a^{c_1 p + c_2} b^{c'_1 p'+c'_2} \mid p \in \Nat, p' \in \Nat\})$ is the disjunction of $L^{(i, i')}_{1} \times L^{(j, j')}_{2}$ for $i, j, i', j' \in \Nat$ with $i + j = c_2$, and $i' + j' = c'_2$, where $L^{(i,i')}_1 = \{u \in \ialphabet^* \mid |u|_a \ge i, |u|_a \equiv i \bmod c_1, |u|_b \ge i',  |u|_b \equiv i' \bmod c'_1\}$, $L^{(j, j')}_2 = \{v \in \ialphabet^* \mid |v|_a \ge j, |v|_a \equiv j \bmod c_1, |v|_b \ge j',  |v|_b \equiv j' \bmod c'_1\}$. Evidently, $L^{(i, i')}_{1}$ and  $L^{(j, j')}_{2}$ are regular languages. Therefore, $\Pre_{R_f}(\{a^{c_1 p + c_2} b^{c'_1 p'+c'_2} \mid p \in \Nat, p' \in \Nat\})$ is a finite union of cartesian products of regular languages, and thus a recognisable relation. \qed
\end{example}

Not every string operation satisfies the \prerec{} condition, as demonstrated by Example~\ref{exmp-not-prerec}.
\begin{example}\label{exmp-not-prerec}
Let us consider the string function $f$ on the alphabet $\{0,1\}$ that transforms the unary representations of natural numbers into their binary representations, namely, $f(1^n) = b_0 b_1 \ldots b_m$ such that $n = 2^m b_0  + \cdots + 2 b_{m-1} + b_m$ and $b_0 = 1$. For instance, $f(1^4) = 100$. We claim that $f$ does not satisfy the \prerec{} condition. To see this, consider the regular language $L=\{10^i \mid i \in \Nat\}$. Then $\Pre_{R_f}(L)$ comprises the strings $1^{2^j}$ with $j \in \Nat$, which is evidently non-regular.
    Incidentally, this is an instance of the well-known Cobham's theorem
    (cf.~\cite{pippenger-book}) that the sets of numbers definable by finite
    automata in unary are strictly subsumed by the sets of numbers definable 
    by finite automata in binary. \qed
\end{example}

\medskip
We are ready to state the main result of this section.
\begin{theorem}
	The path feasibility problem is decidable for symbolic executions satisfying the \regmondec{} and \prerec{} conditions.
	\label{th:gen}
\end{theorem}

\begin{proof}[Proof of Theorem~\ref{th:gen}]
We present a nondeterministic decision procedure from which the theorem follows.

	Let $S$  be a symbolic execution, $y := f(\vec{x})$ (where
	$\vec{x} = x_1,\ldots,x_\arity$) be the last
	assignment in $S$, and $\rho := \{ g_1(\vec{z}_1),\ldots, g_s(\vec{z}_s)\}$ be the set of all constraints in assertions of $S$ that involve $y$ (i.e. $y$ occurs in $\vec{z}_i$ for all $i \in [s]$). For each $i \in [s]$, let $\vec{z}_i = (z_{i,1}, \ldots, z_{i, \ell_i})$. Then by the \regmondec{} assumption, $g_i$ is a recognisable relation and a representation of it, say $\left(\Aut^{(j)}_{i, 1}, \ldots, \Aut^{(j)}_{i, \ell_i} \right)_{1 \le j \le n_i}$ with $n_i \ge 1$, can be effectively computed.

For each $i \in [s]$, we nondeterministically choose one tuple $(\Aut^{(j_i)}_{i, 1}, \ldots, \Aut^{(j_i)}_{i, \ell_i})$ (where $1 \le j_i \le n_i$), and for all $i \in [s]$, replace $\ASSERT{g_i(\vec{z}_i)}$ in $S$ with $\ASSERT{z_{i,1} \in \Aut^{(j_i)}_{i, 1}}; \ldots; \ASSERT{z_{i,\ell_i} \in \Aut^{(j_i)}_{i, \ell_i}}$. Let $S'$ denote the resulting program.

We use $\sigma$ to denote the set of all the \FA{}s $\Aut^{(j_i)}_{i, i'}$ such that $1 \le i \le s$, $1 \le i' \le \ell_i$, and $\ASSERT{y \in \Aut^{(j_i)}_{i, i'}}$ occurs in $S'$. We then compute the product \FA{} $\Aut$ from \FA{}s $\Aut^{(j_i)}_{i, i'} \in \sigma$ such that  $\Lang(\Aut)$ is the intersection of  the languages defined by \FA{}s in $\sigma$.
By the 
	\reginvrel{} assumption, $g' = \Pre_{R_f}(\Aut)$ is a recognisable relation and a representation of it can be effectively computed.
%

Let $S''$ be the symbolic execution obtained from $S'$ by  (1) removing $y := f(\vec{x})$ along with all assertions involving $y$ (i.e. the assertions $\ASSERT{y \in \Aut^{(j_i)}_{i, i'}}$ for $\Aut^{(j_i)}_{i, i'} \in \sigma$), (2)  and adding the assertion $\ASSERT{g'(x_1,\ldots, x_r)}$.

\hide{
Let $S'$ be the symbolic execution 
 obtained by (1) removing $y := f(\vec{x})$ along with all assertions with conditionals in $\rho$, and (2) adding all the assertions $\ASSERT{z_{i,i'} \in \Aut^{(j_i)}_{i, i'}}$ such that $i \in [s]$, $i' \in [\ell_i]$ and $z_{i,i'} \neq y$. We use $\sigma$ to denote the set of all the \FA{}s $\Aut^{(j_i)}_{i, i'}$ such that $1 \le i \le s$, $1 \le i' \le \ell_i$, and $z_{i,i'} = y$. We then compute the product \FA{} $\Aut$ from \FA{}s $\Aut^{(j_i)}_{i, i'} \in \sigma$ such that  $\Lang(\Aut)$ is the intersection of  the languages defined by \FA{}s in $\sigma$.
By the 
	\reginvrel{} assumption, $g' = \Pre_{R_f}(\Aut)$ is a recognisable relation and a representation of it can be effectively computed.
%
Let $S''$ be the symbolic execution 
	\[
	S'; \quad \ASSERT{g'(x_1,\ldots, x_r)}\enspace.
	\]
%
}

It is straightforward to verify that $S$ is path-feasible iff there is a nondeterministic choice resulting in $S'$ that is path-feasible, moreover, $S'$ is path feasible iff $S''$ is path-feasible. Evidently, $S''$ has one less assignment than $S$. 
Repeating these steps, the procedure will terminate when $S$ becomes a conjunction of
	assertions on input variables, the feasibility of which can be checked via language nonemptiness checking of \FA{}s. To sum up, the correctness of the (nondeterministic) procedure follows since the path-feasibility is preserved for each step, and the termination is guaranteed by the finite number of assignments. 
\end{proof}

Let us use the following example to illustrate the generic decision procedure. 
\begin{example}
Consider the symbolic execution 
	\[
	\ASSERT{x \in \Aut_0};\  y_1 := f(x);\  z := y_1 \concat y_2;\ 
	 \ASSERT{y_1 \in \Aut_1};\ \ASSERT{y_2 \in \Aut_2};\ \ASSERT{z \in \Aut_3}
	\]
	where $\Aut_0,\Aut_1, \Aut_2, \Aut_3$ are \FA{}s illustrated in Figure~\ref{fig-gen-dec}, and $f: \ialphabet^* \to 2^{\ialphabet^*}$ is the function mentioned in Section~\ref{sec:intro} that nondeterministically outputs a substring delimited by -. At first, we remove the assignment $z = y_1 \concat y_2$ as well as the assertion $\ASSERT{z \in \Aut_3}$. 
Moreover, since the pre-image of $\concat$ under $\Aut_3$, denoted by $g$, is a recognisable relation represented by $(\Aut_3(q_0, \{q_i\}), \Aut_3(q_i, \{q_0\}))_{0 \le i \le 2}$, we add the assertion $\ASSERT{g(y_1,y_2)}$, and get following program
\[
	\begin{array}{l}
	\ASSERT{x \in \Aut_0};\  y_1 := f(x);\  \ASSERT{y_1 \in \Aut_1};\ \ASSERT{y_2 \in \Aut_2};\ \ASSERT{g(y_1,y_2)}. 	
	\end{array}
\]
To continue, we nondeterministically choose one tuple, say $(\Aut_3(q_0, \{q_1\}), \Aut_3(q_1, \{q_0\}))$,  from the representation of $g$, and replace $\ASSERT{g(y_1,y_2)}$ with $\ASSERT{y_1 \in \Aut_3(q_0, \{q_1\})}; \ASSERT{y_2 \in  \Aut_3(q_1, \{q_0\})}$, and get the program
\[
	\begin{array}{l}
	\ASSERT{x \in \Aut_0};\  y_1 := f(x);\  \ASSERT{y_1 \in \Aut_1};\ \ASSERT{y_2 \in \Aut_2}; \\
	\hspace{4mm} \ASSERT{y_1 \in \Aut_3(q_0, \{q_1\})};\ \ASSERT{y_2 \in  \Aut_3(q_1, \{q_0\})}. 	
	\end{array}
\]
Let $\sigma$ be $\{\Aut_1,  \Aut_3(q_0, \{q_1\})\}$, the set of \FA{}s occurring in the assertions for $y_1$ in the above program. Compute the product $\Aut'=\Aut_1 \times \Aut_3(q_0, \{q_1\})$ and $\Aut'' = \Pre_{R_f}(\Aut')$ (see Figure~\ref{fig-gen-dec}).

Then we remove $y_1:=f(x)$, as well as the  assertions that involve $y_1$,  namely, $\ASSERT{y_1 \in \Aut_1}$ and $\ASSERT{y_1 \in \Aut_3(q_0, \{q_1\})}$, and add the assertion $\ASSERT{x \in \Aut''}$, resulting in the program 
\[
	\begin{array}{l}
	\ASSERT{x \in \Aut_0};\   \ASSERT{y_2 \in \Aut_2}; \ \ASSERT{y_2 \in \Aut_3(q_1, \{q_0\})};\ \ASSERT{x \in \Aut''}. 	
	\end{array}
\]
It is not hard to see that $\mbox{-}\ a\ \mbox{-} \in \Lang(\Aut_0) \cap \Lang(\Aut'') $ and $abb \in \Lang(\Aut_2) \cap \Lang(\Aut_3(q_1, \{q_0\}))$. Then the assignment $x = \mbox{-}\ a\ \mbox{-}$, $y_1 = a$, $y_2 = abb$, and $z =  aa bb$ witnesses the path feasibility of the original symbolic execution.\qed
\end{example}

\begin{figure}[htbp]
\includegraphics[scale=0.6]{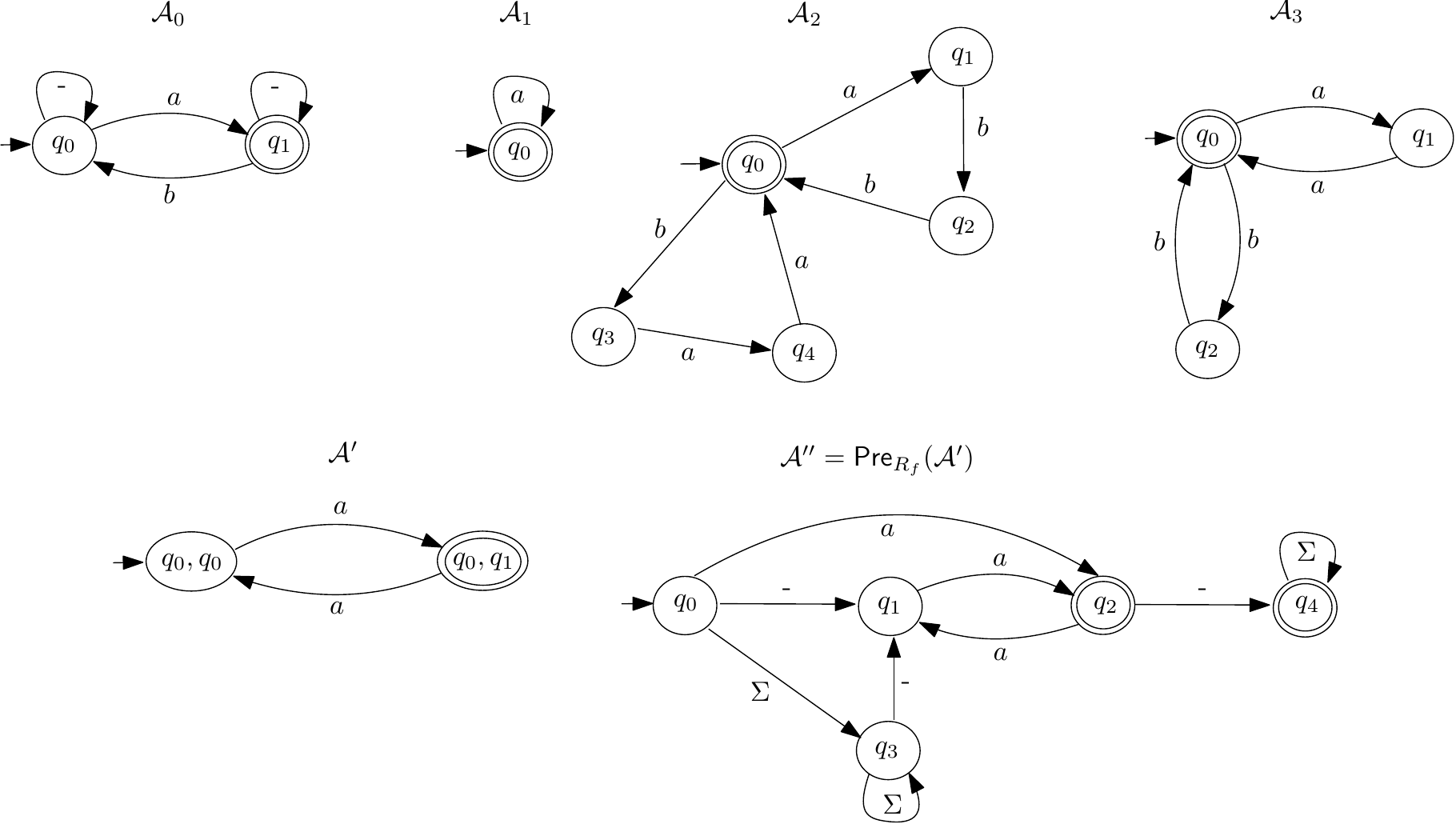}
\caption{$\Aut_0, \Aut_1, \Aut_2, \Aut_3, \Aut', \Pre_{R_f}(\Aut')$, where $\Sigma = \{a, b, \mbox{-}\}$}
\label{fig-gen-dec}
\end{figure}

\begin{remark}
Theorem~\ref{th:gen} 
gives two semantic conditions which are sufficient to render the path feasibility problem decidable. A natural question, however, is how 
to check whether a given symbolic execution satisfies the two semantic conditions. The answer to this meta-question highly depends on the classes of string operations and relations under consideration. Various classes of relations which admit finite representations have been studied in the literature. They include, in an ascending order of expressiveness,  recognisable relations, synchronous relations, deterministic rational relations, and rational relations, giving rise to a strict hierarchy. (We note that slightly different terminologies tend to be used in the literature, for instance, synchronous relations in \cite{CCG06} are called regular relations in \cite{BFL13} which are also known as automatic relations, synchronised rational relations, etc. One may consult the survey \cite{choffrut-survey} and \cite{CCG06}.)
It is known \cite{CCG06} that determining whether a given deterministic rational relation is recognisable is decidable (for binary relations, this can be done in doubly exponential time), and deciding whether a synchronous relation is recognisable can be done in exponential time \cite{CCG06}.  
Similar results are also mentioned in \cite{BLSS03,Libkin03}.


 
By these results, one can check, for a given symbolic execution where the string relations in the assertion  and the relations induced by the string operation are all \emph{deterministic rational} relations, whether it satisfies the two semantic conditions.   Hence, one can check algorithmically whether Theorem~\ref{th:gen} is applicable. 
\qed
\end{remark}

\OMIT{ 
\subsection{Complexity}

In this section, we provide a refined version of Theorem \ref{th:gen} that accounts for
complexity. 
To this end, 
we consider a \emph{representation} of 
the recognisable relation $R$, 
%
which is a collection of 
tuples $(\Aut^{(i)}_1, \ldots, \Aut^{(i)}_\arity)_{i\in [n]}$  such that 
$L^{(i)}_j = \Lang(\Aut^{(i)})_j$ for each $j$.
Each tuple $(\Aut^{(i)}_1, \ldots, \Aut^{(i)}_\arity)$ is called a 
\emph{disjunct}, and  each \FA{} $\Aut^{(i)}_j$ is called an \emph{atom}.

We propose
to use (a variant of) conjunctions of regular constraints called
\emph{conjunctive \FA{}}. 
More precisely, each conjunctive \FA{} is a tuple $\Aut = ((\controls, \transrel), \conacc)$
such that $(\controls, \transrel)$ is a transition graph of an \FA{} and 
$\conacc \subseteq \controls \times \controls$ is a \emph{conjunctive acceptance 
	condition}, i.e., a string $w$ is accepted by $\Aut$ if 
for \emph{every} $(q, q') \in \conacc$, there is a run of $(\controls, \transrel)$ on
$w$ from $q$ to $q'$. 
%
Evidently, a  conjunctive \FA{} $\Aut=((\controls, \transrel), \conacc)$ is a \emph{succinct} representation of the product of \FA{}s $\cB_{q,q'}$, which has a size exponentially large than that of $\Aut$ in the worst case.
%
%
Accordingly, a representation of $R$ is called a \emph{conjunctive
	representation} if every atom in the representation is a conjunctive \FA{}.
Since each \FA{} $\Aut = 
(\ialphabet, \controls, q_0, \finals, \transrel)$
can be assumed to have only one final state (i.e. $\finals = \{q_F\}$), each 
\FA{} can be identified with the conjunctive \FA{}
$((\controls,\transrel),\{(q_0,q_F)\})$.
The \emph{size} of a conjunctive representation $((\controls, \transrel), \conacc)$ is defined as $|\controls|$, and  the  \emph{atom size} of a conjunctive representation of $R$ is defined to be the maximum size of the atoms therein.

The regularity condition \prerec{} now incorporates conjunctive
\FA{}s as representations of atoms: 
\emph{for each function $f$ in $S$ and each conjunctive \FA{} $\Aut$, 
	there is an algorithm that runs in $\ell(|f|,|\Aut|)^{c_0}$ space which enumerates 
	each disjunct of a conjunctive representation of $f^{-1}(\Aut)$, whose atom size is bounded by $\ell(|f|,|\Aut|)$, where $\ell(\cdot, \cdot)$ is a monotone function and $c_0 > 0$ is a constant}.

Here $|f|$ is the size of a representation of $f$ to be made concrete when it is concerned. 
%

We also define several size parameters of $S$ as follows:
\begin{itemize}
\item $\rcdep(S)$: the number of assignments in $S$, 
\item $\rcdim(S)$: the maximum arity of string functions in the assignments of $S$, 
\item $\rcphi(S)$: the maximum size of the representations of these string functions,  
\item $\rcasrt(S)$: the number of assertions in $S$, and 
\item $\rcpsi(S)$:the maximum size of the \FA{}s appearing in the assertions of 
$S$.
\end{itemize}
Moreover,
for $\ell:\Nat^2\rightarrow \Nat$, we define $\ell^{\langle n \rangle}(j, k)$ ($n\geq 1$) as $\ell^{\langle 1 \rangle}(j,k)= \ell(j, k)$ and $\ell^{\langle n+1 \rangle }(j, k) = \ell(j, \ell^{\langle n \rangle}(j,k))$.

\begin{theorem}\label{thm-generic-dec}
	Given a symbolic execution $S$ satisfying \prerec{}, the 
	path feasibility problem of $S$ can be decided \emph{nondeterministically} 
	with space
	$$(\rcdim(S)+1)^{\rcdep(S)}  \rcasrt(S) \left(\ell^{\langle \rcdep(S)
		\rangle}(\rcphi(S), \rcpsi(S)) \right)^{O(1)}$$
	%
	%
\end{theorem}


\subsection*{Proof of Theorem \ref{thm-generic-dec}}\label{app:algo}

In this proof, we will use the following proposition.
\begin{proposition}\label{prop-conj-fa-prod}
For a pair of conjunctive \FA{}s $\Aut_1 = ((\controls_1, \transrel_1), S_1)$
    and $\Aut_2 = ((\controls_2, \transrel_2), S_2)$, the conjunctive \FA{} 
    $\Aut = (\controls_1 \cup \controls_2, \transrel_1 \cup \transrel_2, S_1
    \cup S_2)$ recognises $\Lang(\Aut_1) \cap \Lang(\Aut_2)$.
\end{proposition}
This is immediate from the definition of acceptance for \FA{}.
\OMIT{
\begin{proof}
Let $\Aut_1 = ((\controls_1, \transrel_1), S_1)$ and $\Aut_2 = ((\controls_2, \transrel_2), S_2)$ be two conjunctive \FA{}s such that $S_1 = \{(p_1, p'_1),\ldots, (p_m, p'_m)\}$ and $S_2 = \{(q_1, q'_1), \ldots, (q_n, q'_n)\}$ with $m \le n$. Then we construct a conjunctive \FA{} $\Aut = (\controls_1 \times \controls_2, \transrel_1 \times \transrel_2, S)$, where 
$$S= \{((p_i, q_i), (p'_i, q'_i)) \mid i \in [m]\} \cup \{((p_m, q_{m+i}), (p'_m, q_{m+i})) \mid i \in [n-m]\}.$$ 
Note $|S|  = \max(|S_1|, |S_2|)$. Evidently $\Lang(\Aut) = \Lang(\Aut_1) \cap \Lang(\Aut_2)$. 
\end{proof}
}

\paragraph{Spelling out the algorithm in full.}
To facilitate the complexity analysis, we will spell out the modified algorithm 
in detail. Let $S$ be a symbolic execution satisfying \prerec{} . 
We give a nondeterministic algorithm to decide the path feasibility of
$S$. 

By definition, $S$ has $\rcdep(S)$ assignments.
First, for each variable $x \in \vars(S)$,
define $\cE_{\rcdep(S)}(x)$  as the set of conjunctive \FA{}s $((\controls, \transrel), \{(q_0, q_F)\})$ such that $\ASSERT{x \in \Aut}$ occurs in $S$, where $\Aut= (\controls, q_0, \{q_F\}, \transrel)$. 
Note that at this stage each conjunctive \FA{} in $\cE_{\rcdep(S)}(x)$ is 
actually just a normal \FA{}. Let $i:=\rcdep(S)$ and iterate the following procedure until $i=0$.


\OMIT{
We start from the last assignment of $S$,  set $i:=\rcdep(S)$, and construct $\cE_{\rcdep(S)}(x)$ as follows: For each 
$x \in \Aut$ appearing in an assertion of $S$, where $\Aut$ is an FA, 
we \emph{nondeterministically} select one state $q \in \finals$ and include $((\controls, \transrel), \{(q_0, q)\})$ into $\cE_{\rcdep(S)}(x)$. Then we iterate the following procedure to compute $\cE_{i}$ until $i=0$.  
}
%

Suppose the $i$-th assignment of $S$ is $y:= f(x_1,
\cdots, x_\arity)$ with $\arity\geq 1$ (note here possibly $x_{j_1} = x_{j_2}$ for some $j_1 < j_2$), and 
$\cE_i(y)=\{\Aut_1, \ldots, \Aut_n\}$, 
where $\Aut_j$ is a conjunctive FA for each $j \in [n]$.
By \prerec{}, we can enumerate each disjunct $(\Aut^{(\iota)}_{j, 1}, \ldots, \Aut^{(\iota)}_{j, \arity})$ of a conjunctive representation of 
$f^{-1}(\Aut_j)$ in $(\ell(|f|,|\Aut_j|))^{c_0}$ space, where each atom is of size at most $\ell(|f|,|\Aut_j|)$.
Then $\cE_{i-1}(x)$ for $x \in  \{x_1,\cdots, x_\arity\}$ is constructed as follows: 
\begin{enumerate}
\item For each $j \in [n]$, nondeterministically select a disjunct $\left(\Aut_{j, 1}, \ldots, \Aut_{j, \arity}\right)$ of the conjunctive representation of $f^{-1}(\Aut_j)$.
\item 
Let
\[
    \cE_{i-1}(x):= \cE_{i}(x) \cup \left\{\Aut_{j, j'} \mid  j \in [n], j' \in [\arity], x_{j'} = x
        \right\}.
\]
[For each $x{\notin} \{x_1,\ldots, x_\arity\}$, let $\cE_{i-1}(x) := 
        \cE_i(x)$.]
%
\item Let $i: = i-1$.
\end{enumerate}
%
Let $\cE(x):=\cE_{0}(x)$ for each variable $x$.
To decide the path feasibility of $S$, we simply show that each collection
$\cE(x)$ (for each $x \in \vars(S)$) of conjunctive \FA{}s has a nonempty 
language intersection.

\hide{
\tl{tbh, I think the algo is clear enough and the example does not add too much.}\zhilin{maybe, but the example is at least more concrete. Moreover, the description of the algorithm is short and dry, maybe just use the example to let the reader spend a bit more time on the algorithm and digest, although with some redundancy. We may remove it if we run out of space. May decide this later.}
Let us use an example to help the reader understand the computation of $\cE_{i+1}$ from $\cE_i$.  Suppose that in $\cG_i$, to compute $\cE_{i+1}$, we select a variable $x$, which has no predecessors and three outgoing edges, say $(x, (\Transducer, 0), y)$, $(x, (\Transducer, 1), z)$, and $(x, (\Transducer, 2), y)$, where $y$ and $z$ are two distinct variables, moreover, $\cE_i(x) = \{\Aut_1, \Aut_2\}$. Let us also assume that  $\Pre_{\Transducer}(\Aut_1)$ (resp. $\Pre_{\Transducer}(\Aut_2)$) is represented by $(\Aut^{(j')}_{1, 0}, \Aut^{(j')}_{1, 1}, \Aut^{(j')}_{1, 2})_{ j'  \in [2]}$ (resp. $(\Aut^{(j')}_{2, 0}, \Aut^{(j')}_{2, 1}, \Aut^{(j')}_{2, 2})_{ j'  \in [3]}$). If for $j = 1$ (resp. $j=2$), $(\Aut^{(1)}_{1, 0}, \Aut^{(1)}_{1, 1}, \Aut^{(1)}_{1, 2})$  (resp. $\left(\Aut^{(3)}_{2, 0}, \Aut^{(3)}_{2, 1}, \Aut^{(3)}_{2, 2}\right)$)  is selected, then $\cE_{i+1}(y) = \cE_i(y) \cup \{\Aut^{(1)}_{1, 0}, \Aut^{(1)}_{1, 2}, \Aut^{(3)}_{2, 0} \cup \Aut^{(3)}_{2, 2}\}$ and $\cE_{i+1}(z) = \cE_{i}(z) \cup \{\Aut^{(1)}_{1, 1}, \Aut^{(3)}_{2, 1}\}$. 
\tl{I have not updated this part as likely the proof will go to appendix}
}

    \OMIT{
To decide the path feasibility of $S$, we have the following nondeterministic algorithm: first (nondeterministically) construct the sets $\cE(x)$ for $x \in \vars(S)$ as detailed above, then 
checking the nonemptiness of the product of conjunctive \FA{}s in $\cE(x)$ for each $x \in \vars(S)$.
}

\paragraph{Detailed complexity analysis.}
For each $i$, 
let $N_i$ be the maximum size of the conjunctive \FA{}s in $\bigcup \limits_{x \in \vars(S)} \cE_i(x)$.
Since each string function $f$ satisfies $|f| \le \rcphi(S)$, we have that $N_{i-1} \le \ell(|f|, N_i) \le \ell(\rcphi(S), N_i)$ (note that we have assumed that $\ell$ is monotone). Furthermore, because each conjunctive \FA{} in $\cE_{\rcdep(S)}(x)$ for each $x \in \vars(S)$ is of size bounded by $\rcpsi(S)$, we have 
    that for each $x \in \vars(S)$, each conjunctive \FA{} in $\cE(x)$ is of
    size bounded by \\ $\ell^{\langle \rcdep(S) \rangle}(\rcphi(S), \rcpsi(S))$. 
We emphasise that, according to the \prerec{} assumption, the construction of these conjunctive \FA{}s can be done in nondeterministic $(\ell^{\langle \rcdep(S) \rangle}(\rcphi(S), \rcpsi(S)))^{c_0}$ space. 

According to Proposition~\ref{prop-conj-fa-prod}, for each $x \in \vars(S)$, the conjunctive product \FA{} $\Aut_x=((\controls_x, \transrel_x), S_x)$ of these conjunctive \FA{}s  in $\cE(x)$ is of size 
    $$|\cE(x)|(\ell^{\langle \rcdep(S) \rangle}(\rcphi(S), \rcpsi(S))),$$
and $|S_x| \le (\ell^{\langle \rcdep(S) \rangle}(\rcphi(S), \rcpsi(S)))^2$. Therefore, the size of the (standard) \FA{} corresponding to $\Aut_x$ is 
    $$(|\cE(x)| \ell^{\langle \rcdep(S) \rangle}(\rcphi(S), \rcpsi(S)))^{(\ell^{\langle \rcdep(S) \rangle}(\rcphi(S), \rcpsi(S)))^2}.$$

Since the nonemptiness of an \FA{} can be solved in nondeterministic logarithmic space, we conclude that the nonemptiness of $\Aut_x$ can be solved in nondeterministic
{\small
    $$(\ell^{\langle \rcdep(S) \rangle}(\rcphi(S), \rcpsi(S)))^2 
    \log(|\cE(x)|\ell^{\langle \rcdep(S) \rangle}(\rcphi(S), \rcpsi(S)))$$
}
space.

We now analyse $|\cE(x)|$, the size of $\cE(x)$.
For each $i$, let $M_i$ be the maximum number of elements in $\cE_i(x)$ for $x  \in \vars(S)$.
Then we have $M_{i-1} \le (r+1)M_i \le (\rcdim(S)+1)M_i $. Since $\cE_{\rcdep(S)}(x)$ contains at most $\rcasrt(S)$ elements, we deduce that $\cE(x)$ contains at most $(\rcdim(S)+1)^{\rcdep(S)}\rcasrt(S)$ elements. 

Therefore, for each $x \in \vars(S)$, 
the nonemptiness of $\Aut_x$ can be solved in nondeterministic 
{\small
$$
    (\ell^{\langle \rcdep(S) \rangle}(\rcphi(S), \rcpsi(S)))^2 
    \log[(\rcdim(S)+1)^{\rcdep(S)}\rcasrt(S)
    \ell^{\langle \rcdep(S) \rangle}(\rcphi(S), \rcpsi(S))]$$
}
space.

Moreover, from the aforementioned analysis, we know that each \FA{} $\Aut$ occurring in $S$ gives rise to at most $(\rcdim(S)+1)^{\rcdep(S)}$ conjunctive \FA{}s in $\bigcup \limits_{x \in \vars(S)} \cE(x)$. Therefore, 
$$\sum \limits_{x \in \vars(S)} |\cE(x)| \le (\rcdim(S)+1)^{\rcdep(S)} \rcasrt(S).$$
We then deduce that the sum of the sizes of all the conjunctive \FA{}s in
    $\cE(x)$ for all $x \in \vars(S)$ is bounded by 
$$
\begin{array}{l l}
& \sum \limits_{x \in \vars(S)} (|\cE(x)| \ell^{\langle \rcdep(S) \rangle}(\rcphi(S), \rcpsi(S)))  \\
\le & (\rcdim(S)+1)^{\rcdep(S)}\rcasrt(S) \ell^{\langle \rcdep(S) \rangle}(\rcphi(S), \rcpsi(S)) .
\end{array}
$$

In summary, the aforementioned nondeterministic algorithm takes 
$$(\rcdim(S)+1)^{\rcdep(S)}  \rcasrt(S) \left(\ell^{\langle \rcdep(S) \rangle}(\rcphi(S), \rcpsi(S)) \right)^c$$
 space for some constant $c > 0$.



\tl{Here to instantiate with FTs, replaceall, reverse}
\paragraph{Instantiation with \FT{}s, $\replaceall$, and $\reverse$.}
To obtain a complexity upper-bound for the core constraint language, we verify the regularity condition \prerec{}. 
%
%
%
We start with the $\replaceall$ function with a slightly different (but equivalent by Currying the second argument) form 
$\replaceall_{p}(sub, rep)$ where $p$ is a fixed regular expression. 


%
\begin{proposition}\label{prop-replace-pt}
	For the function $\replaceall_p$ with the pattern $p$ being a regular expression, we can compute a 
	\PT{} $\Transducer_{\sf replaceall}$ of size $2^{\bigO(|p|^c)}$ for some $c > 0$ such that $\Tran(\Transducer_{\sf replaceall})$ expresses $\replaceall_p$.
	\label{prop:replaceAll}
\end{proposition}

\begin{lemma}\label{lem-1pt}
	The regularity condition \prerec{} holds for 
	\begin{itemize}
		\item the \FT{} $T$ with $\ell_T(|\Transducer|, |\Aut|) = |\Transducer||\Aut|$.
		\item the $\reverse$ function with $\ell_\reverse(|\Aut|) = |\Aut|$.
		\item the $\replaceall$ function $\replaceall_p$ with $\ell(|\Aut|)=2^{\bigO(|p|^c)}\cdot |\Aut|$ 
	\end{itemize}
\end{lemma}
}


\section{An Expressive Language Satisfying The Semantic Conditions} \label{sec-core}

Section~\ref{sec-dec} has identified general \emph{semantic} conditions under which the \emph{decidability} of the path feasibility problem can be attained.
Two questions naturally arise:
\begin{enumerate}
	\item How general are these semantic conditions? In particular, do string functions commonly used in practice satisfy these semantic conditions?
	\item What is the computational complexity of checking path feasibility?
\end{enumerate}
The next two sections will be devoted to answering these questions.

For the first question, we shall introduce a \emph{syntactically} defined string constraint language $\strline$, which includes general string operations such as the $\replaceall$ function and those definable by two-way transducers, as well as recognisable relations.
[Here, $\strline$ stands for ``straight-line'' because our work generalises
the straight-line logics of \cite{LB16,CCHLW18}.]
We 
first recap the $\replaceall$ function that allows a general (i.e. variable)
replacement string \cite{CCHLW18}. Then we give the definition of two-way
transducers whose special case (i.e. one-way transducers) has been given in Section~\ref{sec-prel}.

\subsection{The $\replaceall$ function and two-way transducers}

The $\replaceall$ function has three parameters: the first parameter is the \emph{subject} string, the second parameter is a \emph{pattern} that is a regular expression, and the third parameter is the \emph{replacement} string. For the semantics of $\replaceall$ function, in particular when the pattern is a regular expression, we adopt the \emph{leftmost and longest} matching. For instance, $\replaceall(aababaab, (ab)^+, c) =ac\cdot \replaceall(aab, (ab)^+, c)= acac$, since the leftmost and longest matching of $(ab)^+$ in $aababaab$ is $abab$. Here we require that the language defined by the pattern parameter does \emph{not} contain the empty string, in order to avoid the troublesome definition of the semantics of the matching of the empty string. We refer the reader to \cite{CCHLW18} for the formal semantics of the $\replaceall$ function. To be consistent with the notation in this paper, for each regular expression $e$, we define
 the string function $\replaceall_e: \ialphabet^* \times \ialphabet^* \rightarrow \ialphabet^*$ such that for $u, v \in \ialphabet^*$, $\replaceall_e(u, v) = \replaceall(u, e, v)$, and we write $\replaceall(x, e, y)$ as $\replaceall_e(x,y)$.

As in the one-way case, we start with a definition of two-way finite-state automata.
\begin{definition}[Two-way finite-state automata] \label{def:2nfa}
	A \emph{(nondeterministic) two-way finite-state automaton}
	(\FFA{}) over a finite alphabet $\ialphabet$ is a tuple $\Aut =
	(\ialphabet, \EndLeft, \EndRight, \controls, q_0, \finals, \transrel)$ where
	$\controls, q_0, \finals$ are as in \FA{}s, $\EndLeft$ (resp.~$\EndRight$) is a left (resp.~right) input tape end
	marker,
	and the
	transition relation  $\transrel\subseteq \controls \times
	\overline{\ialphabet}\times \{\Left, \Right\}\times \controls$, where $\overline{\ialphabet} = \ialphabet \cup \{\EndLeft, \EndRight\}$.
	Here, we assume 
	that
	there are no transitions that take the head of the tape past the left/right
	end marker (i.e.~$(p,\EndLeft,\Left,q), (p,\EndRight,\Right,q) \notin
	\transrel$ for every $p, q \in \controls$).

	Whenever they can be easily understood, we will not mention $\ialphabet$,
	$\EndLeft$, and $\EndRight$ in $\Aut$.
\end{definition}

The notion of runs of \FFA{} on an input string is exactly the same as that of
Turing machines on a read-only input tape. More precisely, for a string
$w = a_1 \dots a_n$, a \emph{run} of $\Aut$ on $w$
is a sequence of pairs $(q_0,i_0),\ldots, (q_m,i_m) \in \controls \times [0, n+1]$
defined as follows. Let $a_0 = \EndLeft$ and $a_{n+1} = \EndRight$. The
following conditions then have to be satisfied: $i_0 = 0$, and for every $j \in [0, m-1]$, we have $(q_j,a_{i_j}, dir, q_{j+1}) \in
\transrel$ and $i_{j+1} = i_j + dir$ for some $dir \in  \{\Left, \Right\}$.

The run is said to be \defn{accepting} if $i_m = n+1$ and $q_m \in \finals$.
A string $w$ is \defn{accepted} by $\Aut$ if there is an accepting run of
$\Aut$ on $w$. The set of strings accepted by $\Aut$ is denoted by $\Lang(\Aut)$,
a.k.a., the language \defn{recognised} by $\Aut$.
The \defn{size} $|\Aut|$ of $\Aut$ is defined to be $|\controls|$; this will
be needed when we talk about computational complexity.

Note that	an \FA{} can be seen as a \FFA{} such that $\transrel \subseteq \controls \times \overline{\ialphabet} \times
	\{\Right\} \times \controls $, with the two end markers $\EndLeft, \EndRight$ omitted.
\FFA{} and \FA{} recognise precisely the same class of languages, i.e.,
\emph{regular languages}. The following result is standard and can be found in textbooks on automata theory
(e.g. \cite{HU79}).

\begin{proposition}\label{prop-2nfa-nfa}
	Every \FFA{} $\Aut$ can be transformed in exponential time into an equivalent \FA{} of size $2^{\bigO(|\Aut| \log |\Aut|)}$.
\end{proposition}


\begin{definition}[Two-way finite-state transducers]
	Let $\ialphabet$ be an alphabet. A \emph{nondeterministic two-way finite  transducer} (\FFT{}) $\Transducer$  over $\ialphabet$ is a tuple $(\ialphabet, \EndLeft, \EndRight, \controls, q_0, \finals, \transrel)$, where $\ialphabet, \controls, q_0, \finals$ are as in \FT{}s, and $\transrel \subseteq \controls \times \overline{\ialphabet}\times \{\Left, \Right\}\times
	\controls \times \ialphabet^*$, satisfying the syntactical constraints of \FFA{}s, and the additional constraint that the output must be $\epsilon$ when reading $\EndLeft$ or $\EndRight$. Formally, for each transition $(q, \EndLeft, dir, q', w)$ or $(q, \EndRight, dir, q', w)$ in $\delta$, we have $w=\epsilon$.
\end{definition}
The notion of runs of \FFT{}s on an input string can be seen as a generalisation
of \FFA{}s by adding outputs. More precisely, given a string $w = a_1 \dots a_n$, a \emph{run} of $\Transducer$ on $w$
is a sequence of tuples $(q_0, i_0, w'_0),\ldots, (q_m, i_m, w'_m) \in \controls \times
[0, n+1] \times \Sigma^*$
such that, if $a_0 =\ \EndLeft$ and $a_{n+1} =\ \EndRight$,
we have $i_0 = 0$, and  for every $j \in [0, m-1]$, $(q_j, a_{i_j}, dir, q_{j+1}, w'_j) \in
\transrel$, $i_{j+1} = i_j + dir$ for some $dir \in \{\Left, \Right\}$, and $w'_0 = w'_m = \varepsilon$.
The run is said to be \defn{accepting} if $i_m = n+1$ and $q_m \in \finals$.
When a run is accepting, $w'_0 \ldots w'_m$ is said to be the \emph{output} of the
run. Note that some of these $w'_i$s could be empty strings.
A word $w'$ is said to be an output of $\Transducer$ on $w$ if there is an accepting run of
$\Transducer$ on $w$ with output $w'$. We use $\Tran(\Transducer)$ to denote the
\emph{transduction} defined by $\Transducer$, that is, the relation comprising
the pairs $(w,w')$ such that $w'$ is an output of $\Transducer$ on $w$.


Note that an \FT{} over	$\ialphabet$ is a \FFT{} such that $\transrel \subseteq \controls \times \overline{\ialphabet} \times \{\Right\} \times \controls \times \ialphabet^*$, with the two endmarkers $\EndLeft, \EndRight$ omitted.

\begin{example}\label{exmp-2ft}
	We give an example of \FFT{} for the function $f(w)=ww^R$. The transducer has three states $Q=\{q_0, q_1, q_2\}$, and the transition relation $\delta$ comprises $(q_0, \ell, 1, q_0, \ell)$ for $\ell\in \Sigma$, $(q_0, \EndLeft, 1, q_0, \epsilon)$, $(q_0, \EndRight, -1, q_1, \epsilon)$,  $(q_1, \ell, -1, q_1, \ell)$ for $\ell\in \Sigma$, $(q_1,\EndLeft, 1,q_2,\epsilon)$, $(q_2,\ell, 1,q_2,\epsilon)$ for $\ell\in \Sigma$. The final state $F=\{q_2\}$.	\qed
\end{example}


\subsection{The constraint language $\strline$}

The constraint language $\strline$ is defined by the following rules,
\begin{equation}
S ::=\quad z:= x\ \concat\ y \ |\ z := \replaceall_e(x, y) \ | \  y := \reverse(x) \ |\  y := \Transducer(x)\ |\  \text{\ASSERT{$R(\vec{x})$}}\ |\
S; S\
\label{eq:SL}
\end{equation}
where $\concat$ is the string concatenation operation which concatenates two strings, $e$ is a regular expression, 
$\reverse$ is the string function which reverses a string, $\Transducer$ is a \FFT{},  and $R$ is a recognisable relation represented by a collection of tuples of \FA{}s.

For the convenience  of Section~\ref{sec-frag},
for a class of string operations $\opset$, we will use $\strline[\opset]$ to denote the fragments of $\strline$ that only use the string operations from $\opset$. Moreover, we will use $\sreplaceall$ to denote the special case of the $\replaceall_e$ function where the replacement parameters are restricted to be \emph{string constants}. Note that, according to the result in \cite{CCHLW18}, an instance of the $\sreplaceall$ function $\replaceall_e(x, u)$ with $e$ a regular expression and $u$ a string constant can be captured by \FT{}s. However, such a transformation incurs an exponential blow-up. We also remark that 
we do not present $\strline$ in the most succinct form. For instance, it is known that the concatenation operation can be simulated by the $\replaceall$ function, specifically, $z = x \concat y \equiv z' = \replaceall_a(ab, x) \wedge z = \replaceall_b(z', y)$, where $a, b$ are two fresh letters.
Moreover, it is evident that the $\reverse$ function is subsumed by \FFT{}s.

We remark that $\strline$ is able to encode some string functions with  multiple (greater than two) arguments by transducers and repeated use of $\replaceall$, which is practically convenient particularly for user-defined  functions.

The following theorem answers the two questions raised in the beginning of this section.
\begin{theorem} \label{thm:strlineall}
	The path feasibility problem of  $\strline$ is decidable with a non-elementary lower-bound.
\end{theorem}

To invoke the result of the previous section, we have the following proposition.   

\begin{proposition}\label{prop-sem-cond}
The $\strline$ language satisfies the two semantic conditions $\regmondec$ and $\reginvrel$.
\end{proposition}

\begin{proof}
It is sufficient to show that the $\replaceall_e$ functions and the string operations defined by \FFT{}s satisfy the  \prerec{}  condition.

The fact that $\replaceall_e$ for a given regular expression $e$ satisfies the \prerec{}  condition was shown in \cite{CCHLW18}.

That the pre-image of an \FFT{} $T$ under a regular language defined by an \FA{} $\Aut$  is effectively regular is folklore.
Let $T=(\controls, q_0, \finals, \transrel)$ be a \FFT{} and $\Aut=(\controls', q'_0, \finals', \transrel')$ be an \FA{}. Then  $\Pre_{\Tran(T)}(\Aut)$ is the regular language defined by the \FFA{} $\Aut' = (\controls \times \controls', (q_0, q'_0), \finals \times \finals', \transrel'')$, where $\transrel''$ comprises the tuples $((q_1,q'_1), a, (q_2, q'_2))$ such that there exists $w \in \ialphabet^*$ satisfying that $(q_1, a, q_2, w) \in \delta$ and $q'_1 \xrightarrow[\Aut]{w} q'_2$. From Proposition~\ref{prop-2nfa-nfa}, an equivalent \FA{} can be built from $\Aut'$ in exponential time.
\end{proof}

From Proposition~\ref{prop-sem-cond} and Theorem~\ref{th:gen}, the path feasibility problem of $\strline$ is \emph{decidable}.

To address the complexity (viz.\ the second question raised at the beginning of this section), we show that the path feasibility problem of $\strline$ is \nonelementary.

\begin{proposition} \label{prop:nexpspace-hardness}
	The path feasibility problem of the following two fragments is \nonelementary: $\strline$ with \FFT{}s,  and $\strline$ with \FT{}s+$\replaceall$.
\end{proposition}

For each $\expheight$ we reduce from a tiling problem that is hard for $\expheight$-$\expspace$.
For this we need to use large numbers that act as indices.
Similar encodings of large numbers appear in the study of higher-order programs (e.g.~\cite{J01,CW07}) except quite different machinery is needed to enforce the encoding.
The complete reduction is given in
\shortlong{the full version of this article}
{Appendix~\ref{sec:two-way-lower} in the supplementary material},
with some intuition given here.

\label{sec:tiling-def}
A \emph{tiling problem} consists of a finite set of tiles $\tiles$ as well as horizontal and vertical tiling relations
$\hrel, \vrel \subseteq \tiles \times \tiles$.
Given a tiling \emph{corridor} of a certain width, as well as initial and final tiles
$\inittile, \fintile \in \tiles$
the task is to find a tiling where the first (resp. last) tile of the first (resp. last) row is $\inittile$ (resp. $\fintile$), and horizontally (resp. vertically) adjacent tiles $\tile, \tile'$ have
$\tup{\tile, \tile'} \in \hrel$
(resp. $\vrel$).
Corridor width can be considered equivalent to the space of a Turing machine.
We will consider problems where the corridor width is
$2^{:^{2^m}}$
where the height of the stack of exponentials is $n$.
E.g.\ when $n$ is $0$ the width is $m$, when $n$ is $1$ the width is $2^m$, when $n$ is $2$ the width is $2^{2^m}$ and so on.
Solving tiling problems of width
$2^{:^{2^m}}$
is complete for the same amount of space.

Solving a tiling problem of corridor width $m$ can be reduced to checking whether a \FFT\ of size polynomial in $m$ and the number of tiles can output a specified symbol $\top$.
Equivalently, we could use a \FFA.
A solution is a word
\[
    t_{1,1} t_{1,2} \ldots t_{1,m} \numsep \ldots \numsep t_{h,1} t_{h,2} \ldots t_{h,m}
\]
where $\numsep$ separates rows.
The \FFT\ performs $m+1$ passes.
During the first pass it checks that the tiling begins with $\inittile$, ends with $\fintile$, and
$\tup{t_{i,j}, t_{i,j+1}} \in \hrel$
for all
$1 \leq j < m$.
In $m$ more passes we verify that $\vrel$ is obeyed; the $j$th pass verifies the $j$th column.

Now consider two \FFT{}s and a tiling problem of width $2^m$.
Intuitively, we precede each tile with its column number in $m$ binary bits.
That is
\[
0\ldots00\ \tile_{1,1}\ %
0\ldots01\ \tile_{1,2}\ %
\ldots
1\ldots11\ \tile_{1,2^m}\ %
\numsep
\ldots
\numsep
0\ldots00\ \tile_{m,1}\ %
0\ldots01\ \tile_{m,2}\ %
\ldots
1\ldots11\ \tile_{m,2^m} \ .
\]
The first \FFT{} checks the solution similarly to the width $m$ problem, but needs to handle the large width when checking $\vrel$.
For this it will use a second \FFT.
For each column, the first \FFT{} nondeterministically selects all the tiles in this column
(verifying $\vrel$ on-the-fly).
The addresses of the selected tiles are output to the second \FFT{} which checks that they are equal.
The first \FFT{} goes through a non-deterministic number of such passes and the second \FFT{} enforces that there are $2^m$ of them (in column order).
To do this, the second \FFT{} checks that after the addresses of the $i$-th column are output by the first \FFT{}, then the addresses of the $(i+1)$-th column are output next.
Length $m$ binary numbers are checked similarly to width $m$ tiling problems.

With another \FFT\ we can increase the corridor width by another exponential.
For doubly-exponential numbers, we precede each tile with a binary sequence of exponential length.
For this we precede each bit with another binary sequence, this time of length $m$.
The first \FFT\ outputs queries to the second, which outputs queries to the third \FFT, each time removing one exponential.
With $(n+1)$ \FFT, we can encode tiling problems over an $n$-fold exponentially wide corridor.

The same proof strategy can be used for FTs+$\replaceall$.
The \FFT s used in the proof above proceed by running completely over the word and producing some output, then silently moving back to the beginning of the word.
An arbitrary number of passes are made in this way.
We can simulate this behaviour using \FT s and $\replaceall$.

\label{sec:two-way-simulation}
To simulate
$y := \Transducer(x)$
for a \FFT\ $\Transducer$ making an arbitrary number of passes over the contents of a variable $x$, as above, we use fresh variables $x_1$ and $x_2$, and an automaton $\Aut_a$ recognising $(a \passsep)^\ast$ for some arbitrary character $a$ and delimiter $\passsep$ not used elsewhere.
With these we use the constraint
\[
    \text{\ASSERT{$x_1 \in \Aut_a$}};
    x_2 := \replaceall_a(x_1, x);
    y := \Transducer'(x_2)
\]
where $\Transducer'$ simulates $\Transducer$ in the forwards direction, and simulates (simply by changing state) a silent return to the beginning of the word when reading $\passsep$.
It can be seen that $x_2$ contains an arbitrary number of copies of $x$, separated by $\passsep$, hence $\Transducer'$ simulates $\Transducer$.

It was stated in Section~\ref{sec:intro} that the \nonelementary{} complexity will be caused by repeated product constructions. These product constructions are not obvious here, but are hidden in the treatment of $\replaceall_a$. This treatment is elaborated on in Section~\ref{sec:replaceall-explanation}.
The key point is that to show $\replaceall_a$ satisfies \reginvrel\ one needs to produce a constraint on $x$ that is actually the product of several automata.





\section{More ``Tractable'' Fragments}\label{sec-frag}

In this section, we show that the \nonelementary{}  lower-bound of the preceding section should not be read too pessimistically. As we demonstrate in this section, the complexity of the path feasibility problem can be dramatically reduced (\expspace-complete) for the following two fragments,
\begin{itemize}
\item \strlinefft{}, where \FFT{}s in $\strline$ are restricted to be \emph{one-way and functional},
\item \strlineconcat{}, where the \emph{replacement} parameter of the $\replaceall$ function is restricted to be a string \emph{constant}, and \FFT{}s  are restricted to be \emph{one-way} (but \emph{not} necessarily functional).
\end{itemize}
These two fragments represent the most practical usage of string functions. 
In particular, instead of very general two-way transducers, one-way transducers are commonly used to model, for instance, browser transductions \cite{LB16}.




\subsection{The fragment \strlinefft}
\label{sec:implemented-alg}

The main result of this subsection is stated in the following theorem.

\begin{theorem}\label{thm-fft}
	The path feasibility of \strlinefft{} is $\expspace$-complete.
\end{theorem}

To show the upper bound of Theorem~\ref{thm-fft}, we  will refine the (generic) decision procedure in Section~\ref{sec-dec}
in conjunction with a careful complexity analysis.
%
The crucial idea is to \emph{avoid the product construction} before each pre-image computation
in the algorithm given in the proof of 
Theorem~\ref{th:gen}. This is possible now since all string operations in \strlinefft{} are (partial) functions, and regular constraints are distributive with respect to the pre-image of string (partial) functions.

\begin{fact}\label{fact-dist}
For every string (partial) function $f: (\ialphabet^*)^\arity \to \ialphabet^*$ and regular languages $L_1$ and $L_2$, it holds that $\Pre_{R_f}(L_1 \cap L_2) =
\Pre_{R_f}(L_1) \cap \Pre_{R_f}(L_2)$.
\end{fact} 

To see this, suppose $\vec{u} \in \Pre_{R_f}(L_1) \cap \Pre_{R_f}(L_2)$. Then $\vec{u} \in \Pre_{R_f}(L_1)$ and $\vec{u} \in \Pre_{R_f}(L_2)$. Therefore, there are $v_1\in L_1, v_2\in L_2$ such that $(\vec{u}, v_1) \in R_f$ and $(\vec{u}, v_2) \in R_f$. Since $f$ is a (partial) function, it follows that $v_1 = v_2 \in L_1\cap L_2$, thus $\vec{u} \in \Pre_{R_f}(L_1 \cap L_2)$.  This equality does \emph{not} hold in general if $f$ is \emph{not} functional, as shown by the following example.

\begin{example} \label{ex:nondistr}
Let $\mbox{-} \in \ialphabet$ and $f: \ialphabet^* \rightarrow 2^{\ialphabet^*}$ be the nondeterministic function mentioned in the introduction (see Figure~\ref{fig:pair}) that nondeterministically outputs a substring delimited by \verb+-+. Moreover, let $a, b$ be two distinct letters from $\ialphabet \setminus \{\mbox{-}\}$, $L_1 = a (\ialphabet \setminus \{\mbox{-}\})^*$, and $L_2= (\ialphabet \setminus \{\mbox{-}\})^* b$. Then 
$$
\begin{array}{l c l}
\Pre_{R_f}(L_1) \cap \Pre_{R_f}(L_2) & = & \left(a (\ialphabet \setminus \{\mbox{-}\})^* \cup a (\ialphabet \setminus \{\mbox{-}\})^*\mbox{-} \ialphabet^* \cup  \ialphabet^*  \mbox{-} a (\ialphabet \setminus \{\mbox{-}\})^* \cup \ialphabet^* \mbox{-} a (\ialphabet \setminus \{\mbox{-}\})^* \mbox{-} \ialphabet^* \right)\ \cap \\ 
&& \left((\ialphabet \setminus \{\mbox{-}\})^* b \cup (\ialphabet \setminus \{\mbox{-}\})^* b  \mbox{-} \ialphabet^* \cup \ialphabet^*  \mbox{-} (\ialphabet \setminus \{\mbox{-}\})^* b  \cup \ialphabet^* \mbox{-} (\ialphabet \setminus \{\mbox{-}\})^* b \mbox{-} \ialphabet^*\right),
\end{array}
$$ 
which is different from 
$$\Pre_{R_f}(L_1 \cap L_2)=a (\ialphabet \setminus \{\mbox{-}\})^* b \cup a (\ialphabet \setminus \{\mbox{-}\})^* b \mbox{-} \ialphabet^* \cup \ialphabet^*  \mbox{-} a (\ialphabet \setminus \{\mbox{-}\})^* b  \cup \ialphabet^* \mbox{-} a (\ialphabet \setminus \{\mbox{-}\})^* b \mbox{-} \ialphabet^*.$$
For instance, $a \mbox{-} b \in \Pre_{R_f}(L_1) \cap \Pre_{R_f}(L_2) $, but $a \mbox{-} b \not \in \Pre_{R_f}(L_1 \cap L_2)$. \qed
\end{example}

The distributivity of the pre-image of string functions means that, for each $y := f(\vec{x})$ and $y\in \Aut$ with $\Lang(\Aut) = \Lang(\Aut_1) \cap \cdots \cap \Lang(\Aut_s)$, we can compute $\Pre_{R_f}(\Aut)$ by \emph{separately} computing the pre-image $\Pre_{R_f}(\Aut_i)$ for each assertion $y \in \Aut_i$, i.e., no product construction is performed.

Moreover, to obtain the \expspace{} upper bound, we need 
to carefully examine the complexity of the pre-image computation for each string operation.
%
The pre-image computation of the $\replaceall$ function has recently been addressed in \cite{CCHLW18}. In the following, we 
tackle other string functions in \strlinefft{}.
To this end, we utilise a succinct representation of the conjunction of a special class of regular languages, called \emph{conjunctive \FA{}s}.
\begin{definition}[Conjunctive \FA]
A conjunctive \FA{} is a pair $(\Aut, \conacc)$, where $\Aut=(\controls, \transrel)$ is a transition graph and $\conacc \subseteq \controls \times \controls$ is called a \emph{conjunctive acceptance condition}.  The language defined by $(\Aut, \conacc)$, denoted by $\Lang(\Aut, \conacc)$, is $\bigcap \limits_{(q,q') \in \controls} \Lang((\controls, q, \{q'\}, \transrel))$. The size of $(\Aut, \conacc)$, is defined as $|\controls|$.
\end{definition}
Note that the conjunctive \FA{} $(\Aut, \conacc)$ is exponentially more succinct than the product automaton of all \FA{}s $(\controls, q, \{q'\}, \transrel)$ for $(q,q') \in \conacc$. For a string operation $f: (\ialphabet^*)^r \rightarrow \ialphabet^*$, we use $\Pre_{R_f}(\Aut, \conacc)$ to denote the pre-image of $\Lang(\Aut, \conacc)$ under $R_f$. A \emph{conjunctive representation } of $\Pre_{R_f}(\Aut,\conacc)$ is a collection $((\Aut^{(1)}_j, \conacc^{(1)}_j), \ldots, (\Aut^{(r)}_j, \conacc^{(r)}_j))_{1 \le j \le n}$, where each atom $(\Aut^{(i)}_j, \conacc^{(i)}_j)$ is a conjunctive \FA{}.

Based on conjunctive \FA{}s and the fact that the product construction of regular constraints can be avoided, we show the $\expspace$ upper bound for \strlinefft{}. 

\begin{proposition}\label{prop-fft-upper}
	The path feasibility of \strlinefft{} is in $\expspace$.
\end{proposition}

The proof of Proposition~\ref{prop-fft-upper} can be found in \shortlong{the full version of this article}{Appendix~\ref{app-fft-upper} in the supplementary material}.
For the \expspace{} lower bound, 
it has been shown in \cite{LB16} that $\strline$ with $\FunFT$ and $\concat$ is  $\expspace$-hard. (Note that all transducers used in the reduction therein are functional.) To complement this result, we show that $\strline$ with $\replaceall$ solely is already $\expspace$-hard. This result is interesting in its own right. In particular, it solves an open problem left over in \cite{CCHLW18}.

\begin{proposition} \label{prop:expspace-lower}
	The path feasibility problem for $\strline[\replaceall]$ is $\expspace$-hard.
\end{proposition}

The proof of the above proposition is by a reduction from a tiling problem over an exponentially wide corridor (see Section~\ref{sec:tiling-def} for a definition).
We give the full proof in \shortlong{the full version of this article}{Appendix~\ref{sec:expspace-hardness-appendix} in the supplementary material}.

\subsection{The fragment \strlineconcat{}} \label{sec:strlineconact}

Theorem~\ref{thm-fft} shows that in $\strline$, if the transducers are restricted to be functional and one-way, then the complexity of the path feasibility problem becomes \expspace-complete. In the following, we show that the same complexity bound holds if 
the $\replaceall_e$ function is made unary, whereas the transducers are allowed to be relational.  
We remark that the proof of the complexity upper-bound deviates from that of  \strlinefft.

\begin{theorem}\label{thm-ftconrev}
	The path feasibility of \strlineconcat{} is \expspace-complete.
\end{theorem}

The lower-bound in Theorem~\ref{thm-ftconrev} follows from that in \cite{LB16} for $\concat$ and \FT{}s. We focus on the upper-bound. Let $S$ be a symbolic execution in \strlineconcat{}. Without loss of generality, we assume that the subject parameters of the $\sreplaceall$ functions in $S$ are always string variables (otherwise it can be eliminated). 
We have the four-step procedure below:
\begin{description}
\item[(1)] At first, for each assertion $\ASSERT{R(\vec{x})}$ in $S$, we nondeterministically choose one disjunct $(\Aut_1, \ldots, \Aut_\arity)$ of the representation of $R$, replace $\ASSERT{R(\vec{x})}$ with the sequence of assertions $\ASSERT{x_1 \in \Aut_1}; \ldots; \ASSERT{x_\arity \in \Aut_\arity}$. Let $S_1$ be the resulting program. Note that the size of $S_1$ is linear in that of $S$.
\item[(2)] Transform each assignment of the form $y:= \replaceall_e(x, u)$ in $S_1$ with $e$ a regular expression and $u$ a string constant, into $y:= T_{e, u}(x)$, where $T_{e, u}$ is an \FT{} corresponding to $\replaceall_e(\cdot, u)$ that can be constructed from $e, u$ in exponential time  (\cite{CCHLW18}). Let $S_2$ denote the resulting program.
\item[(3)] Remove all the occurrences of the $\concat$ operator from $S_2$, resulting in $S_3$. This step can be done in nondeterministic exponential time w.r.t. the size of $S_1$ (not $S_2$), thus the size of $S_3$ is at most exponential in that of $S_1$.

%
\item[(4)] Finally, reduce $S_3$ into a program $S_4$ that contains no occurrences of the $\reverse$ function. The reduction is done in polynomial time.
\end{description}
Note that $S_4$ is a program that contains only \FT{}s and assertions of the form $y \in \Aut$ the size of which is exponential in that of $S$. According to \cite{BFL13}[Theorem 6.7], the path feasibility of a symbolic execution that contains only \FT{}s and assertions of the form $y \in \Aut$ can be solved in polynomial space. Therefore, we conclude that the path feasibility of \strlineconcat{} is in nondeterministic exponential space, thus in \expspace{} by Savitch's theorem.

Since the first step is clear and the second step was shown in \cite{CCHLW18}, we will focus on the third and fourth step.

The third step is similar to the proof of Theorem~5 in \cite{LB16}: The main idea is to do a bottom-up induction on the structure of the dependency graph (namely starting from the input variables) and split each variable into multiple segments represented by fresh variables. (In the current setting, one additional inductive case should be added to deal with the $\reverse$ function.)
Crucially the number of fresh variables introduced in the third step depends only on the structure of the dependency graph, but is independent of 
the size  of the transducers or automata in $S_2$. Therefore, there are at most exponentially (in the size of $S_1$) many fresh variables are introduced. The transducers or automata  in $S_3$ are obtained from those of $S_2$ by designating one initial and one final state respectively. Therefore, the size of $S_3$ is at most exponential in that of $S_1$.

We mainly describe the fourth step, i.e., to eliminate the $\reverse$ function while preserving path feasibility. During this course, we need to introduce \emph{reversed transducers}. For an \FT{} $T$, we define \FT{} $T^\revsym$ by reversing
the direction of each transition of $T$ and swapping initial and final states.  
%
%

Let $X$ denote the set of variables occurring in $S_3$. The fourth step comprises the following substeps.
\begin{description}
\item[(4.1)] For each $x \in X$, add a new variable $x^{(\revsym)}$, which intuitively denotes the reverse of $x$.
\item[(4.2)] Remove each $y := \reverse(x)$ in $S_3$ and replace each occurrence of $y$  in $S_3$ with $x^{(\revsym)}$.

\item[(4.3)] In this substep, we intend to make sure for each variable $x \in X$, it cannot be the case that both $x$ and $x^{(\revsym)}$ occur. To this end, 
we sequentially process the statements of $S_3$, as follows.

Mark all the remaining assignments as unprocessed. Repeat the following procedure until all the assignments are processed:
\begin{itemize}
\item If the first unprocessed statement is of the form $y := T(x^{(\revsym)})$ (resp. $y^{(\revsym)} = T(x^{(\revsym)})$) and $x$ occurs in some processed assignment, then replace $y := T(x^{(\revsym)})$ (resp. $y^{(\revsym)} := T(x^{(\revsym)})$) by $y^{(\revsym)} := T^{\revsym}(x)$ (resp. $y := T(x)$) and mark the new statement as processed.

%
\item If the first unprocessed statement is of the form $y^{(\revsym)} := T(x)$ (resp. $y := T(x)$) and $x^{(\revsym)}$ occurs in some processed assignment, then replace $y^{(\revsym)} := T(x)$ (resp. $y := T(x)$) by $y := T^\revsym(x^{(\revsym)})$ (resp. $y^{(\revsym)} := T^\revsym(x^{(\revsym)})$) and mark the new statement as processed. 
%
\item For all the other cases, mark the first unprocessed statement as processed.
\end{itemize}
By induction, we can show that in each step of the above procedure, for each variable $x \in X$, at most one of $x$ or $x^{(\revsym)}$ occurs in the processed assignments. We then have that, after the step ({\bf 4.3}), for each $x \in X$, 
$x$ and $x^{(\revsym)}$ can \emph{not} both occur in the assignments.
\item[(4.4)] For each variable $x \in X$, if $x$ occurs in the assignments, then replace each regular constraint of the form $x^{(\revsym)} \in \Lang(\Aut)$ by $x \in \Lang(\Aut^\revsym)$, otherwise, replace each regular constraint of the form $x \in \Lang(\Aut)$ by $x^{(\revsym)} \in \Lang(\Aut^\revsym)$.
\end{description}
Recall that we use $S_4$ to denote the symbolic execution that results from executing the fourth step on $S_3$.
From the arguments above, we know that for each $x \in X$, exactly one of $x$ or $x^{(\revsym)}$ occurs in $S_4$, but not both. Therefore, $S_4$ is a symbolic execution in \strlineconcat{} that contains only \FT{}s and regular constraints.

The following example illustrates the fourth step.
\begin{example}
Consider the symbolic execution
\[
z := T(x);\ y := \reverse(x);\  z' := T'(y);\ \ASSERT{x \in \Aut_1};\ \ASSERT{z \in \Aut_2};\ \ASSERT{z' \in \Aut_3}
\]
where $T, T'$ are \FT{}s and $\Aut_1, \Aut_2, \Aut_3$ are \FA{}s. In the first substep, we add the variables $x^{(\revsym)}, y^{(\revsym)}$, $z^{(\revsym)}, (z')^{(\revsym)}$. In the second substep, we remove $y: = \reverse(x)$ and replace $y$ with $x^{(\revsym)}$, resulting in the program
\[
z := T(x);\  z' := T'(x^{(\revsym)});\ \ASSERT{x \in \Aut_1};\ \ASSERT{z \in \Aut_2};\ \ASSERT{z' \in \Aut_3}.
\]
In the third substep, since when processing $z' := T'(x^{(\revsym)})$, $x$ has already occurred in the processed assignment $z:= T(x)$, we transform $z' := T'(x^{(\revsym)})$ into $(z')^{(\revsym)} := T'(x)$, resulting in the program
\[
z := T(x);\  (z')^{(\revsym)} := T'(x);\ \ASSERT{x \in \Aut_1};\ \ASSERT{z \in \Aut_2};\ \ASSERT{z' \in \Aut_3}.
\]
Note that after the third substep, $x, z, (z')^{(\revsym)}$ occur in the assignments, but none of $x^{(\revsym)}, z^{(\revsym)}$ and $z'$ do. In the fourth substep, we replace $\ASSERT{z' \in \Aut_3}$ with $\ASSERT{(z')^{(\revsym)} \in \Aut^\revsym_3}$ and get the following  symbolic execution in \strlineconcat{} where only \FT{}s and regular constraints occur,
\[
z := T(x);\  (z')^{(\revsym)} := T'(x);\ \ASSERT{x \in \Aut_1};\ \ASSERT{z \in \Aut_2};\ \ASSERT{(z')^{(\revsym)} \in \Aut^\revsym_3}.
\]
\qed
\end{example}

As mentioned in Section~\ref{sec:intro}, our algorithm
reduces the problem to constraints which can be handled by the
existing solver SLOTH \cite{HJLRV18}.


\paragraph{Extensions with length constraints.} Apart from regular constraints in the assertion, length constraints are another class of commonly used constraints in string manipulating programs. Some simple length constraints can be encoded by regular constraints, as partially exemplified in Example~\ref{ex:length}. Here, we show that when \strlineconcat{} is extended with (much more) general length constraints, the \expspace{} upper bound can still be preserved.
We remark that, in contrast, if \strlinefft{} is extended with length constraints, then the path feasibility problem becomes undecidable, which has already been shown in \cite{CCHLW18}.

To specify length constraints properly, we need to slightly extend our constraint language. In particular, we consider variables of type $\intnum$, which are usually referred to as \emph{integer variables} and range over the set $\Nat$ of natural numbers. Recall that, in previous sections, we have used $x, y, z, \ldots$ to denote the variables of $\str$ type.  Hereafter we typically use $\mathfrak{l}, \mathfrak{m}, \mathfrak{n}, \ldots$ to denote the variables of $\intnum$. The
choice of omitting negative integers is for simplicity. Our
results can be easily extended to the case where $\intnum$ includes negative integers.

\begin{definition}[Length assertions] \label{def:intconst}
	A length assertion is of the form
	$\ASSERT{a_1t_1+\cdots+a_nt_n\leq d}$,
where $a_1, \ldots, a_n,d\in \mathbb{Z}$ are  integer constants (represented in binary), and each \emph{term} $t_i$ is either
	\begin{enumerate}
		\item an integer variable $\mathfrak{n}$;
		\item $|x|$ where $x$ is a  string variable; or
		\item $|x|_a$ where $x$ is string variable and $a\in \Sigma$ is a constant letter.
	\end{enumerate}
Here, $|x|$ and $|x|_a$ denote the length of $x$ and the number of occurrences of $a$ in $x$, respectively.
\end{definition}

\begin{theorem} \label{thm:length}
The path feasibility of \strlineconcat{} extended with length assertions is \expspace-complete.
\end{theorem}
The proof of the theorem is given in \shortlong{the full
version of this article}{Appendix~\ref{app:thmlength} in the supplementary material}.
%
%
%
%
%
%
%
Other related constraints, such as character assertions and $\indexof$ assertions, can be encoded by length assertions defined in Definition~\ref{def:intconst} together with regular constraints. Since they are not the focus of this paper, we omit the details here. 

We observe that Theorem~\ref{thm:length} suggests that the two semantic conditions (roughly speaking, being a recognisable relation) is only a sufficient, but not necessary, condition of decidability of path feasibility. This is because length assertions can easily deviate from recognisable relations (for instance $|x_1|=|x_2|$ does not induce a recognisable relation over $\Sigma^*$), but decidability still remains.

\section{Implementation} \label{sec:impl}

We have implemented our decision procedure for path feasibility in a
new tool OSTRICH, which is built on top of the SMT solver Princess~\cite{princess08}.
OSTRICH implements the optimised decision procedure for string functions as
described in Section~\ref{sec:implemented-alg} (i.e. using distributivity
of regular constraints across pre-images of functions)
and has built-in support for concatenation, reverse, \FunFT{}, and $\replaceall$.
Moreover, since the optimisation only requires that string operations are functional, we can also support additional functions that satisfy \reginvrel, such as $\replace_e$ which replaces only the first (leftmost and longest) match of a regular expression.
OSTRICH is extensible and new string functions can be added quite simply (Section~\ref{sec:extensibility}).

Our implementation adds a new theory solver for conjunctive formulas
representing path feasibility problems to Princess
(Section~\ref{sec:dfs}).
This means that we support disjunction as well as conjunction in formulas,
as long as every conjunction of literals fed to the theory solver
corresponds to a path feasibility problem.
OSTRICH also implements a number of
optimisations to efficiently compute pre-images of relevant functions
(Section~\ref{sec:preImgOpt}). OSTRICH is entirely written in Scala and
is open-source.
We report on our experiments with OSTRICH in Section~\ref{sec:experiments}.
The tool is available on GitHub\footnote{\url{https://github.com/pruemmer/ostrich}}.
The artifact is available on the ACM DL.

\subsection{Depth-First Path Feasibility}
\label{sec:dfs}

\begin{algorithm}[h]
  \small
  \KwIn{Sets~$\mathit{active}, \mathit{passive}$ of regex constraints
    $x \in L$;
    acyclic set~$\mathit{funApps}$ of assignments $x := f(\bar y)$.}
  \KwResult{Either $\mathit{Model}(m)$ with $m$ a model satisfying
    all constraints and function applications;\newline
    or $\mathit{Conflict}(s)$ with $s$ a set of regex constraints that
    is inconsistent with $\mathit{funApps}$.}

  \Begin{
    \eIf{$\mathit{active} = \emptyset$}{
      \tcc{Extract a model by solving constraints and evaluating functions}
      $\mathit{leafTerms} \leftarrow
      \{ x \mid x \text{~occurs in~} \mathit{passive} \cup \mathit{funApps} \}
      \setminus
      \{x \mid (x := f(\bar y)) \in \mathit{funApps} \}$\;
      $\mathit{leafModel} \leftarrow
      \{ x \mapsto w \mid x \in \mathit{leafTerms}, ~w
      \text{~satisfies all constraints on~} x \text{~in~} \mathit{passive}\}$\;
      $m \leftarrow
      \mu m\,.\, \mathit{leafModel} \cup
        \{ x \mapsto f(\bar w) \mid (x := f(\bar y)) \in \mathit{funApps},~
        m(\bar y) = \bar w \text{~is defined}\}$\;
      \Return{$\mathit{Model}(m)$}
    }{
      \tcc{Compute the pre-image of one of the active constraints}
      choose a constraint~$x \in L$ in $\mathit{active}$\;
      \eIf{there is an assignment~$x := f(y_1, \ldots, y_r)$ defining $x$
        in $\mathit{funApps}$}{
        $\mathit{cset} \leftarrow \{x \in L\}$ \tcc*{start constructing a conflict set}
        compute the pre-image $f^{-1}(L) = \bigcup_{i=1}^n L_1^{(i)} \times \cdots \times L_r^{(i)}$\;
        $\mathit{act} \leftarrow \mathit{active} \setminus \{x \in L\},~~
        \mathit{pas} \leftarrow \mathit{passive} \cup \{x \in L\}$\;
        \For{$i \leftarrow 1$ \KwTo $n$}{
          $\mathit{newRegexes} \leftarrow
          \{y_1 \in L_1^{(i)}, \ldots, y_r \in L_r^{(i)}\} \setminus (\mathit{act} \cup \mathit{pas})$\;
          \eIf{$\mathit{act} \cup \mathit{pas} \cup \mathit{newRegexes}$ is inconsistent}{
            compute an unsatisfiable core $c \subseteq \mathit{act} \cup \mathit{pas} \cup \mathit{newRegexes}$\;
            $\mathit{cset} \leftarrow \mathit{cset} \cup (c \setminus \{y_1 \in L_1^{(i)}, \ldots, y_r \in L_r^{(i)}\})$\;
          }{
          \Switch{$\mathit{findModel}(\mathit{act}
            \cup \mathit{newRegexes},\;
          \mathit{pas},\;
          \mathit{funApps})$}{
          \uCase{$\mathit{Model}(m)$}{
            \Return{$\mathit{Model}(m)$}\;
          }
          \Case{$\mathit{Conflict}(s)$}{
            \eIf{$s \cap \mathit{newRegexes} = \emptyset$}{
              \Return{$\mathit{Conflict}(s)$}  \tcc*{back-jump}
            }{
              $\mathit{cset} \leftarrow \mathit{cset} \cup (s \setminus \{y_1 \in L_1^{(i)}, \ldots, y_r \in L_r^{(i)}\})$\;
            }
          }
          }
        }}
        \Return{$\mathit{Conflict}(cset)$} \tcc*{backtrack}
      }{
        \Return{$\mathit{findModel}(\mathit{active} \setminus \{x \in L\},\;
          \mathit{passive} \cup \{x \in L\},\;
          \mathit{funApps})$}\;
      }
    }
  }

  \caption{Recursive function~$\mathit{findModel}$
    defining depth-first model construction for
    SL \label{alg:dfs}}
\end{algorithm}

We first discuss the overall decision procedure for path feasibility
implemented in OSTRICH. The procedure performs depth-first exploration
of the search tree resulting from repeatedly splitting the
disjunctions (or unions) in the recognisable pre-images of
functions. Similar to the DPLL/CDCL architecture used in SAT solvers,
our procedure computes conflict sets and applies
back-jumping~\cite{tinelli_dpll_2004} to skip irrelevant parts
of the search tree.

For solving, a path feasibility problem is represented as a
set~$\mathit{funApps}$ of assignments~$x := f(\bar y)$, and a set
$\mathit{regex}$ of regular expression constraints~$x \in L$, with $x$
being a string variable and $\bar y$ a vector of string variables. The
set~$\mathit{funApps}$ by definition contains at most one assignment
for each variable~$x$, and by nature of a path   there are no
cyclic dependencies. We make two further simplifying assumptions:
(i)~the set $\mathit{regex}$ is on its own satisfiable, i.e., the
constraints given for each variable~$x$ are consistent; and (ii)~for
each variable that occurs as the left-hand side of an
assignment~$x := f(\bar y)$, there is at least one
constraint~$x \in L$ in the set $\mathit{regex}$.  Assumption~(i)
boils down to checking the non-emptiness of intersections of regular
languages, i.e., to a reachability check in the
product transition system, and can be done efficiently in practical
cases.  Both assumptions could easily be relaxed, at the cost of making the
algorithm slightly more involved.

\subsubsection{The exploration function}

The algorithm is presented as a recursive
function~$\mathit{findModel}$ in pseudo-code in
Algorithm~\ref{alg:dfs}. The function maintains two sets
$\mathit{active}, \mathit{passive}$ of regular expression constraints,
kept as parameters, and in addition receives the set~$\mathit{funApps}$
as parameter. \emph{Active} regular expression constraints are those
for which no pre-images have been computed yet, while
$\mathit{passive}$ are the constraints that have already gone through
pre-image computation. In the initial
call~$\mathit{findModel}(\mathit{regex}, \emptyset, \mathit{funApps})$
of the function, $\mathit{active}$ is chosen to be $\mathit{regex}$,
while $\mathit{passive}$ is empty.

Depending on the status of a path feasibility problem,
$\mathit{findModel}$ can produce two outcomes:
\begin{itemize}
\item $\mathit{Model}(m)$, where $m$ is maps string
  variables to words that satisfy the regular expression
  constraints in $\mathit{active} \cup \mathit{passive}$ and the
  assignments in $\mathit{funApps}$, interpreted as equations.
\item $\mathit{Conflict}(s)$, with
  $s \subseteq \mathit{active} \cup \mathit{passive}$ being a
  \emph{conflict set}, i.e., a set of constraints that is inconsistent
  with $\mathit{funApps}$. The set~$s \cup \mathit{funApps}$ can be
  interpreted as an unsatisfiable core of the path feasibility
  problem, and is used to identify irrelevant splits during the
  search.
\end{itemize}

\subsubsection{Implementation in detail}

Model construction terminates with a positive result when the
set~$\mathit{active}$ becomes empty (line~2), in which case it is only
necessary to compute a model~$m$ (lines~3--6). For this, the algorithm
first computes the set of variables that do not occur on the left-hand
side of any assignment (line~3). The values of such leaf terms can be
chosen arbitrarily, as long as all derived regular expression
constraints are satisfied (line~4). The values of all other variables
are extracted from the assignments in $\mathit{funApps}$:
whenever an assignment~$x := f(\bar y)$ is found for which all
argument variables already have a value, also the value of the
left-hand side~$x$ is known (line~5).

Otherwise, one of the active constraints~$x \in L$ is selected for
pre-image computation~(line~8). For the correctness it is irrelevant
in which order the constraints are selected, and branching heuristics
from the SAT world might be applicable. Our current implementation
selects the constraints in the fixed order in which the constrained
variables occur on the path, starting with constraints at the end of
the path. If $x$ is a left-hand side of an assignment, the pre-image
of $L$ is a recognisable relation that can be represented through
regular languages (line~11); the constraint~$x \in L$ then becomes
passive in subsequent recursive calls (line~12).

The algorithm then iterates over the disjuncts of the pre-image
(line~13), generates new regular expression constraints for the
function arguments (line~14), and checks whether any disjuncts
lead to a solution. During the iteration over disjuncts, the
algorithm builds up a conflict set~$\mathit{cset}$ collecting
constraints responsible for absence of a solution (lines~10,
17, 26).  If the new constraints are inconsistent with
existing constraints (line~15), the disjunct does not have to be
considered further; the algorithm then computes a (possibly minimal)
unsatisfiable subset of the constraints, and adds it to the conflict
set~$\mathit{cset}$. Otherwise, $\mathit{findModel}$ is called
recursively (line~19). If the recursive call produces a model, no
further search is necessary, and the function returns
(lines~20-21). Similarly, if the recursive call reports a conflict
that is independent of the generated regular expression constraints,
it follows that no solution can exist for any of the disjuncts of the
pre-image, and the function can return immediately (lines~23--24). In
case of other conflicts, the conflict set~$\mathit{cset}$ is extended
(line~26), and finally returned to explain why no model could be found
(line~27).

\subsection{Optimisation of Pre-Image Computation}
\label{sec:preImgOpt}

We have optimised the pre-image computation of the concatenation and $\replaceall$ operations.

\subsubsection{Concatenation}

The pre-image of a regular language~$L$ for
concatenation~$x := y \circ z$ can be computed by representing $L$ as
an FA, say with $n$ states, and generating a
union~$\bigcup_{i=1}^n L_1^{(i)} \times L_2^{(i)}$ in which the
accepting state of $L_1^{(i)}$ and the initial state of $L_2^{(i)}$
iterate over all $n$ states of $L$~\cite{Abdulla14}. In practice, most of the resulting $n$~cases are
immediately inconsistent with other regular expression constraints in
$\mathit{active} \cup \mathit{passive}$; this happens for instance
when the length of $y$ or $z$ are already predetermined by other
constraints. Our implementation therefore filters the considered
languages~$L_1^{(i)}, L_2^{(i)}$ upfront using length information
extracted from other constraints.

\subsubsection{The $\replaceall$ Function}
\label{sec:replaceall-explanation}

We implement $\replaceall_e$ by reduction to $\replaceall_a$ for a single character $a$.
We translate all
$x := \replaceall_e(y, z)$
into
$y' := T_e(y); x := \replaceall_\internalchar(y', z)$
where $\internalchar$ is a special character disjoint from the rest of the alphabet.
The transducer $T_e$ copies the contents of $y$ and replaces all left-most and longest subwords satisfying the regular expression $e$ with $\internalchar$.  This construction uses a \emph{parsing automaton} and details can be found in~\cite{CCHLW18}.
We first recall how $\replaceall_a$ can be tackled, before explaining the inefficiencies and the solution we use.

\paragraph{Naive Recognisability}

Suppose we have an \FA{} $\Aut_x$ giving a regular constraint on $x$.
We need to translate this automaton into a sequence of regular constraints on $y$ and $z$.
To do this, we observe that all satisfying assignments to $x$ must take the form
$u_1 w_z u_2 w_z \ldots w_z u_n$
where $w_z$ is the value of $z$, and
$u_1 a u_2 a \ldots a u_n$
is the value of $y$.
Moreover, each word $u_i$ cannot contain the character $a$ since it would have been replaced by $w_z$.
The (satisfying) assignment to $x$ must be accepted by $\Aut_x$.
Thus, we can extract from an accepting run of $\Aut_x$ a set of pairs of states $Q_z$, which is the set of all pairs $(q, q')$ such that the run of $\Aut_x$ moves from $q$ to $q'$ while reading a copy of $w_z$.
Then, we can obtain a new \FA{} $\Aut_y$ by removing all $a$-transitions from $\Aut_x$ and then adding $a$-transitions
$(q, a, q')$
for each
$(q, q') \in Q_z$.
It is easy to verify that there is an accepting run of $\Aut_y$ over
$u_1 a u_2 a \ldots a u_n$.
Similarly, we define $\Aut_z$ to be the intersection of $\Aut_x(q, \{q'\})$ for all $(q, q') \in Q_z$.
We know by design that there is an accepting run of $\Aut_z$ over $w_z$.

The value of $Q_z$ above was extracted from an accepting run of $\Aut_x$.
There are, of course, many possible accepting runs of $\Aut_x$, each leading to a different value of $Q_z$, and thus a different $\Aut_y$ and $\Aut_z$.
Since each $Q_z$ a set of pairs of states of $\Aut_x$, there are only a finite number of values that can be taken by $Q_z$.
Thus, we can show the pre-image of
$x := \replaceall_a(y, a, z)$
is recognisable by enumerating all possible values of $Q_z$.
For each $Q_z$ we can produce a pair of automata
$(\Aut^{Q_z}_y, \Aut^{Q_z}_z)$
as described above.
Thus, the pre-image can be expressed by
$\bigcup_{Q_z} \brac{\lang{\Aut^{Q_z}_x} \times \lang{\Aut^{Q_z}_y}}$.

\paragraph{Optimised Recognisability}

A problem with the above algorithm is that there are an exponential number of possible values of $Q_z$.
For example, if $\Aut_x$ has $10$ states, there are $2^{100}$ possible values of $Q_z$;
it is infeasible to enumerate them all.
To reduce the number of considered pairs, we use the notion of a \emph{Cayley Graph}~\cite{Z81,D67}.
Note, this technique was already used by Chen~\cite{ChenThesis18,yan-tool} as part of an implementation of~\cite{CCHLW18}.

Given an automaton $\Aut$ and a word $w$, we define
\[
    \caleybox{w} =
    \setcomp{(q_0, q_n)}
            {\text{there exists a run $q_0, \ldots, q_n$ of $\Aut$ over $w$}}.
\]
The number of distinct
$\caleybox{w}$
is finite for a given \FA{}.
We define Cayley Graphs in the context of \FA{}.

\begin{definition}[Cayley Graph]
    Given an \FA{}
    $\Aut = (\ialphabet, \controls, q_0, \finals, \transrel)$
    the \emph{Cayley Graph} of $\Aut$ is a graph
    $(V, E)$
    with nodes
    $V = \setcomp{\caleybox{w}}{w \in \ialphabet^\ast}$
    and edges
    $E = \setcomp{(\caleybox{w}, a, \caleybox{wa})}
                 {w \in \ialphabet^\ast \land a \in \ialphabet}$.
\end{definition}

For a given automaton, it is straight-forward to compute the Cayley Graph using a fixed point iteration: begin with $\caleybox{\emptyword} \in V$, then, until a fixed point is reached, take some $\caleybox{w}$ in $V$ and add $\caleybox{wa}$ for all $a \in \ialphabet$.
Note $\caleybox{wa}$ is a simple composition of $\caleybox{w}$ and
$\caleybox{a} = \setcomp{(q, q')}{(q, a, q') \text{ is an edge of $\Aut$}}$.

We claim that instead of enumerating all $Q_z$, one only needs to consider all
$\caleybox{w} \in V$.
Since $\caleybox{w}$ can be a value of $Q_z$, the restriction does not increase the set of potential solutions.
We need to argue that it does not reduce them.
Hence, take some satisfying value
$u_1 w_z u_2 w_z \ldots w_z u_n$
of $x$.
One can easily verify that $w_z$ is accepted by the $\Aut_z$ constructed from $\caleybox{w_z}$.
Moreover,
$u_1 a u_2 a \ldots a u_n$
is accepted by the corresponding $\Aut_y$.
Thus we have not restricted the algorithm.

From experience, it is reasonable to hope that the Cayley Graph has far fewer nodes than the set of all potential $Q_z$.
Moreover, we can further improve the enumeration by considering any other regular constraints
$\Aut^1_z, \ldots, \Aut^m_z$
that may exist on the value of $z$.
As a simple example, if we had
$\text{\ASSERT{$z \in b^\ast$}};
 x := \replaceall_a(y, z);
 \text{\ASSERT{$x \in \Aut_x$}}$
where $\Aut_x$ has initial state $q_0$ and accepting states $q_1$ and $q_2$ with transitions
$(q_0, a, q_1)$ and $(q_0, b, q_1)$, there is no need to consider
$\caleybox{a} = \set{(q_0, q_1)}$
since $z$ cannot take the value $a$ without violating
$\text{\ASSERT{$z \in b^\ast$}}$.

Using this observation, assume we already know that $z$'s value must be accepted by
$\Aut^1_z, \ldots, \Aut^m_z$.
Instead of building the Cayley Graph alone, we build a product of the Cayley Graph and
$\Aut^1_z, \ldots, \Aut^m_z$
on the fly.
This product has states
$(\caleybox{w}, q_1, \ldots, q_m)$
for some $w \in \ialphabet^\ast$ and
$q_1$, \ldots, $q_m$
states of
$\Aut^1_z, \ldots, \Aut^m_z$
respectively.
The only $\caleybox{w}$ we need to consider are those such that
$(\caleybox{w}, q_1, \ldots, q_m)$
is a (reachable) product state, and, moreover, $q_1, \ldots, q_m$ are accepting states of
$\Aut^1_z, \ldots, \Aut^m_z$.

This technique first speeds up the construction of the Cayley Graph by limiting which nodes are generated, and second, avoids values of $Q_z$ which are guaranteed to be unsatisfying.

\subsection{Extensibility}
\label{sec:extensibility}

One may extend OSTRICH to include any
string function with a recognisable pre-image \emph{without having to worry
about other parts of the solver}.  We give an example of adding a new
$\reverse$ function.

We have defined a Scala trait \mintinline{scala}{PreOp}.
To add a string function, one defines a new Scala object with this trait.
This requires two methods described below.
An example object is given in Figure~\ref{fig:preop-reverse}.

The first method is \mintinline{scala}{eval}, which implements the string
function.
It takes a sequence of strings represented as sequence of integers.
For $\reverse$, this method reverses the argument.
For $\replaceall_a(x, y)$, the \mintinline{scala}{eval} function would take a sequence of two arguments (the values of $x$ and $y$ respectively) and replace all $a$ characters in $x$ with the value of $y$ to produce the result.
The return value can be \mintinline{scala}{None} if the function is not applicable to the given arguments.

The second method \mintinline{scala}{apply} performs the pre-image computation.
It takes two arguments: a sequence of constraints on the arguments (\mintinline{scala}{argumentConstraints}), and a constraint on the result (\mintinline{scala}{resultConstraint}).
The result constraint is represented as an \mintinline{scala}{Automaton} that accepts the language for which we are computing $f^{-1}$.
The argument constraints are represented as sequences of sequences of automata:
for each argument of the function there will be one sequence of automata.
These constraints give further information on what is known about the constraints on the arguments of the function.
For example, if we are computing
$x := \reverse(y)$
and elsewhere we have determined $y$ must be accepted by both $\Aut_1$ and $\Aut_2$, then the first (and only) element of the argument constraints will be the sequence $\Aut_1, \Aut_2$.
It is not necessary to use these constraints, but they may help to optimise the pre-image computation (as in the case of $\replaceall$ described above).

The return value of \mintinline{scala}{apply} is a pair.
The first element is the pre-image (a recognisable relation).
It is represented as an iterator over sequences of automata, where each sequence corresponds to a tuple
$(\Aut_1, \ldots, \Aut_n)$
of the relation.
The second element is a list of the argument constraints used during the pre-image computation; this information is needed to compute correct
conflict sets in Algorithm~\ref{alg:dfs}.
If the argument constraints were not used to optimise the pre-image computation, this value can be an empty list.
If the arguments were used, then those constraints which were used should be returned in the same format as the argument constraints.

For convenience, assume that we have already implemented a reversal operation on automata as
\mintinline{scala}{AutomataUtils.reverse}.
A \mintinline{scala}{PreOp} object for the $\reverse$ function is given in Figure~\ref{fig:preop-reverse}.
The \mintinline{scala}{apply} method reverses the result constraint, and returns an iterator over this single automaton.
The second component of the return value is an empty list since the argument constraints were not used.
The \mintinline{scala}{eval} method simply reverses its first (and only) argument.

\begin{figure}
    \begin{minted}[fontsize=\footnotesize]{scala}
        object ReversePreOp extends PreOp {
          def apply(argumentConstraints : Seq[Seq[Automaton]], resultConstraint : Automaton)
                  : (Iterator[Seq[Automaton]], Seq[Seq[Automaton]]) = {
            val revAut = AutomataUtils.reverse(resultConstraint)
            (Iterator(Seq(revAut)), List())
          }
          def eval(arguments : Seq[Seq[Int]]) : Option[Seq[Int]] = Some(arguments(0).reverse)
        }
    \end{minted}
    \caption{\label{fig:preop-reverse}A \mintinline{scala}{PreOp} for the $\reverse$ function.}
\end{figure}

To complete the addition of the function, the reverse function is
registered in the OSTRICH string theory object. The function is then
ready to be used in an SMT-LIB problem, e.g. by writing the assertion
\verb+(assert (= x (user_reverse y)))+.

\subsection{Experiments}
\label{sec:experiments}

We have compared OSTRICH with a number of existing solvers on a wide range of benchmarks.
In particular, we compared OSTRICH with
    CVC4~1.6~\cite{cvc4},
    SLOTH~\cite{HJLRV18}, and
    Z3\footnote{Github version 2aeb814f4e7ae8136ca5aeeae4d939a0828794c8} 
  configured to use the Z3-str3 string solver~\cite{Z3-str3}.
We considered several sets of benchmarks which are described in the next sub-section.
The results are given in Section~\ref{sec:results}.

In~\cite{HJLRV18} SLOTH was compared with S3P~\cite{TCJ16} where inconsistent behaviour was reported.
We contacted the S3P authors who report that the current code is unsupported;
moreover, S3P is being integrated with Z3.
Hence, we do not compare with this tool.

\subsubsection{Benchmarks}

The first set of benchmarks we call \transducerbench.
It combines the benchmark sets of Stranger \cite{Stranger} and the mutation XSS benchmarks of~\cite{LB16}.
The first (sub-)set appeared in~\cite{HJLRV18} and contains instances manually derived from PHP programs taken from the website of Stranger~\cite{Stranger}.
It contains 10 formulae (5 sat, 5 unsat) each testing for the absence of the vulnerability pattern \verb+.*<script.*+ in the program output.
These formulae contain between 7 and 42 variables, with an average of 21.
The number of atomic constraints ranges between 7 and 38 and averages 18.
These examples use disjunction, conjunction, regular constraints, and 
concatenation, $\replaceall$. They also contain several one-way functional 
transducers (defined in SMTLIB in \cite{HJLRV18}) encoding functions such as 
\verb+addslashes+ and \verb+trim+ used by the programs. Note that transducers
have been known for some time to be a good framework for
specifying sanitisers and browser transductions 
(e.g., see the works by Minamide, Veanes, Saxena, and others
\cite{Min05,BEK,web-model,DV13}), and a library of transducer specifications for
such functions is available (e.g.~see the language BEK \cite{BEK}). 

The second (sub-)set was used by~\cite{LB16,HJLRV18} and consists of 8
formulae taken from~\cite{LB16,Kern14}.
These examples explore mutation XSS vulnerabilities in JavaScript programs.
They contain between 9 and 12 variables, averaging 9.75, and 9-13 atomic constraints, with an average of 10.5.
They use conjunctions, regular constraints, concatenation, $\replaceall$, and transducers providing functions such as \verb+htmlEscape+ and \verb+escapeString+.

Our next set of benchmarks, \slogbench,
came from the SLOG tool~\cite{fang-yu-circuits} and consist of 3,392 instances.
They are derived from the security analysis of real web applications and contain 1-211 string variables (average 6.5) and 1-182 atomic formula (average 5.8).
We split these benchmarks into two sets \slogbenchr\ and \slogbenchra.
Each use conjunction, disjunction, regular constraints, and concatenation.
The set \slogbenchr\ contains 3,391 instances and uses $\replace$.
\slogbenchra\ contains 120 instances using the $\replaceall$ operation.

Our next set of benchmarks \kaluzabench\ is the well-known set of \emph{Kaluza}
benchmarks~\cite{Berkeley-JavaScript} restricted to those instances which
satisfy our semantic conditions (roughly ${\sim}80$\% of the benchmarks).
Kaluza contains concatenation, regular
constraints, and length constraints, most of which admit regular monadic
decomposition.
There are 37,090 such benchmarks (28\,032 sat).

We also considered the benchmark set of~\cite{yan-tool,ChenThesis18}.
This contains 42 hand-crafted benchmarks using regular constraints, concatenation, and $\replaceall$ with variables in both argument positions.
The benchmarks contain 3-7 string variables and 3-9 atomic constraints.

\subsubsection{Results}
\label{sec:results}

We compared the tools on an AMD Opteron 2220 SE machine,
running 64-bit Linux and Java~1.8, with the heap space of
each job limited to 2~GB.
We used a timeout of 240s for each Kaluza problem, and 600s for
the other benchmarks.
Figure~\ref{fig:results} summarises our findings as cactus plots.
For each benchmark set, we plot the time in seconds on a cubic-root scale against the number of benchmarks solved (individually) within that time.
The extent of each line on the Time axis gives the maximum time in seconds required to solve any instance in the set.
When a solver is not plotted it is because it was unable to analyse the benchmark set.

\begin{figure}
    \includegraphics[width=0.9\linewidth]{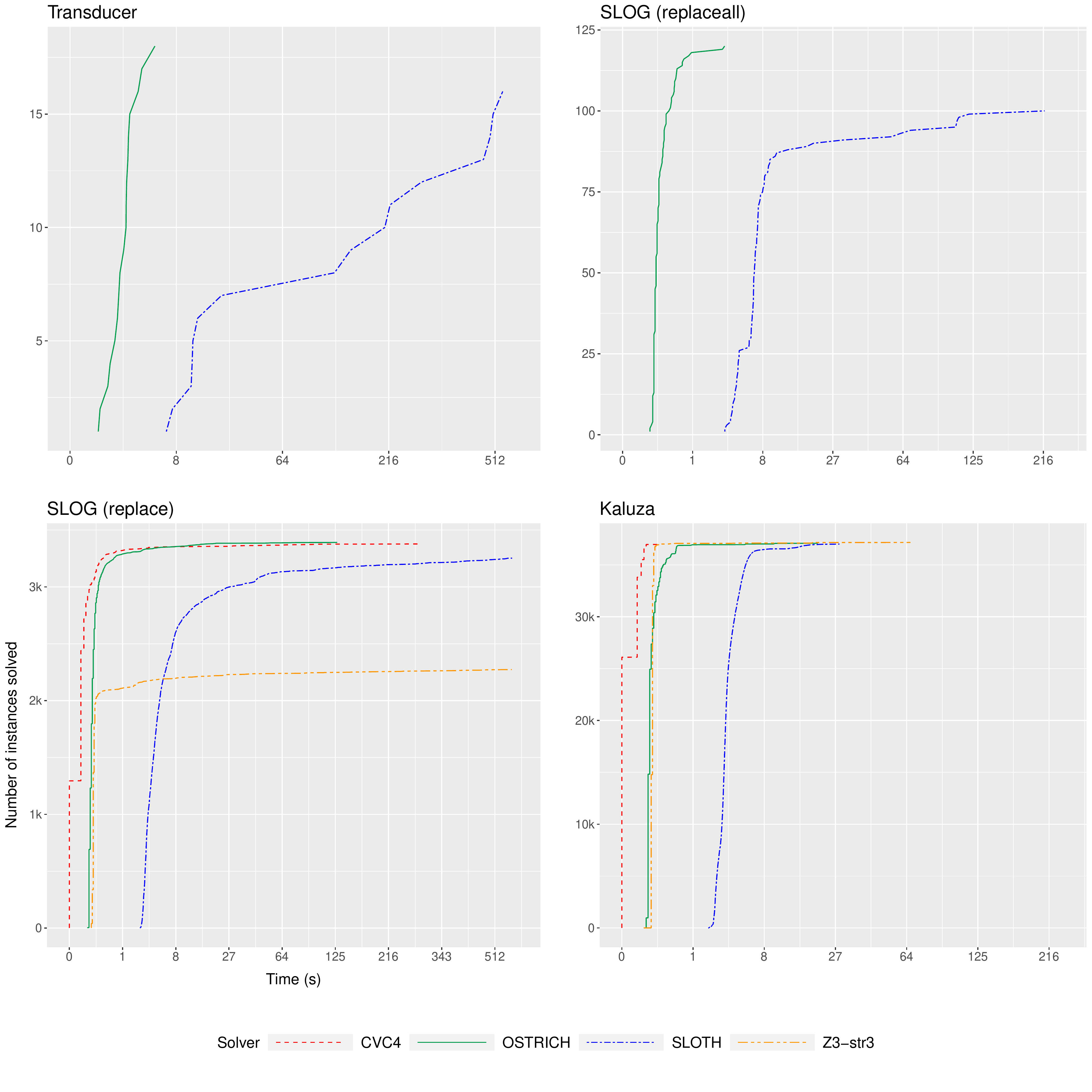}
    \vspace{-4ex}
    \caption{\label{fig:results}Comparison of solvers on several sets of benchmarks.}
\end{figure}

For the \transducerbench\ set, OSTRICH solved all benchmarks taking a maximum of 4s.
SLOTH did not answer 2 instances and was slower on the rest.
This set is not supported by CVC4 or Z3-str3.

For the \slogbenchra\ set, OSTRICH solved all 120 instances within a few seconds, while SLOTH only solved 100.
CVC4 and Z3-str3 were omitted as they do not support $\replaceall$ constraints.

For the \slogbenchr\ set, OSTRICH was also able to solve all instances within the timeout.
CVC4 was able to solve all but 13 of the benchmarks.
Similarly, SLOTH was unable to solve 138 instances, while Z3-str3 could not solve 1,118.
We note that Z3-str3 gave inconsistent results for 18 of these benchmarks,
an issue that could not be conclusively resolved before submission..

For the \kaluzabench\ set, CVC4, Z3-str3, and OSTRICH were able to solve all instances.
Since SLOTH does not support length constraints, it reported errors in 81 of these benchmarks.
Otherwise all instances were solved within the timeout.

We were unable to install the tool of~\cite{yan-tool} for comparison with the benchmarks of~\cite{yan-tool,ChenThesis18}.
Since OSTRICH was the only available tool supporting variables in both arguments of $\replaceall$, we do not provide a plot.
We note that the most difficult benchmark took OSTRICH 1.56s to answer, with the second hardest requiring 0.34s.

Overall, CVC4 was the fastest for the constraints it was able to answer,
while
OSTRICH was the only solver which answered all benchmark instances.
However, the runtime differences are fractions of a second.
In terms of completeness guarantees and the type of string constraints supported, SLOTH is our nearest competitor, with the main difference being that OSTRICH supports variables in both argument positions of $\replaceall$, while SLOTH will only accept constant strings as the second argument.
Our results show that OSTRICH outperforms SLOTH on all benchmark sets.




\section{Conclusion}\label{sec-conc}

We proposed two general semantic conditions which together ensure the 
decidability of path feasibility with complex string operations including
$\replaceall$, transducers, and concatenation. Our semantic conditions are 
satisfied by a 
wide range of complex string operations and subsume various existing string 
constraint languages
(e.g.~\cite{LB16,CCHLW18}) and many current existing benchmarks
(e.g.~\cite{Berkeley-JavaScript,HJLRV18,Stranger,fang-yu-circuits,LB16}).
Based on the semantic conditions, we developed a conceptually simple and
generic decision procedure with an extensible architecture that allow a user to 
easily incorporate a user-defined function.
After providing theoretical evidence via computational complexity that these 
semantic conditions might be too general to provide an efficient decision
procedure, we proposed two different solutions.
We advocated the first solution (prohibit nondeterminism in string operations)
whenever possible since it permits a highly effective optimisation of the solver
based on a kind of distributivity property of regular constraints across string
functions. In fact, the extra restriction imposed by this is satisfied by
most (but not all) existing benchmarking examples.
We developed a new string solver OSTRICH that implements the first solution and 
demonstrate its
%
    efficacy against other competitive solvers on most existing benchmarks.


\subsubsection*{Acknowledgments}

We are grateful to the anonymous referees for their constructive and detailed comments.
T. Chen is supported by the
Engineering and Physical Sciences Research Council under Grant No.~{EP/P00430X/1} and the
Australian Research Council under Grant No.~{DP160101652, DP180100691}.
    M. Hague is supported by the
    Engineering and Physical Sciences Research Council
    under Grant No.~{EP/K009907/1}.
    A.~Lin is supported by the European Research Council (ERC) under the European
    Union's Horizon 2020 research and innovation programme (grant agreement no
    759969).
    P.\ R\"ummer is supported by the Swedish Research Council (VR)
    under grant~2014-5484, and by the Swedish Foundation for Strategic
    Research (SSF) under the project WebSec (Ref.\ RIT17-0011).
    Z. Wu is supported by the
    National Natural Science Foundation of China
    under Grant No.~61472474, ~61572478, and ~61872340,
    the INRIA-CAS joint research project ``Verification, Interaction, and Proofs''.


\newpage

\bibliography{references}

\begin{thebibliography}{10}

\bibitem{Abdulla17}
Parosh~Aziz Abdulla, Mohamed~Faouzi Atig, Yu{-}Fang Chen, Bui~Phi Diep,
  Luk{\'{a}}s Hol{\'{\i}}k, Ahmed Rezine, and Philipp R{\"{u}}mmer.
\newblock Flatten and conquer: a framework for efficient analysis of string
  constraints.
\newblock In {\em Proceedings of the 38th {ACM} {SIGPLAN} Conference on
  Programming Language Design and Implementation, {PLDI} 2017, Barcelona,
  Spain, June 18-23, 2017}, pages 602--617, New York, NY, USA, 2017. ACM.

\bibitem{trau18}
Parosh~Aziz Abdulla, Mohamed~Faouzi Atig, Yu{-}Fang Chen, Bui~Phi Diep,
  Luk{\'{a}}s Hol{\'{\i}}k, Ahmed Rezine, and Philipp R{\"{u}}mmer.
\newblock {TRAU:} {SMT} solver for string constraints.
\newblock In {\em Formal Methods in Computer Aided Design, {FMCAD 2018}}, 2018.
\newblock To appear.

\bibitem{Abdulla14}
Parosh~Aziz Abdulla, Mohamed~Faouzi Atig, Yu{-}Fang Chen, Luk{\'{a}}s
  Hol{\'{\i}}k, Ahmed Rezine, Philipp R{\"{u}}mmer, and Jari Stenman.
\newblock String constraints for verification.
\newblock In {\em Computer Aided Verification - 26th International Conference,
  {CAV} 2014}, pages 150--166. Springer, 2014.

\bibitem{AD11}
Rajeev Alur and Jyotirmoy~V. Deshmukh.
\newblock Nondeterministic streaming string transducers.
\newblock In {\em Automata, Languages and Programming - 38th International
  Colloquium, {ICALP} 2011, Zurich, Switzerland, July 4-8, 2011, Proceedings,
  Part {II}}, pages 1--20. Springer, 2011.

\bibitem{BK08}
Christel Baier and Joost-Pieter Katoen.
\newblock {\em Principles of Model Checking (Representation and Mind Series)}.
\newblock The MIT Press, 2008.

\bibitem{BFL13}
Pablo Barcel{\'{o}}, Diego Figueira, and Leonid Libkin.
\newblock Graph logics with rational relations.
\newblock {\em Logical Methods in Computer Science}, 9(3), 2013.

\bibitem{SMT-chapter}
Clark~W. Barrett, Roberto Sebastiani, Sanjit~A. Seshia, and Cesare Tinelli.
\newblock Satisfiability modulo theories.
\newblock In {\em Handbook of Satisfiability}, pages 825--885. {IOS} Press,
  2009.

\bibitem{BLSS03}
Michael Benedikt, Leonid Libkin, Thomas Schwentick, and Luc Segoufin.
\newblock Definable relations and first-order query languages over strings.
\newblock {\em J. {ACM}}, 50(5):694--751, 2003.

\bibitem{Z3-str3}
Murphy Berzish, Vijay Ganesh, and Yunhui Zheng.
\newblock Z3str3: {A} string solver with theory-aware heuristics.
\newblock In {\em 2017 Formal Methods in Computer Aided Design, {FMCAD} 2017,
  Vienna, Austria, October 2-6, 2017}, pages 55--59. {IEEE}, 2017.

\bibitem{BTV09}
Nikolaj Bj{\o}rner, Nikolai Tillmann, and Andrei Voronkov.
\newblock Path feasibility analysis for string-manipulating programs.
\newblock In {\em Tools and Algorithms for the Construction and Analysis of
  Systems, {TACAS 2009}}, pages 307--321. Springer, 2009.

\bibitem{BGG97}
Egon B{\"{o}}rger, Erich Gr{\"{a}}del, and Yuri Gurevich.
\newblock {\em The Classical Decision Problem}.
\newblock Perspectives in Mathematical Logic. Springer, 1997.

\bibitem{buchi}
J~Richard B{\"u}chi and Steven Senger.
\newblock Definability in the existential theory of concatenation and
  undecidable extensions of this theory.
\newblock In {\em The Collected Works of J. Richard B{\"u}chi}, pages 671--683.
  Springer, 1990.

\bibitem{CW07}
Thierry Cachat and Igor Walukiewicz.
\newblock The complexity of games on higher order pushdown automata.
\newblock {\em CoRR}, abs/0705.0262, 2007.

\bibitem{KLEE}
Cristian Cadar, Daniel Dunbar, and Dawson Engler.
\newblock Klee: Unassisted and automatic generation of high-coverage tests for
  complex systems programs.
\newblock In {\em Proceedings of the 8th USENIX Conference on Operating Systems
  Design and Implementation}, OSDI'08, pages 209--224, Berkeley, CA, USA, 2008.
  USENIX Association.

\bibitem{EXE}
Cristian Cadar, Vijay Ganesh, Peter~M. Pawlowski, David~L. Dill, and Dawson~R.
  Engler.
\newblock Exe: Automatically generating inputs of death.
\newblock In {\em Proceedings of the 13th ACM Conference on Computer and
  Communications Security}, CCS '06, pages 322--335, New York, NY, USA, 2006.
  ACM.

\bibitem{symbex-survey}
Cristian Cadar and Koushik Sen.
\newblock Symbolic execution for software testing: Three decades later.
\newblock {\em Commun. ACM}, 56(2):82--90, February 2013.

\bibitem{CCG06}
Olivier Carton, Christian Choffrut, and Serge Grigorieff.
\newblock Decision problems among the main subfamilies of rational relations.
\newblock {\em {ITA}}, 40(2):255--275, 2006.

\bibitem{CCHLW18}
Taolue Chen, Yan Chen, Matthew Hague, Anthony~W. Lin, and Zhilin Wu.
\newblock What is decidable about string constraints with the replaceall
  function.
\newblock {\em {PACMPL}}, 2({POPL}):3:1--3:29, 2018.

\bibitem{ChenThesis18}
Yan Chen.
\newblock {Solving String Constraints with ReplaceAll Function}.
\newblock Master's thesis, State Key Laboratory of Computer Science, Institute
  of Software, Chinese Academy of Sciences, China, 2018.

\bibitem{yan-tool}
Yan Chen.
\newblock {Z3}-replaceall.
\newblock \url{https://github.com/TinyYan/z3-replaceAll}, 2018.
\newblock Referred in Jan 2018.

\bibitem{choffrut-survey}
Christian Choffrut.
\newblock Relations over words and logic: {A} chronology.
\newblock {\em Bulletin of the {EATCS}}, 89:159--163, 2006.

\bibitem{DV13}
Loris D'Antoni and Margus Veanes.
\newblock Static analysis of string encoders and decoders.
\newblock In {\em Verification, Model Checking, and Abstract Interpretation,
  {VMCAI}}, pages 209--228. Springer, 2013.

\bibitem{SMT-CACM}
Leonardo De~Moura and Nikolaj Bj{\o}rner.
\newblock Satisfiability modulo theories: introduction and applications.
\newblock {\em Communications of the ACM}, 54(9):69--77, 2011.

\bibitem{D67}
J.~D\'enes.
\newblock Connections between transformation semigroups and graphs.
\newblock In {\em Theory of Graphs}. Gordon \& Breach, 1967.

\bibitem{diekert}
Volker Diekert.
\newblock Makanin's {A}lgorithm.
\newblock In M.~Lothaire, editor, {\em Algebraic Combinatorics on Words},
  volume~90 of {\em Encyclopedia of Mathematics and its Applications},
  chapter~12, pages 387--442. Cambridge University Press, 2002.

\bibitem{EH01}
Joost Engelfriet and Hendrik~Jan Hoogeboom.
\newblock {MSO} definable string transductions and two-way finite-state
  transducers.
\newblock {\em {ACM} Trans. Comput. Log.}, 2(2):216--254, 2001.

\bibitem{FGRS13}
Emmanuel Filiot, Olivier Gauwin, Pierre{-}Alain Reynier, and
  Fr{\'{e}}d{\'{e}}ric Servais.
\newblock From two-way to one-way finite state transducers.
\newblock In {\em 28th Annual {ACM/IEEE} Symposium on Logic in Computer
  Science, {LICS} 2013, New Orleans, LA, USA, June 25-28, 2013}, pages
  468--477. {IEEE} Computer Society, 2013.

\bibitem{Vijay-length}
Vijay Ganesh, Mia Minnes, Armando Solar{-}Lezama, and Martin~C. Rinard.
\newblock Word equations with length constraints: What's decidable?
\newblock In {\em Hardware and Software: Verification and Testing - 8th
  International Haifa Verification Conference, {HVC} 2012, Haifa, Israel,
  November 6-8, 2012. Revised Selected Papers}, pages 209--226. Springer, 2012.

\bibitem{DART}
Patrice Godefroid, Nils Klarlund, and Koushik Sen.
\newblock Dart: Directed automated random testing.
\newblock {\em SIGPLAN Not.}, 40(6):213--223, June 2005.

\bibitem{HJLRV18}
Luk{\'{a}}s Hol{\'{\i}}k, Petr Janku, Anthony~W. Lin, Philipp R{\"{u}}mmer, and
  Tom{\'{a}}s Vojnar.
\newblock String constraints with concatenation and transducers solved
  efficiently.
\newblock {\em {PACMPL}}, 2({POPL}):4:1--4:32, 2018.

\bibitem{BEK}
Pieter Hooimeijer, Benjamin Livshits, David Molnar, Prateek Saxena, and Margus
  Veanes.
\newblock Fast and precise sanitizer analysis with {BEK}.
\newblock In {\em 20th {USENIX} Security Symposium, San Francisco, CA, USA,
  August 8-12, 2011, Proceedings}. {USENIX} Association, 2011.

\bibitem{HU79}
John~E. Hopcroft and Jeffrey~D. Ullman.
\newblock {\em Introduction to Automata Theory, Languages and Computation}.
\newblock Addison-Wesley, 1979.

\bibitem{J16}
Artur Jez.
\newblock Recompression: {A} simple and powerful technique for word equations.
\newblock {\em J. {ACM}}, 63(1):4:1--4:51, 2016.

\bibitem{J01}
Neil~D. Jones.
\newblock The expressive power of higher-order types or, life without cons.
\newblock {\em Journal of Functional Programming}, 11(1):55--94, 2001.

\bibitem{Kern14}
Christoph Kern.
\newblock Securing the tangled web.
\newblock {\em Commun. {ACM}}, 57(9):38--47, 2014.

\bibitem{HAMPI}
Adam Kiezun, Vijay Ganesh, Shay Artzi, Philip~J. Guo, Pieter Hooimeijer, and
  Michael~D. Ernst.
\newblock {HAMPI:} {A} solver for word equations over strings, regular
  expressions, and context-free grammars.
\newblock {\em {ACM} Trans. Softw. Eng. Methodol.}, 21(4):25, 2012.

\bibitem{king76}
James~C. King.
\newblock Symbolic execution and program testing.
\newblock {\em Commun. {ACM}}, 19(7):385--394, 1976.

\bibitem{Kozen-automata}
Dexter Kozen.
\newblock {\em Automata and Computability}.
\newblock Springer, 1997.

\bibitem{KS08}
Daniel Kroening and Ofer Strichman.
\newblock {\em Decision Procedures}.
\newblock Springer, 2008.

\bibitem{cvc4}
Tianyi Liang, Andrew Reynolds, Cesare Tinelli, Clark Barrett, and Morgan
  Deters.
\newblock A {DPLL(T)} theory solver for a theory of strings and regular
  expressions.
\newblock In {\em Computer Aided Verification - 26th International Conference,
  {CAV} 2014}, pages 646--662. Springer, 2014.

\bibitem{Libkin03}
Leonid Libkin.
\newblock Variable independence for first-order definable constraints.
\newblock {\em {ACM} Trans. Comput. Log.}, 4(4):431--451, 2003.

\bibitem{LB16}
Anthony~W. Lin and Pablo Barcel\'{o}.
\newblock String solving with word equations and transducers: Towards a logic
  for analysing mutation {XSS}.
\newblock In {\em Proceedings of the 43rd Annual ACM SIGPLAN-SIGACT Symposium
  on Principles of Programming Languages}, POPL '16, pages 123--136. Springer,
  2016.

\bibitem{McMillan}
K.~L. McMillan.
\newblock {\em Symbolic model checking}.
\newblock Kluwer, 1993.

\bibitem{Min05}
Yasuhiko Minamide.
\newblock Static approximation of dynamically generated web pages.
\newblock In {\em Proceedings of the 14th international conference on World
  Wide Web, {WWW} 2005}, pages 432--441. {ACM}, 2005.

\bibitem{Morvan00}
Christophe Morvan.
\newblock On rational graphs.
\newblock In {\em Foundations of Software Science and Computation Structures,
  Third International Conference, {FOSSACS} 2000}, pages 252--266. Springer,
  2000.

\bibitem{tinelli_dpll_2004}
Robert Nieuwenhuis, Albert Oliveras, and Cesare Tinelli.
\newblock Abstract {DPLL} and abstract {DPLL} modulo theories.
\newblock In {\em Logic for Programming, Artificial Intelligence, and
  Reasoning, 11th International Conference, {LPAR} 2004}, volume 3452 of {\em
  LNCS}, pages 36--50. Springer, 2004.

\bibitem{pippenger-book}
Nicholas Pippenger.
\newblock {\em Theories of Computability}.
\newblock Cambridge University Press, 2010.

\bibitem{Plandowski}
Wojciech Plandowski.
\newblock Satisfiability of word equations with constants is in {PSPACE}.
\newblock {\em J. ACM}, 51(3):483--496, 2004.

\bibitem{princess08}
Philipp R{\"u}mmer.
\newblock A constraint sequent calculus for first-order logic with linear
  integer arithmetic.
\newblock In {\em Logic for Programming, Artificial Intelligence, and
  Reasoning, 15th International Conference, {LPAR} 2008, LNCS 5330}, pages
  274--289. Springer, 2008.

\bibitem{Berkeley-JavaScript}
Prateek Saxena, Devdatta Akhawe, Steve Hanna, Feng Mao, Stephen McCamant, and
  Dawn Song.
\newblock A symbolic execution framework for javascript.
\newblock In {\em 31st {IEEE} Symposium on Security and Privacy, S{\&}P 2010,
  16-19 May 2010, Berleley/Oakland, California, {USA}}, pages 513--528. {IEEE},
  2010.

\bibitem{jalangi}
Koushik Sen, Swaroop Kalasapur, Tasneem~G. Brutch, and Simon Gibbs.
\newblock Jalangi: a selective record-replay and dynamic analysis framework for
  javascript.
\newblock In {\em Joint Meeting of the European Software Engineering Conference
  and the {ACM} {SIGSOFT} Symposium on the Foundations of Software Engineering,
  ESEC/FSE'13, Saint Petersburg, Russian Federation, August 18-26, 2013}, pages
  488--498. {ACM}, 2013.

\bibitem{CUTE}
Koushik Sen, Darko Marinov, and Gul Agha.
\newblock Cute: a concolic unit testing engine for c.
\newblock In {\em Proceedings of the 10th European Software Engineering
  Conference held jointly with 13th {ACM} {SIGSOFT} International Symposium on
  Foundations of Software Engineering, 2005, Lisbon, Portugal, September 5-9,
  2005, {ESEC/SIGSOFT FSE} 2005}, pages 263--272. {ACM}, 2005.

\bibitem{S3}
Minh{-}Thai Trinh, Duc{-}Hiep Chu, and Joxan Jaffar.
\newblock {S3:} {A} symbolic string solver for vulnerability detection in web
  applications.
\newblock In {\em Proceedings of the 2014 {ACM} {SIGSAC} Conference on Computer
  and Communications Security, {CCS} 2014}, pages 1232--1243. {ACM}, 2014.

\bibitem{TCJ16}
Minh{-}Thai Trinh, Duc{-}Hiep Chu, and Joxan Jaffar.
\newblock Progressive reasoning over recursively-defined strings.
\newblock In {\em Computer Aided Verification - 28th International Conference,
  {CAV} 2016, Toronto, ON, Canada, July 17-23, 2016, Proceedings, Part {I}},
  pages 218--240. Springer, 2016.

\bibitem{owasp17}
Andrew van~der Stock, Brian Glas, Neil Smithline, and Torsten Gigler.
\newblock {OWASP} {T}op 10 -- 2017.
\newblock \url{https://www.owasp.org/index.php/Top\_10-2017\_Top\_10}, 2017.
\newblock Referred January 2018.

\bibitem{monadic-decomposition}
Margus Veanes, Nikolaj Bj{\o}rner, Lev Nachmanson, and Sergey Bereg.
\newblock Monadic decomposition.
\newblock {\em J. {ACM}}, 64(2):14:1--14:28, 2017.

\bibitem{fang-yu-circuits}
Hung{-}En Wang, Tzung{-}Lin Tsai, Chun{-}Han Lin, Fang Yu, and Jie{-}Hong~R.
  Jiang.
\newblock String analysis via automata manipulation with logic circuit
  representation.
\newblock In {\em Computer Aided Verification - 28th International Conference,
  {CAV} 2016, Toronto, ON, Canada, July 17-23, 2016, Proceedings, Part {I}},
  volume 9779 of {\em Lecture Notes in Computer Science}, pages 241--260.
  Springer, 2016.

\bibitem{web-model}
Joel Weinberger, Prateek Saxena, Devdatta Akhawe, Matthew Finifter, Eui
  Chul~Richard Shin, and Dawn Song.
\newblock A systematic analysis of {XSS} sanitization in web application
  frameworks.
\newblock In {\em Computer Security - {ESORICS} 2011 - 16th European Symposium
  on Research in Computer Security, Leuven, Belgium, September 12-14, 2011.
  Proceedings}, pages 150--171. Springer, 2011.

\bibitem{Stranger}
Fang Yu, Muath Alkhalaf, and Tevfik Bultan.
\newblock Stranger: An automata-based string analysis tool for {PHP}.
\newblock In {\em Tools and Algorithms for the Construction and Analysis of
  Systems, 16th International Conference, {TACAS} 2010}, pages 154--157.
  Springer, 2010.
\newblock Benchmark can be found at
  \url{http://www.cs.ucsb.edu/~vlab/stranger/}.

\bibitem{YABI14}
Fang Yu, Muath Alkhalaf, Tevfik Bultan, and Oscar~H. Ibarra.
\newblock Automata-based symbolic string analysis for vulnerability detection.
\newblock {\em Form. Methods Syst. Des.}, 44(1):44--70, 2014.

\bibitem{Z81}
Bohdan Zelinka.
\newblock Graphs of semigroups.
\newblock {\em \v{C}asopis pro p\v{e}stov\'{a}n\'{i} matematiky},
  106(4):407--408, 1981.

\bibitem{Z3-str2}
Yunhui Zheng, Vijay Ganesh, Sanu Subramanian, Omer Tripp, Julian Dolby, and
  Xiangyu Zhang.
\newblock Effective search-space pruning for solvers of string equations,
  regular expressions and length constraints.
\newblock In {\em Computer Aided Verification - 27th International Conference,
  {CAV} 2015, San Francisco, CA, USA, July 18-24, 2015, Proceedings, Part {I}},
  pages 235--254. Springer, 2015.

\bibitem{Z3-str}
Yunhui Zheng, Xiangyu Zhang, and Vijay Ganesh.
\newblock {Z3-str: a Z3-based string solver for web application analysis}.
\newblock In {\em Joint Meeting of the European Software Engineering Conference
  and the {ACM} {SIGSOFT} Symposium on the Foundations of Software Engineering,
  {ESEC/FSE} 2013, Saint Petersburg, Russian Federation, August 18-26, 2013},
  pages 114--124. {ACM}, 2013.

\end{thebibliography}

\shortlong{}{

\newpage


\appendix
\begin{center}
{\huge Appendix} \\
{\large `` Decision Procedures for Path Feasibility of String-Manipulating 
	Programs with Complex Operations''}
\end{center}

\bigskip


\section{Proof of Proposition~\ref{prop:nexpspace-hardness}: Hardness of $\strline$ with 2FTs+conc and $\strline$ with FTs+$\replaceall$}
\label{sec:two-way-lower}

We first give a full definition of tiling problems before showing the result.
We give the details here of the proof for $\strline$ with 2FTS+conc.
The proof can easily be adapted to the case of FTs~$\replaceall$ using the simulation described in Section~\ref{sec:two-way-simulation}.

\subsection{Tiling Problems}

A \emph{tiling problem} is a tuple
$\tup{\tiles, \hrel, \vrel, \inittile, \fintile}$
where
    $\tiles$ is a finite set of tiles,
    $\hrel \subseteq \tiles \times \tiles$ is a horizontal matching relation,
    $\vrel \subseteq \tiles \times \tiles$ is a vertical matching relation, and
    $\inittile, \fintile \in \tiles$ are initial and final tiles respectively.

A solution to a tiling problem over a $\linlen$-width corridor is a sequence
\[
    \begin{array}{c}
        \tile^1_1 \ldots \tile^1_\linlen \\
        \tile^2_1 \ldots \tile^2_\linlen \\
        \ldots \\
        \tile^\tileheight_1 \ldots \tile^\tileheight_\linlen
    \end{array}
\]
where
$\tile^1_1 = \inittile$,
$\tile^\tileheight_\linlen = \fintile$,
and for all
$1 \leq i < \linlen$
and
$1 \leq j \leq \tileheight$
we have
$\tup{\tile^j_i, \tile^j_{i+1}} \in \hrel$
and for all
$1 \leq i \leq \linlen$
and
$1 \leq j < \tileheight$
we have
$\tup{\tile^j_i, \tile^{j+1}_i} \in \vrel$.
Note, we will assume that $\inittile$ and $\fintile$ can only appear at the beginning and end of the tiling respectively.

Tiling problems characterise many complexity classes~\cite{BGG97}.
In particular, we will use the following facts.
\begin{itemize}
\item
    For any $\linlen$-space Turing machine, there exists a tiling problem of size polynomial in the size of the Turing machine, over a corridor of width $\linlen$, that has a solution iff the $\linlen$-space Turing machine has a terminating computation.
    In particular we will consider problems where the corridor width is
    $2^{:^{2^\linlen}}$
    where the height of the stack of exponentials is $\expheight$.
    E.g.\ when $\expheight$ is $0$ the width is $\linlen$, when $\expheight$ is $1$ the width is $2^\linlen$, when $\expheight$ is $2$ the width is $2^{2^\linlen}$ and so on.
    Solving tiling problems of width
    $2^{:^{2^\linlen}}$
    is complete for the same amount of space.

\item
    There is a fixed
    $\tup{\tiles, \hrel, \vrel, \inittile, \fintile}$
    such that for any width $\linlen$ there is a unique solution
    \[
        \begin{array}{c}
            \tile^1_1 \ldots \tile^1_\linlen \\
            \tile^2_1 \ldots \tile^2_\linlen \\
            \ldots \\
            \tile^\tileheight_1 \ldots \tile^\tileheight_\linlen
        \end{array}
    \]
    and moreover $\tileheight$ is exponential in $\linlen$.
    One such example is a Turing machine where the tape contents represent a binary number.
    The Turing machine starts from a tape containing only $0$s and finishes with a tape containing only $1$s by repeatedly incrementing the binary encoding on the tape.
    This Turing machine can be encoded as the required tiling problem.
\end{itemize}

\subsection{Large Numbers}

The crux of the proof is encoding large numbers that can take values between $1$ and $\expheight$-fold exponential.

A linear-length binary number could be encoded simply as a sequence of bits
\[
    b_0 \ldots b_\linlen \in \set{0,1}^\linlen \ .
\]
To aid with later constructions we will take a more oblique approach.
Let
$\tup{\tilesnum{1}, \hrelnum{1}, \vrelnum{1}, \inittilenum{1}, \fintilenum{1}}$
be a copy of the fixed tiling problem from the previous section for which there is a unique solution, whose length must be exponential in the width.
In the future, we will need several copies of this problem, hence the indexing here.
Fix a width $\linlen$ and let $\nmax{1}$ be the corresponding corridor length.
A \emph{level-1} number can encode values from $1$ to $\nmax{1}$.
In particular, for $1 \leq i \leq \nmax{1}$ we define
\[
    \tenc{1}{i} = \tile^i_1 \ldots \tile^i_\linlen
\]
where
$\tile^i_1 \ldots \tile^i_\linlen$
is the tiling of the $i$th row of the unique solution to the tiling problem.

A \emph{level-2} number will be derived from tiling a corridor of width $\nmax{1}$, and thus the number of rows will be doubly-exponential.
For this, we require another copy
$\tup{\tilesnum{2}, \hrelnum{2}, \vrelnum{2}, \inittilenum{2}, \fintilenum{2}}$
of the above tiling problem.
Moreover, let $\nmax{2}$ be the length of the solution for a corridor of width $\nmax{1}$.
Then for any
$1 \leq i \leq \nmax{2}$
we define
\[
    \tenc{2}{i} =
        \tenc{1}{1} \tile^i_1
        \tenc{1}{2} \tile^i_2
        \ldots
        \tenc{1}{\nmax{1}} \tile^i_{\nmax{1}}
\]
where
$\tile^i_1 \ldots \tile^i_{\nmax{1}}$
is the tiling of the $i$th row of the unique solution to the tiling problem.
That is, the encoding indexes each tile with it's column number, where the column number is represented as a level-1 number.

In general, a \emph{level-$\expheight$} number is of length $(\expheight-1)$-fold exponential and can encode numbers $\expheight$-fold exponential in size.
We use a copy
$\tup{\tilesnum{\expheight},
      \hrelnum{\expheight},
      \vrelnum{\expheight},
      \inittilenum{\expheight},
      \fintilenum{\expheight}}$
of the above tiling problem and use a corridor of width
$\nmax{\expheight-1}$.
We define $\nmax{\expheight}$ as the length of the unique solution to this problem.
Then, for any $1 \leq i \leq \nmax{\expheight}$ we have
\[
    \tenc{\expheight}{i} =
        \tenc{(\expheight-1)}{1} \tile^i_1
        \tenc{(\expheight-1)}{2} \tile^i_2
        \ldots
        \tenc{(\expheight-1)}{\nmax{(\expheight-1)}} \tile^i_{\nmax{(\expheight-1)}}
\]
where
$\tile^i_1 \ldots \tile^i_{\nmax{\expheight-1}}$
is the tiling of the $i$th row of the unique solution to the tiling problem.

\subsection{Recognising Large Numbers}

We first define a useful program
$\goodnums{\expheight}{x}$
with a single input string $x$ which can only be satisfied if $x$ is of the following form, where
$\numeq$, $\numplus$, and $\numsep$
are auxiliary symbols.
\[
    \begin{array}{c}
        \goodnums{\expheight}{x} \text{ is satisfiable} \\
        \iff \\
        x \in \brac{
            \brac{\tenc{\expheight}{1} \brac{\numeq \tenc{\expheight}{1}}^\ast}
            \numplus
            \brac{\tenc{\expheight}{2} \brac{\numeq \tenc{\expheight}{2}}^\ast}
            \numplus
            \ldots
            \numplus
            \brac{
                \tenc{\expheight}{\nmax{\expheight}}
                    \brac{\numeq \tenc{\expheight}{\nmax{\expheight}}}^\ast
            }
            \numsep
        }^\ast
    \end{array}
\]
That is, the string must contain sequences of strings that count from $1$ to $\nmax{\expheight}$.
This counting may stutter and repeat a number several times before moving to the next.
A separator $\numeq$ indicates a stutter, while $\numplus$ indicates that the next number must add one to the current number.
Finally, a $\numsep$ ends a sequence and may start again from $\tenc{\expheight}{1}$.

\paragraph{Base case $\expheight = 1$.}

We define
\[
    \goodnums{1}{x} = \left\{
        \begin{array}{l}
            y := \ap{T}{x}; \\
            \ASSERT{y \in \top}
        \end{array}
    \right.
\]
where $\top$ is a character output by $T$ when $x$ is correctly encoded.
Otherwise $T$ outputs $\bot$.

We describe how $T$ operates.
Because $T$ is two-way, it may perform several passes of the input $x$.
Recall that $T$ requires $x$ to contain a word of the form
\[
    \brac{
        \brac{\tenc{1}{1} \brac{\numeq \tenc{1}{1}}^\ast}
        \numplus
        \brac{\tenc{1}{2} \brac{\numeq \tenc{1}{2}}^\ast}
        \numplus
        \ldots
        \numplus
        \brac{
            \tenc{1}{\nmax{1}}
                \brac{\numeq \tenc{1}{\nmax{1}}}^\ast
        }
        \numsep
    }^\ast
\]
and each $\tenc{1}{i}$ is the $i$th row of the unique solution to the tiling problem of width $\linlen$.
The passes proceed as follows.
If a passes fails, the transducer outputs $\bot$ and terminates.
If all passes succeed, the transducer outputs $\top$ and terminates.
\begin{itemize}
\item
    During the first pass $T$ verifies that the input is of the form
    \[
        \brac{
            \tilesnum{1}^\linlen
            \brac{\set{\numeq,\numplus} \tilesnum{1}^\linlen}^\ast
            \numsep
        }^\ast \ .
    \]

\item
    During the second pass the transducer verifies that the first block of
    $\tilesnum{1}^\linlen$,
    and all blocks of
    $\tilesnum{1}^\linlen$
    immediately following a $\numsep$ have $\inittilenum{1}$ as the first tile.
    Simultaneously, it can verify that all blocks of
    $\tilesnum{1}^\linlen$
    immediately preceding a $\numsep$ finish with the tile $\fintilenum{1}$.
    Moreover, it checks that $\inittilenum{1}$ and $\fintilenum{1}$ do not appear elsewhere.

\item
    During the third pass $T$ verifies the horizontal tiling relation.
    That is, every contiguous pair of tiles $\tile, \tile'$ in $x$ must be such that
    $(\tile, \tile') \in \hrelnum{1}$.
    This can easily be done by storing the last character read into the states of $T$.

\item
    The vertical tiling relation and equality checks are verified using $\linlen$ more passes.
    During the $j$th pass, the $j$th column is tested.
    The transducer $T$ stores in its state the tile in the $j$th column of the first block of $\tilesnum{1}^\linlen$ or any block immediately following $\numsep$.
    (The transducer can count to $\linlen$ in its state.)
    It then moves to the $j$th column of the next block of $\tilesnum{1}^\linlen$, remembering whether the blocks were separated with $\numeq$, $\numplus$, or $\numsep$.
    If the separator was $\numeq$ the transducer checks that the $j$th tile of the current block matches the tile stored in the state (i.e.~is equal to the preceding block).
    If the separator was $\numplus$ the transducer checks that the $j$th tile of the current block is related by $\vrelnum{1}$ to the stored tile.
    In this case the current $j$th tile is stored and the previously stored tile forgotten.
    Finally, if the separator was $\numsep$ there is nothing to check.
    If any check fails, the pass will also fail.
\end{itemize}

If all passes succeed, we know that $x$ contains a word where blocks separated by $\numeq$ are equal
(since all positions are equal, as verified individually by the final $\linlen$ passes),
blocks separated by $\numplus$ satisfy $\vrelnum{1}$ in all positions,
the $\inittilenum{1}$ tile appears at the start of all sequences separated by $\numsep$ and each such sequence ends with $\fintilenum{1}$, and
finally the horizontal tiling relation is satisfied at all times.
Thus, $x$ must be of the form
\[
    \brac{
        \brac{\tenc{1}{1} \brac{\numeq \tenc{1}{1}}^\ast}
        \numplus
        \brac{\tenc{1}{2} \brac{\numeq \tenc{1}{2}}^\ast}
        \numplus
        \ldots
        \numplus
        \brac{
            \tenc{1}{\nmax{1}}
                \brac{\numeq \tenc{1}{\nmax{1}}}^\ast
        }
        \numsep
    }^\ast
\]
as required.

\paragraph{Inductive case $\expheight$.}

We define a program
$\goodnums{\expheight}{x}$
such that
\[
    \begin{array}{c}
        \goodnums{\expheight}{x} \text{ is satisfiable} \\
        \iff \\
        x \in \brac{
            \brac{\tenc{\expheight}{1} \brac{\numeq \tenc{\expheight}{1}}^\ast}
            \numplus
            \brac{\tenc{\expheight}{2} \brac{\numeq \tenc{\expheight}{2}}^\ast}
            \numplus
            \ldots
            \numplus
            \brac{
                \tenc{\expheight}{\nmax{\expheight}}
                    \brac{\numeq \tenc{\expheight}{\nmax{\expheight}}}^\ast
            }
            \numsep
        }^\ast \ .
    \end{array}
\]
Assume, by induction, we have a program
$\goodnums{\expheight-1}{x}$
which already satisfies this property (for $\expheight-1$).
We define
\[
    \goodnums{\expheight}{x} =
    \left\{
        \begin{array}{l}
            y = \ap{T}{x}; \\
            \goodnums{\expheight-1}{y}
        \end{array}
    \right.
\]
where $T$ is a transducer that behaves as described below.
Note, the reference to
$\goodnums{\expheight-1}{y}$
is not a procedure call, since these are not supported by our language.
Instead, the procedure is inlined, with its input variable $x$ replaced by $y$ and other variables renamed to avoid clashes.

The transducer will perform several passes to make several checks.
If a check fails it will halt and output a symbol $\bot$, which means that $y$ can no longer satisfy
$\goodnums{\expheight-1}{y}$.
During normal execution $T$ will make checks that rely on level-$(\expheight-1)$ numbers appearing in the correct sequence or being equal.
To ensure these properties hold, $T$ will write these numbers to $y$ and rely on these properties then being verified by
$\goodnums{\expheight-1}{y}$.
The passes behave as follows.
\begin{itemize}
\item
    During the first pass $T$ verifies that $x$ belongs to the regular language
    \[
        \brac{
            \brac{
                \brac{
                    \brac{
                        \brac{\tilesnum{1}^\linlen \tilesnum{2}}^\ast \tilesnum{3}
                    }^\ast
                    \cdots
                }^\ast
                \tilesnum{\expheight}
            }^\ast
            \brac{
                \set{\numeq,\numplus}
                \brac{
                    \brac{
                        \brac{
                            \brac{\tilesnum{1}^\linlen \tilesnum{2}}^\ast \tilesnum{3}
                        }^\ast
                        \cdots
                    }^\ast
                    \tilesnum{\expheight}
                }^\ast
            }^\ast
            \numsep
         }^\ast
         \ .
    \]
    This can be done with a polynomial number of states.

\item
    During the second pass $T$ will verify that the first instance of
    $\tilesnum{\expheight}$
    appearing in the word or after a $\numsep$ is $\inittilenum{\expheight}$.
    Similarly, the final instance of any
    $\tilesnum{\expheight}$
    before any $\numsep$ is $\fintilenum{\expheight}$.
    Moreover, it checks that
    $\inittilenum{\expheight}$
    and
    $\fintilenum{\expheight}$
    do not appear elsewhere.

\item
    During the third pass $T$ will verify the horizontal tiling relation
    $\hrelnum{\expheight}$.
    Each block (separated by $\numeq$, $\numplus$, or $\numsep$) is checked in turn.
    There are two components to this.
    \begin{itemize}
    \item
        The indexing of the tiles must be correct.
        That is, the first tile of the block must be indexed
        $\tenc{\expheight-1}{1}$,
        the second
        $\tenc{\expheight-1}{2}$,
        through to
        $\tenc{\expheight-1}{\nmax{\expheight-1}}$.
        Thus, $T$ copies directly the instance of
        $\brac{
            \brac{
                \brac{\tilesnum{1}^\linlen \tilesnum{2}}^\ast \tilesnum{3}
            }^\ast
            \cdots
        }^\ast$
        preceding each
        $\tilesnum{\expheight}$
        to the output tape, followed immediately by
        $\numplus$
        as long as the character after
        $\tilesnum{\expheight}$
        is not a separator from
        $\set{\numeq,\numplus,\numsep}$.
        Otherwise, it is the end of the block and $\numsep$ is written.

        Hence,
        $\goodnums{\expheight-1}{y}$
        will verify that the output for each block is
        $\tenc{\expheight-1}{1}
         \numplus \cdots \numplus
         \tenc{\expheight-1}{\nmax{\expheight-1}}$
        which enforces that the indexing of the tiles is correct.

    \item
        Horizontally adjacent tiles must satisfy
        $\hrelnum{\expheight}$.
        This is done by simply storing the last read tile from
        $\tilesnum{\expheight}$
        in the state of $T$.
        Then whenever a new tile from
        $\tilesnum{\expheight}$
        is seen without a separator $\numeq$, $\numplus$, or $\numsep$, then it can be checked against the previous tile and
        $\hrelnum{\expheight}$.
    \end{itemize}

\item
    The transducer $T$ then performs a non-deterministic number of passes to check the vertical tiling relation.
    We will use
    $\goodnums{\expheight-1}{y}$
    to ensure that $T$ in fact performs
    $\nmax{\expheight-1}$
    passes, the first checking the first column of the tiling over
    $\tilesnum{\expheight}$,
    the second checking the second column, and so on up to the
    $\nmax{\expheight-1}$th column.

    Note, we know from the previous pass that each row of the tiling is indexed correctly.
    In the sequel, let us use the term ``session'' to refer to the sequences of characters separated by $\numsep$.

    Each pass of $T$ checks a single column (across all sessions).
    At the start of each session, $T$ moves non-deterministically to the start of some block
    $\brac{
        \brac{
            \brac{\tilesnum{1}^\linlen \tilesnum{2}}^\ast \tilesnum{3}
        }^\ast
        \cdots
    }^\ast
    \tilesnum{\expheight}$
    (without passing $\numeq$, $\numplus$, or $\numsep$).
    It then copies the tiles from
    $\brac{
        \brac{
            \brac{\tilesnum{1}^\linlen \tilesnum{2}}^\ast \tilesnum{3}
        }^\ast
        \cdots
    }^\ast$
    to $y$ and saves the tile from
    $\tilesnum{\expheight}$
    in its state before moving to the next separator from
    $\set{\numeq,\numplus,\numsep}$.
    In the case of $\numsep$ nothing needs to be checked and $T$ continues to the next session or finishes the pass if there are no more sessions.
    In the case of $\numeq$ or $\numplus$ the transducer remembers this separator and moves non-deterministically to the start of some block (without passing another $\numeq$, $\numplus$, or $\numsep$).
    It then writes $\numeq$ to $y$ as it is intended that $T$ choose the same column as before.
    This will be verified by
    $\goodnums{\expheight-1}{y}$.
    To aid with this $T$ copies the tiles from
    $\brac{
        \brac{
            \brac{\tilesnum{1}^\linlen \tilesnum{2}}^\ast \tilesnum{3}
        }^\ast
        \cdots
    }^\ast$
    to $y$.
    It can then check the tile from
    $\tilesnum{\expheight}$.
    If the remembered separator was $\numeq$ then this tile must match the saved one.
    If it was $\numplus$ then this tile must be related by
    $\vrelnum{\expheight}$
    to the saved one.
    If this succeeds , $T$ stores the new tile and forgets the old and continues to the next separator to continue checking
    $\vrelnum{\expheight}$.

    At the end of the pass (checking a single column from all sessions) then $T$ will either have failed and written $\bot$ or written a sequence of level-$(\expheight-1)$ numbers to $y$ separated by $\numeq$. That is
$
        \tenc{\expheight-1}{i_1}
        \numeq
        \cdots
        \numeq
        \tenc{\expheight-1}{i_{\alpha}}
$
    for some $\alpha$.
    Since, by induction,
    $\goodnums{\expheight-1}{y}$
    is correct, then the program can only be satisfied if $T$ chose the same position in each row.
    That is
    $i_1 = \cdots = i_\alpha$.
    Thus, the vertical relation for the $i_1$th column has been verified.

    At this point $T$ can either write $\numsep$ and terminate or perform another pass (non-deterministically).
    In the latter case, it outputs $\numplus$, moves back to the beginning of the tape, and starts again.
    Thus, after a number of passes, $T$ will have written

\smallskip
\hspace{1.5cm} $
        \brac{
            \tenc{\expheight-1}{i^1_1}
            \numeq
            \cdots
            \numeq
            \tenc{\expheight-1}{i^1_{\alpha_1}}
        }
        \numplus
        \cdots
        \numplus
        \brac{
            \tenc{\expheight-1}{i^\beta_1}
            \numeq
            \cdots
            \numeq
            \tenc{\expheight-1}{i^\beta_{\alpha_\beta}}
        }
        \numsep
 $
 \smallskip

    for some $\beta$, $\alpha_1$, \ldots, $\alpha_\beta$.
    Since
    $\goodnums{\expheight-1}{y}$
    will only accept such sequences of the form
    \[
        \brac{
            \tenc{\expheight-1}{1}
            \brac{
                \numeq \tenc{\expheight-1}{1}
            }^\ast
        }
        \numplus
        \cdots
        \numplus
        \brac{
            \tenc{\expheight-1}{\nmax{\expheight-1}}
            \brac{
                \numeq \tenc{\expheight-1}{\nmax{\expheight-1}}
            }^\ast
        }
        \numsep
    \]
    we know that $T$ must check each vertical column in turn, from $1$ to
    $\nmax{\expheight-1}$.
\end{itemize}

Thus, at the end of all passes, if $T$ has not output $\bot$ it has verified that $x$ is a correct encoding of a solution to
$\tup{\tilesnum{\expheight},
      \hrelnum{\expheight},
      \vrelnum{\expheight},
      \inittilenum{\expheight},
      \fintilenum{\expheight}}$.
That is, together with
$\goodnums{\expheight-1}{y}$
we know that
    the word is of the correct format,
    each row has a tile for each index and these indices appear in order,
    the horizontal relation is respected, and
    the vertical tiling relation is respected.
If $x$ is not a correct encoding then $T$ will not be able to produce a $y$ that satisfies
$\goodnums{\expheight-1}{y}$.

\subsection{Reducing from a Tiling Problem}

Now that we are able to encode large numbers, we can encode an $\expheight$-$\expspace$-hard tiling problem as a satisfiability problem of $\strline[T]$ with two-way transducers.
In fact, most of the technical work has been done.

Thus, fix a tiling problem
$\tup{\tiles, \hrel, \vrel, \inittile, \fintile}$
that is $\expheight$-$\expspace$-hard.
In particular, we allow a corridor $\nmax{\expheight}$ tiles wide.
We use the program
\[
    S = \left\{
        \begin{array}{l}
            y = \ap{T}{x}; \\
            \goodnums{\expheight}{y}
        \end{array}
    \right.
\]
where $T$ is defined exactly as in the inductive case of
$\goodnums{\expheight}{y}$
except the tiling problem used is
$\tup{\tiles, \hrel, \vrel, \inittile, \fintile}$
rather than
$\tup{\tilesnum{\expheight},
      \hrelnum{\expheight},
      \vrelnum{\expheight},
      \inittilenum{\expheight},
      \fintilenum{\expheight}}$.

A satisfying tiling
\[
    \begin{array}{c}
        \tile^1_1 \ldots \tile^1_{\nmax{\expheight}} \\
        \cdots \\
        \tile^\tileheight_1 \ldots \tile^\tileheight_{\nmax{\expheight}}
    \end{array}
\]
can be encoded
\[
    \tenc{\expheight}{1} \tile^1_1
    \ldots
    \tenc{\expheight}{\nmax{\expheight}} \tile^1_{\nmax{\expheight}}
    \numplus
    \cdots
    \numplus
    \tenc{\expheight}{1} \tile^\expheight_1
    \ldots
    \tenc{\expheight}{\nmax{\expheight}} \tile^\expheight_{\nmax{\expheight}}
    \numsep
\]
which will satisfy $S$ in the same way as a correct input to
$\goodnums{\expheight}{y}$.
To see this, note that
\[
    \tenc{\expheight}{1} \tile^i_1
    \ldots
    \tenc{\expheight}{\nmax{\expheight}} \tile^i_{\nmax{\expheight}}
\]
acts like some
$\tenc{\expheight+1}{i}$.

In the opposite direction, assume some input satisfying $S$.
Arguing as in the encoding of large numbers, this input must be of the form
\[
    \brac{
        \brac{\tilerow_1 \brac{\numeq \tilerow_1}^\ast}
        \numplus
        \brac{\tilerow_2 \brac{\numeq \tilerow_2}^\ast}
        \numplus
        \cdots
        \numplus
        \brac{
            \tilerow_\tileheight
            \brac{\numeq \tilerow_\tileheight}^\ast
        }
        \numsep
    }^\ast
\]
where each $\tilerow_i$ is a row of a correct solution to the tiling problem.

Thus, with $\expheight+1$ transducers, we can encode a $\expheight$-$\expspace$-hard problem.


\section{Proof of Proposition~\ref{prop-fft-upper}}\label{app-fft-upper}

We start with a remark that, as mentioned the concatenation function $\concat$ can be encoded by the $\replaceall$ function (cf. \cite{CCHLW18} for details), so we shall not discuss  $\concat$ separately in the proof.

The following lemma is crucial for showing Proposition~\ref{prop-fft-upper}.  
\begin{lemma}\label{lem-prerec-comp}
Let $(\Aut, \conacc)$  be a conjunctive FA. Then for each string function  $f$ in \strlinefft{},
	there is an algorithm that runs in $(\ell_f(|f|, |(\Aut,\conacc)|))^{c_0}$ space (where $c_0$ is a constant) which enumerates
	each disjunct of a conjunctive representation of $\Pre_{R_f}((\Aut, \conacc))$, whose atom size is bounded by $\ell_f(|f|, |(\Aut,\conacc)|)$, where
%
\begin{itemize}
\item if $f$ is $\replaceall_e$ for a regular expression $e$, then $|f| = |e|$ and $\ell_f(i, j)= 2^{c_1 i^{c_2}} j$ for some constants $c_1,c_2$,
\item if $f$ is $\reverse$, then $|f| = 1$ and $\ell_f(i, j) = j$,
\item if $f$ is defined by an \FunFT{} $\Transducer$, then $|f| = |\Transducer|$ and $\ell_f(i, j)= i j$.
\end{itemize}
\end{lemma}


%
%
%


\begin{proof}
	Let $\Aut=(\controls, \transrel), \conacc)$ be a conjunctive \FA.
	
	For $f=\replaceall_e$, the result was shown in \cite{CCHLW18} .
	
	For the $\reverse$ function,  $\Pre_{R_\reverse}((\Aut, \conacc))$ is exactly the language defined by $(\Aut^\revsym, \conacc^\revsym)$, where 
	$\conacc^\revsym = \{(q_2, q_1) \mid (q_1, q_2) \in \conacc\}$.
	
	If $f$ is defined by an \FunFT{} $\Transducer=(\controls', q'_0, \finals', \transrel')$, then $\Pre_{R_T}(\Aut, \conacc)$ is conjunctively represented by $((\Aut'', \conacc_{q'}))_{q' \in \finals'}$, where
	\begin{itemize}
		\item $\Aut'' = (\controls'', \transrel'')$, with $\controls'' = \controls' \times \controls$, $\transrel''$ comprises the transitions $((q'_1, q_1), a, (q'_2, q_2))$ such that there exists $w \in \ialphabet^*$ satisfying that $(q'_1, a, w, q'_2) \in \transrel'$ and $q_1 \xrightarrow[\Aut]{w} q_2$,
		\item $\conacc_{q'} = \{((q'_0, q_1), (q', q_2)) \mid (q_1, q_2) \in \conacc\}$.
	\end{itemize}
	Note that the language defined by $((\Aut'', \conacc_{q'}))_{q' \in \finals'}$ is $\bigcup \limits_{q' \in \finals'} \Lang(\Aut'', \conacc_{q'})$. The size of each atom $(\Aut'', \conacc_{q'})$ is $|T||\Aut| = |f| |\Aut|$.
\end{proof}

\begin{proof}[Proof of Proposition~\ref{prop-fft-upper}]
Let $S$ be a program in \strlinefft{}. For technical convenience, we consider the \emph{dependency graph} of $S$, denoted by $G_S=(V_S, E_S)$, where $V_S$ is the set of string variables in $S$, and $E_S$ comprises the edges $(y, x_j)$ for each assignment $y := f(x_1, \ldots, x_r)$ in $S$.

Recall that the decision procedure in the proof of Theorem~\ref{th:gen} works by repeatedly removing the last assignment, say $y := f(\vec{x})$, and generating new assertions involving $\vec{x}$ from the assertions of $y$.
We adapt that decision procedure as follows:
\begin{itemize}
\item Replace \FA{}s with conjunctive \FA{}s and use the conjunctive representations of the pre-images of string operations.
\item Before removing the assignments, a preprocessing is carried out for $S$ as follows: For each assertion $\ASSERT{R(\vec{x})}$ with $\vec{x} = (x_1,\ldots, x_\arity)$ in $S$, nondeterministically guess a disjunct of the conjunctive representation of $R$, say $((\Aut_{1}, \conacc_{1}), \ldots, (\Aut_{\arity}, \conacc_{\arity}))$, and replace $\ASSERT{R(\vec{x})}$ with the sequence of assertions $\ASSERT{x_1 \in (\Aut_1, \conacc_1)}; \ldots; \ASSERT{x_\arity \in (\Aut_\arity, \conacc_\arity)}$. Note that after the preprocessing, each assertion is of the form $\ASSERT{y \in (\Aut, \conacc)}$ for a string variable $y$ and a conjunctive \FA{} $(\Aut, \conacc)$. Let $S_0$ be the resulting program  after preprocessing  $S$.

\item When removing each assignment $y:=f(\vec{x})$ with $\vec{x} = (x_1, \ldots, x_\arity)$, for each conjunctive \FA{} $(\Aut, \conacc)\in \sigma$ (where $\sigma$ is the collective constraints for $y$), nondeterministically guess one disjunct of the pre-image of $f$ under $(\Aut, \conacc)$ (NB.\ here we neither compute the product of the conjunctive \FA{}s from $\sigma$, nor compute an explicit representation of the pre-image), say $((\Aut_{1}, \conacc_{1}), \ldots, (\Aut_{\arity}, \conacc_{\arity}))$, and insert the sequence of assertions $\ASSERT{x_1 \in (\Aut_1, \conacc_1)}; \ldots; \ASSERT{x_\arity \in (\Aut_\arity, \conacc_\arity)}$ to the program.
\end{itemize}
We then show that the resulting (nondeterministic) decision procedure \emph{requires only exponential space}, which 
implies the \expspace{} upper-bound via Savitch's theorem.

Let $S'$ be the program obtained after removing $y:=f(\vec{x})$ and all assertions with conditions in $\rho$,   and $\sigma$ be the set of all conjunctive \FA{}s in $\rho$. Then for each $(\Aut, \conacc) \in \sigma$, a disjunct of the conjunctive representation of $\Pre_{R_f}(\Aut, \conacc)$, say $((\Aut_{1}, \conacc_{1}), \ldots, (\Aut_{\arity}, \conacc_{\arity}))$, is guessed, moreover, for each $j \in [\arity]$, an assertion $\ASSERT{x_j \in (\Aut_{j}, \conacc_{j})}$ is added. Let $S''$ be the resulting program. We say that the assertion $y \in (\Aut, \conacc)$ in $S'$ \emph{generates} the assertion $x_j \in (\Aut_{j}, \conacc_{j})$ in $S''$. One can easily extend this single-step generation relation to multiple steps by considering its transitive closure.

Let $S_1$ be the resulting program after all the assignments are removed. Namely, $S_1$ contains only assertions for input variables.
By induction on the number of removed assignments, we can show that for each input variable $y$ in $S$, each assertion $x \in (\Aut, \conacc)$ in $S_0$, and each path $\pi$ from $x$ to $y$ in $G_S$,  the assertion $x \in (\Aut, \conacc)$ generates \emph{exactly one} assertion $y \in (\Aut', \conacc')$ in $S_1$. Since for each pair of variables $(x,y)$ in $G_S$, there are at most exponentially many paths from $x$ to $y$,  $S_1$ contains at most exponentially many assertions for each input variable. Moreover, according to Lemma~\ref{lem-prerec-comp}, for each assertion $y \in (\Aut', \conacc')$ in $S_1$, suppose that $y \in (\Aut', \conacc')$ is generated by some assertion $x \in (\Aut, \conacc)$ in $S_0$, then $|(\Aut', \conacc')|$ is at most exponential in $|(\Aut, \conacc)|$. Therefore, we conclude that for each input variable $y$, $S_1$ contains at most exponentially many assertions for $y$, where each of them is of at most exponential size.
It follows that the product \FA{} of all the assertions for each input variable $y$ in $S_1$ is of doubly exponential size.

Since the last step of the decision procedure is to decide the nonemptiness of the intersection of all the assertions for each input variable $y$ and  nonemptiness of \FA{}s can be solved in nondeterministic logarithmic space, we deduce that the last step of the decision procedure can be done in nondeterministic exponential space.  We conclude that the aforementioned decision procedure is in nondeterministic exponential space.
\end{proof}



\section{Proof of Proposition~\ref{prop:expspace-lower}: Hardness of $\strline[\replaceall]$}
\label{sec:expspace-hardness-appendix}

We show that the path feasibility problem for $\strline[\replaceall]$ is $\expspace$-hard.
In fact, this holds even when only single characters are replaced.
To show the result, we will give a symbolic execution which ``transforms'' an automaton into exponentially many different automata.
These exponentially many automata will be used to check the vertical tiling relation for a tiling problem with an exponentially wide corridor.
Since such tiling problems are \expspace-hard, the result follows.

After introducing the tiling problem, we will show the form of the symbolic execution without specifying the regular languages used.
Starting from the end of the  symbolic execution, we will work backwards, showing the effect that eliminating the $\replaceall$ and string concatenations has on the regular languages.
We will then show a simple example before giving the precise languages needed to obtain our result.


The reduction builds a constraint that is satisfiable only if a given variable $x_0$ encodes a solution to a tiling problem.
The solution will be encoded in the following form, for some $\tileheight$.
Note, we use different symbols for each bit position of the binary encoding.
We use
\[
\resetchar
\rowdelim \nbit{1}{0} \ldots \nbit{n}{0} \tile^1_1
\nbit{1}{0} \ldots \nbit{n}{1} \tile^1_2
\ldots
\nbit{1}{1} \ldots \nbit{n}{1} \tile^1_\tilewidth
\rowdelim \ldots
\rowdelim \nbit{1}{0} \ldots \nbit{n}{0} \tile^\tileheight_1
\ldots
\nbit{1}{1} \ldots \nbit{n}{1} \tile^\tileheight_\tilewidth
\rowdelim
\resetchar
\]
where $\tile^i_j$ are tiles, preceded by a binary encoding of the column position of the tile.
The $\rowdelim$ character delimits each row and $\resetchar$ delimits the ends of the encoding.
It is a standard exercise to express that a string has the above form as the intersection of a polynomial number of polynomially sized automata.
It is also straight-forward to check the horizontal tiling relation using a polynomially sized automaton (each contiguous pair of tiles needs to appear in the horizontal relation $\hrel$).

The difficulty lies in asserting that the vertical tiling relation $\vrel$ is respected.
For a given position
$\bit_1 \ldots \bit_n$,
it is easy to construct an automaton that checks the tiles in the column labelled
$\bit_1 \ldots \bit_n$
respects $\vrel$.
The crux of our reduction is showing that, using only concatenation and replaceall, we can transform a single automaton constraint into an exponential number of checks, one for each position
$\bit_1 \ldots \bit_n$.
Recall, concatenation can be expressed using $\replaceall$.

This process will be explained in full shortly, but first 
let us give a simple example of how we can obtain the required checks. 

\begin{example}\label{sec:simple-expspace-example}
Let $n = 2$ and $V=\{(t_1, t_2), (t_2, t_1)\}$.

Consider the constraint (which we will explain in the sequel)
\[
\begin{array}{c}
y_1 := \replaceall_{\nbit{1}{0}}(x_0, \simplerepl{1}{0});
\ \ %
z_1 := \replaceall_{\nbit{1}{1}}(x_0, \simplerepl{1}{1});
\\
x_1 := y_1 \concat z_1;
\\
y_2 := \replaceall_{\nbit{2}{0}}(x_1, \simplerepl{2}{0});
\ \ %
z_2 := \replaceall_{\nbit{2}{1}}(x_1, \simplerepl{2}{1});
\\
x_2 := y_2 \concat z_2;\\
x_2 \in \Aut_\vrel
\end{array}
\]
where $\Aut_\vrel$ is the automaton below.
The purpose of this automaton will become clear later in the description.
\begin{figure*}[htbp]
	\includegraphics[scale=0.7]{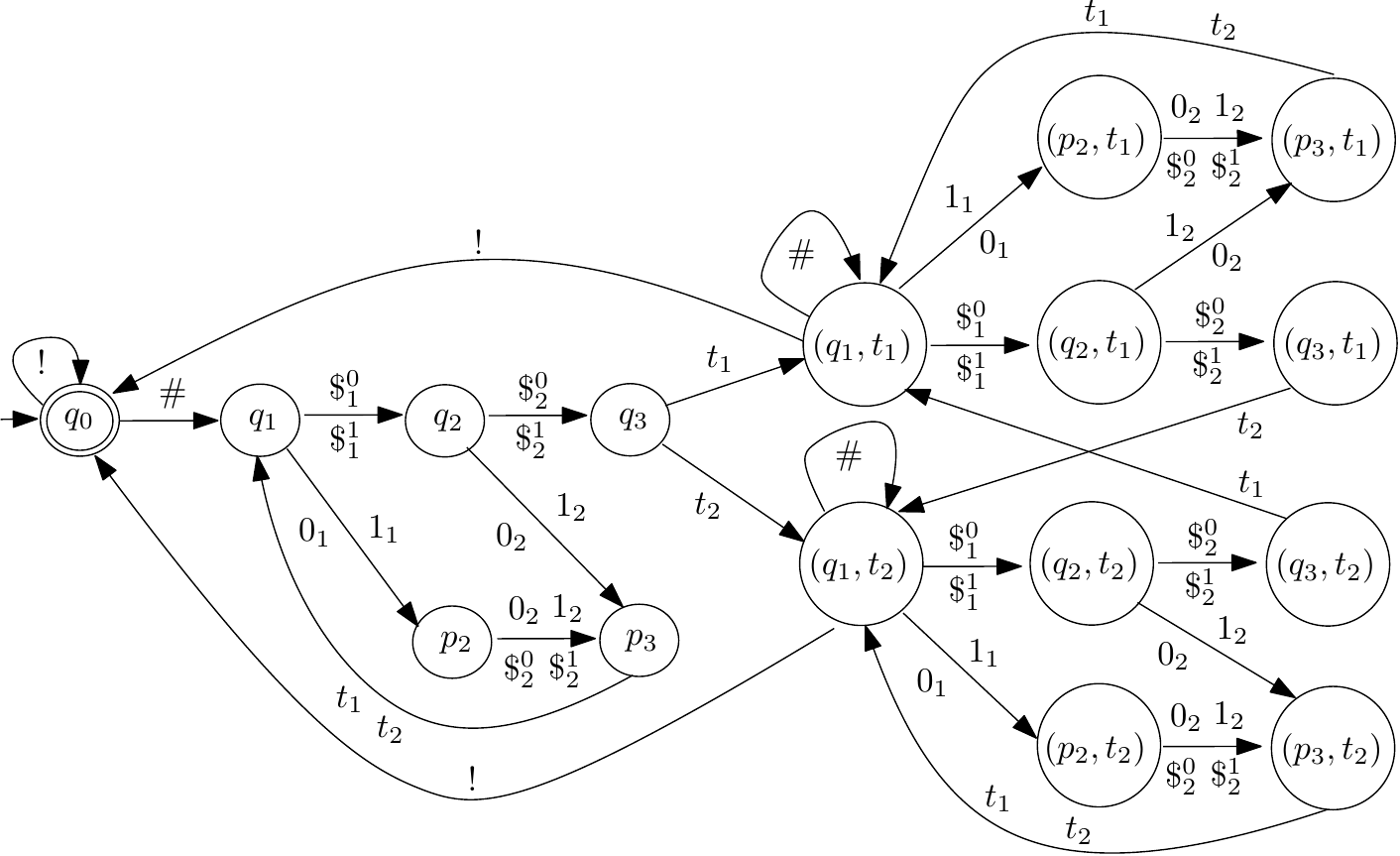}
\end{figure*}

Let
$x_0$ take the value
$$!\# \nbit{1}{0} \nbit{2}{0} t_1
\nbit{1}{0} \nbit{2}{1} t_1
\nbit{1}{1} \nbit{2}{0} t_2
\nbit{1}{1} \nbit{2}{1} t_2\ \#\ 
\nbit{1}{0} \nbit{2}{0} t_2
\nbit{1}{0} \nbit{2}{1} t_2
\nbit{1}{1} \nbit{2}{0} t_1
\nbit{1}{1} \nbit{2}{1} t_1\# !.
$$
That is, there are two rows in $x_0$, separated by $\#$, and in each row, $x_0$ counts from $00$ to $11$ in binary (marked with the bit indices $1,2$).
After the first pair of $\replaceall$ operations and concatenation of $y_1$ and $z_1$, the variable $x_1$ must have the value
%
$$
\begin{array}{l}
!\# \simplerepl{1}{0} \nbit{2}{0} t_1
\simplerepl{1}{0} \nbit{2}{1} t_1
\nbit{1}{1} \nbit{2}{0} t_2
\nbit{1}{1} \nbit{2}{1} t_2\ \#\ 
\simplerepl{1}{0} \nbit{2}{0} t_2
\simplerepl{1}{0} \nbit{2}{1} t_2
\nbit{1}{1} \nbit{2}{0} t_1
\nbit{1}{1} \nbit{2}{1} t_1\# ! \\
!\# \nbit{1}{0} \nbit{2}{0} t_1
\nbit{1}{0} \nbit{2}{1} t_1
\simplerepl{1}{1} \nbit{2}{0} t_2
\simplerepl{1}{1} \nbit{2}{1} t_2\ \#\ 
\nbit{1}{0} \nbit{2}{0} t_2
\nbit{1}{0} \nbit{2}{1} t_2
\simplerepl{1}{1} \nbit{2}{0} t_1
\simplerepl{1}{1} \nbit{2}{1} t_1\# !
\end{array}
$$
After the next $\replaceall$ operations and concatenation of $y_2$ and $z_2$, the variable $x_2$ has the value
$$
\begin{array}{l}
!\# \underline{\simplerepl{1}{0} \simplerepl{2}{0} t_1}
\simplerepl{1}{0} \nbit{2}{1} t_1
\nbit{1}{1} \simplerepl{2}{0} t_2
\nbit{1}{1} \nbit{2}{1} t_2\ \#\ 
\underline{\simplerepl{1}{0} \simplerepl{2}{0} t_2}
\simplerepl{1}{0} \nbit{2}{1} t_2
\nbit{1}{1} \simplerepl{2}{0} t_1
\nbit{1}{1} \nbit{2}{1} t_1\# ! \\
!\# \nbit{1}{0} \simplerepl{2}{0} t_1
\nbit{1}{0} \nbit{2}{1} t_1
\underline{\simplerepl{1}{1} \simplerepl{2}{0} t_2}
\simplerepl{1}{1} \nbit{2}{1} t_2\ \#\ 
\nbit{1}{0} \simplerepl{2}{0} t_2
\nbit{1}{0} \nbit{2}{1} t_2
\underline{\simplerepl{1}{1} \simplerepl{2}{0} t_1}
\simplerepl{1}{1} \nbit{2}{1} t_1\# !\\
!\# \simplerepl{1}{0} \nbit{2}{0} t_1
\underline{\simplerepl{1}{0} \simplerepl{2}{1} t_1}
\nbit{1}{1} \nbit{2}{0} t_2
\nbit{1}{1} \simplerepl{2}{1} t_2\ \#\ 
\simplerepl{1}{0} \nbit{2}{0} t_2
\underline{\simplerepl{1}{0} \simplerepl{2}{1} t_2}
\nbit{1}{1} \nbit{2}{0} t_1
\nbit{1}{1} \simplerepl{2}{1} t_1\# ! \\
!\# \nbit{1}{0} \nbit{2}{0} t_1
\nbit{1}{0} \simplerepl{2}{1} t_1
\simplerepl{1}{1} \nbit{2}{0} t_2
\underline{\simplerepl{1}{1} \simplerepl{2}{1} t_2}\ \#\ 
\nbit{1}{0} \nbit{2}{0} t_2
\nbit{1}{0} \simplerepl{2}{1} t_2
\simplerepl{1}{1} \nbit{2}{0} t_1
\underline{\simplerepl{1}{1} \simplerepl{2}{1} t_1} \# !
\end{array}
$$
Notice that this value has four (altered) copies of the original value of $x_0$ and this value enjoys the property that each copy of $x_0$ contains the occurrences of exactly one of $\simplerepl{1}{0} \simplerepl{2}{0}$, $\simplerepl{1}{1} \simplerepl{2}{0}$, $\simplerepl{1}{0} \simplerepl{2}{1}$, $\simplerepl{1}{1} \simplerepl{2}{1}$.
In each copy of $x_0$, we have underlined the positions where the vertical tiling relation is checked.
In particular, in the first copy, the  vertical tiling relation is checked for the $\nbit{1}{0} \nbit{2}{0}$ position of $x_0$, witnessed by the run of $\Aut_\vrel$ on the first copy sketched below, 
\[
\begin{array}{l}
q_0 \xrightarrow{!} q_0 \xrightarrow{\#} q_1 \xrightarrow{\simplerepl{1}{0}} q_2 \xrightarrow {\simplerepl{2}{0} } q_3 \xrightarrow{t_1} (q_1, t_1) \ldots (p_3, t_1) \xrightarrow{t_2} (q_1, t_1) \xrightarrow {\#} (q_1, t_1)  \\ 
\xrightarrow{\simplerepl{1}{0}} (q_2, t_1) \xrightarrow{\simplerepl{2}{0}} (q_3, t_1) \xrightarrow{t_2} (q_1, t_2) \ldots (q_1, t_2) \xrightarrow{\#} (q_1, t_2) \xrightarrow{!} q_0.
\end{array}
\]
In the next copies, the vertical tiling relation is checked for the positions
$\nbit{1}{1} \nbit{2}{0}$,
$\nbit{1}{0} \nbit{2}{1}$, and
$\nbit{1}{1} \nbit{2}{1}$
in $x_0$ respectively. Note that $\Aut_\vrel$ \emph{does not} have to check that in each copy, the same position in each row is marked (which would need a state space exponential in $n$).
In particular, the occurrences of the subwords $\$^{b_1}_1\$^{b_2}_2$ ($b_1,b_2 \in \{0,1\}$)  are equal in each copy since the value of $x_2$ resulting from the assignments already guarantees this. \qed
\end{example}


We now give the formal proof. Fix a tiling problem with
    tiles $\tiles$,
    initial and final tiles $\inittile$ and $\fintile$ respectively,
    horizontal and vertical tiling relations $\hrel$ and $\vrel$,
    and a width $\tilewidth = 2^n$ for some $n$.
We will create a symbolic execution in $\strline[\replaceall]$ which is path feasible iff the tiling problem has a solution.

In particular, we will require that a certain variable $x_1$ can take a value
\[
    \resetchar
    \rowdelim \nbit{1}{0} \ldots \nbit{n}{0} \tile^1_1
              \nbit{1}{0} \ldots \nbit{n}{1} \tile^1_2
              \ldots
              \nbit{1}{1} \ldots \nbit{n}{1} \tile^1_\tilewidth
    \rowdelim \ldots
    \rowdelim \nbit{1}{0} \ldots \nbit{n}{0} \tile^\tileheight_1
              \ldots
              \nbit{1}{1} \ldots \nbit{n}{1} \tile^\tileheight_\tilewidth
    \rowdelim
    \resetchar
\]
where
$\resetchar, \rowdelim,
 \nbit{1}{0}, \nbit{1}{1},
 \ldots,
 \nbit{n}{0}, \nbit{n}{1} \notin \tiles$.
In particular, the string encodes a solution to the tiling problem, where each tile is preceded by a binary number of length $n$ representing the number of the column in which it appears.
Note, we have a different character for each bit position, e.g. $\nbit{2}{0}$ is a distinct character from $\nbit{3}{0}$.
The $\resetchar$ will be used to mark the beginning and end of the string and $\rowdelim$ is used to separate the rows of the solution.

To reach such a solution, we will use the symbolic execution to effectively generate an exponential number of languages.
We can number these languages in binary, i.e.\ $L_{0\ldots00}$ to $L_{1\ldots11}$.
In addition to $\resetchar$, $\rowdelim$, $\nbit{i}{0}$, and $\nbit{i}{1}$, we will also introduce characters of the form
$\repl{i}{\bit}{\bit'}$.
In language
$L_{\bit_1\ldots\bit_n}$,
the character
$\repl{i}{\bit}{\bit'}$
will mean that if $\bit_i = \bit$, then match the character
$\bit' \in \set{\nbit{i}{0},\nbit{i}{1}}$.
This will become clearer when we show an example, but first we will give the symbolic execution needed to generate the languages.

\subsection*{The symbolic execution}

\[
    \begin{array}{rcl}
        \varphi &=& \text{\ASSERT{$x_1 \in L$}};\ \text{\ASSERT{$x_1 \in L_\hrel$}}; \\
                & & \text{\ASSERT{$x_1 \in L_1$}};
                    \ldots;
                    \text{\ASSERT{$x_1 \in L_n$}}; \\
                \\
                & & y^0_2 := \replaceall_{\nbit{1}{0}}(x_1, \repl{1}{0}{0}); \\
                & & y^1_2 := \replaceall_{\nbit{1}{1}}(y^0_2, \repl{1}{0}{1}); \\
                & & z^0_2 := \replaceall_{\nbit{1}{0}}(x_1, \repl{1}{1}{0}); \\
                & & z^1_2 := \replaceall_{\nbit{1}{1}}(z^0_2, \repl{1}{1}{1}); \\
                & & x_2 := y^1_2 \concat z^1_2; \\
                \\
                & & y^0_3 := \replaceall_{\nbit{2}{0}}(x_2, \repl{2}{0}{0}); \\
                & & y^1_3 := \replaceall_{\nbit{2}{1}}(y^0_3, \repl{2}{0}{1}); \\
                & & z^0_3 := \replaceall_{\nbit{2}{0}}(x_2, \repl{2}{1}{0}); \\
                & & z^1_3 := \replaceall_{\nbit{2}{1}}(z^0_2, \repl{2}{1}{1}); \\
                & & x_3 := y^1_3 \concat z^1_3; \\
 		\\
                & & \ldots \\
                & & y^0_n := \replaceall_{\nbit{n}{0}}(x_{n-1}, \repl{n}{0}{0}); \\
                & & y^1_n := \replaceall_{\nbit{n}{1}}(y^0_n, \repl{n}{0}{1}); \\
                & & z^0_n := \replaceall_{\nbit{n}{0}}(x_{n-1}, \repl{n}{1}{0}); \\
                & & z^1_n := \replaceall_{\nbit{n}{1}}(z^0_n, \repl{n}{1}{1}); \\
                & & x_n := y^1_n \concat z^1_n; \\
                \\
                & & \text{\ASSERT{$x_n \in L_\vrel$}}
    \end{array}
\]

The symbolic execution we will need to encode solutions to the tiling problem is put above.
The regular constraints on $x_1$ will enforce that any solution also gives a solution to the tiling problem if we do not enforce the vertical matching relation.
To enforce the vertical matching relation we use the constraint on $x_n$ at the end.
Notice that each sequence of uses of $\replaceall$ translates between
$\repl{i}{\bit}{\bit'}$
and the values $\nbit{i}{0}$ and $\nbit{i}{1}$, where the $y$ variables handle the case where $\bit$ is $0$, and the $z$ when $\bit$ is $1$.
This process will be illuminated later with an example.
Note, $\concat$ can be expressed with $\replaceall$ and hence is not directly needed.

\subsection{Unravelling the Constraint}

We show how $\varphi$ can lead to $x_1$ having to be included in an exponential number of languages.
To do this, we eliminate each use of $\replaceall$ and $\concat$ from the bottom-up.
The process is illustrated in the forwards direction by the example in Section~\ref{sec:simple-expspace-example}.
However, note, in Section~\ref{sec:simple-expspace-example} we did not use the characters
$\repl{i}{\bit}{\bit'}$
but a simpler version
$\simplerepl{i}{\bit}$.
After our explanation, we give a similar example for the backwards direction.

The first step is to eliminate
$x_n := y^1_n \concat z^1_n$.
This is done by removing
$x_n := y^1_n \concat z^1_n$.
and replacing it with
$\text{\ASSERT{$y^1_n \in L'_{0}$}};
 \text{\ASSERT{$z^1_n \in L'_{1}$}}$
where $L_\vrel = L'_{0} \concat L'_{1}$.
This can be done by taking an automaton $\cA$ such that $L_\vrel = \lang{\cA}$ and guessing the state at the split between $y^1_n$ and $z^1_n$ in an accepting run over the value of $x_n$.
This means that $L'_{0}$ and $L'_{1}$ can be represented by automata with the same states and transitions, but different initial and final states.

Next, we eliminate the $\replaceall$ functions.
This leads to
$\text{\ASSERT{$x_{n-1} \in L_{0}$}}$,
where $L_{0}$ is $L'_{0}$ except all $\repl{n}{0}{\bit}$ characters have been replaced by $\nbit{n}{\bit}$.
Similarly, we also have
$\text{\ASSERT{$x_{n-1} \in L_{1}$}}$,
where $L_{1}$ is $L'_{1}$ except all $\repl{n}{1}{\bit}$ characters have been replaced by $\nbit{n}{\bit}$.
Note that different $\repl{n}{\bit'}{\bit}$ characters have been replaced in $L_{0}$ and $L_{1}$.
The languages have begun to diverge.

We then eliminate
$x_{n-1} := y^1_{n-1} \concat z^1_{n-1}$.
Thus $L_{0}$ needs to be split into $L'_{00}$ and $L'_{10}$.
Similarly, $L_{1}$ needs to be split into $L'_{01}$ and $L'_{11}$.
This results in the constraints
\[
    \text{\ASSERT{$y^1_{n-1} \in L'_{00}$}};
    \text{\ASSERT{$y^1_{n-1} \in L'_{01}$}};
    \text{\ASSERT{$z^1_{n-1} \in L'_{10}$}};
    \text{\ASSERT{$z^1_{n-1} \in L'_{11}$}}
\]
and after eliminating the next batch of $\replaceall$ functions we have
\[
    \text{\ASSERT{$y^1_{n-1} \in L_{00}$}};
    \text{\ASSERT{$y^1_{n-1} \in L_{01}$}};
    \text{\ASSERT{$z^1_{n-1} \in L_{10}$}};
    \text{\ASSERT{$z^1_{n-1} \in L_{11}$}} \ .
\]

By following this procedure, we eventually obtain the constraints
\[
    \text{\ASSERT{$x_1 \in L_{0\ldots00}$}};
    \text{\ASSERT{$x_1 \in L_{0\ldots01}$}};
    \ldots;
    \text{\ASSERT{$x_1 \in L_{1\ldots11}$}} \ .
\]
Notice, furthermore, that
$L_{\bit_1 \ldots \bit_n}$
has all
$\repl{i}{\bit_i}{0}$
replaced by $\nbit{i}{0}$
and all
$\repl{i}{\bit_i}{1}$
replaced by $\nbit{i}{1}$
but all characters
$\repl{i}{\bit'_i}{\bit}$
where
$\bit'_i \neq \bit_i$
are unchanged.

\subsection{Controlling the Unravelling}

We revisit the example from Section~\ref{sec:simple-expspace-example} for the case of $n = 3$ to show how we can use the unravelling to obtain automata which will be useful to our encoding.
In this example, we will ignore the issue of initial and final states, and just show the effect on the automaton transition relation.
Let $n = 3$ and $L_\vrel$ be defined by the automaton below.
\begin{center}
\begin{tikzpicture}[node distance=2cm,on grid,auto]
   \node[state] (q_0)   {$q_0$};
   \node[state] (q_1) [right=of q_0] {$q_1$};
   \node[state] (q_2) [right=of q_1] {$q_2$};
   \node[state] (q_3) [right=of q_2] {$q_3$};
    \path[->]
    (q_0) edge [bend left]  node [above] {$\repl{1}{0}{0}$} (q_1)
          edge [bend right] node [below] {$\repl{1}{1}{1}$} (q_1)
    (q_1) edge [bend left]  node [above] {$\repl{2}{0}{0}$} (q_2)
          edge [bend right] node [below] {$\repl{2}{1}{1}$} (q_2)
    (q_2) edge [bend left]  node [above] {$\repl{3}{0}{0}$} (q_3)
          edge [bend right] node [below] {$\repl{3}{1}{1}$} (q_3);
\end{tikzpicture}
\end{center}

We initially have
$\text{\ASSERT{$x_3 \in L_\vrel$}}$.
We first eliminate
$x_3 := y^1_3 \concat z^1_3$ and then
\[
    \begin{array}{rcl}
        y^0_3 &:=& \replaceall_{\nbit{3}{0}}(x_{2}, \repl{3}{0}{0}); \\
        y^1_3 &:=& \replaceall_{\nbit{3}{1}}(y^0_3, \repl{3}{0}{1}); \\
        z^0_3 &:=& \replaceall_{\nbit{3}{0}}(x_{2}, \repl{3}{1}{0}); \\
        z^1_3 &:=& \replaceall_{\nbit{3}{1}}(z^0_3, \repl{3}{1}{1}) \ .
    \end{array}
\]
Note, from $L'_{0}$ we replace the $\repl{3}{0}{0}$- and $\repl{3}{0}{1}$-transitions (although the latter does not appear in $\cA$), while from $L'_{1}$ we replace the $\repl{3}{1}{0}$- and $\repl{3}{1}{1}$-transitions.
This leaves us with
\[
    \text{\ASSERT{$x_{2} \in L_{0}$}};
    \text{\ASSERT{$\land x_{2} \in L_{1}$}}
\]
where $L_{0}$ is given by
\begin{center}
\begin{tikzpicture}[node distance=2cm,on grid,auto]
   \node[state] (q_0)   {$q_0$};
   \node[state] (q_1) [right=of q_0] {$q_1$};
   \node[state] (q_2) [right=of q_1] {$q_2$};
   \node[state] (q_3) [right=of q_2] {$q_3$};
    \path[->]
    (q_0) edge [bend left]  node [above] {$\repl{1}{0}{0}$} (q_1)
          edge [bend right] node [below] {$\repl{1}{1}{1}$} (q_1)
    (q_1) edge [bend left]  node [above] {$\repl{2}{0}{0}$} (q_2)
          edge [bend right] node [below] {$\repl{2}{1}{1}$} (q_2)
    (q_2) edge [bend left]  node [above] {$\nbit{3}{0}$} (q_3)
          edge [bend right] node [below] {$\repl{3}{1}{1}$} (q_3);
\end{tikzpicture}
\end{center}
and $L_{1}$ is given by
\begin{center}
\begin{tikzpicture}[node distance=2cm,on grid,auto]
   \node[state] (q_0)   {$q_0$};
   \node[state] (q_1) [right=of q_0] {$q_1$};
   \node[state] (q_2) [right=of q_1] {$q_2$};
   \node[state] (q_3) [right=of q_2] {$q_3$};
    \path[->]
    (q_0) edge [bend left]  node [above] {$\repl{1}{0}{0}$} (q_1)
          edge [bend right] node [below] {$\repl{1}{1}{1}$} (q_1)
    (q_1) edge [bend left]  node [above] {$\repl{2}{0}{0}$} (q_2)
          edge [bend right] node [below] {$\repl{2}{1}{1}$} (q_2)
    (q_2) edge [bend left]  node [above] {$\repl{3}{0}{0}$} (q_3)
          edge [bend right] node [below] {$\nbit{3}{1}$} (q_3);
\end{tikzpicture}
\end{center}

After completing the elimination process, we are left with
\[
    \text{\ASSERT{$x_1 \in L_{000}$}};
    \ldots;
    \text{\ASSERT{$\land x_1 \in L_{111}$}}
\]
where, for example, $L_{010}$ is given by
\begin{center}
\begin{tikzpicture}[node distance=2cm,on grid,auto]
   \node[state] (q_0)   {$q_0$};
   \node[state] (q_1) [right=of q_0] {$q_1$};
   \node[state] (q_2) [right=of q_1] {$q_2$};
   \node[state] (q_3) [right=of q_2] {$q_3$};
    \path[->]
    (q_0) edge [bend left]  node [above] {$\nbit{1}{0}$} (q_1)
          edge [bend right] node [below] {$\repl{1}{1}{1}$} (q_1)
    (q_1) edge [bend left]  node [above] {$\repl{2}{0}{0}$} (q_2)
          edge [bend right] node [below] {$\nbit{2}{1}$} (q_2)
    (q_2) edge [bend left]  node [above] {$\nbit{3}{0}$} (q_3)
          edge [bend right] node [below] {$\repl{3}{1}{1}$} (q_3);
\end{tikzpicture}
\end{center}
If we further insist that $x_1$ contains only characters from
$\set{\nbit{1}{0},
      \nbit{1}{1},
      \nbit{2}{1},
      \nbit{2}{1},
      \nbit{3}{1},
      \nbit{3}{1}}$
then the only path from $q_0$ to $q_3$ is via the sequence
$\nbit{1}{0} \nbit{2}{1} \nbit{3}{0}$.
Likewise, in $L_{111}$ the only valid sequence will be
$\nbit{1}{1} \nbit{2}{1} \nbit{3}{1}$.

\subsection{Completing the Reduction}

To finish the reduction, we have to instantiate $\varphi$ by giving
$L$, $L_\hrel$, $L_1$, \ldots, $L_n$ and $L_\vrel$ where all languages but $L_\vrel$ are easily seen to be representable by automata whose size is polynomial in $n$.
\begin{enumerate}
\item
    $L$ will enforce that $x_1$ matches
    \[
        \resetchar
        \brac{
            \rowdelim
            \brac{
                \set{\nbit{1}{0},\nbit{1}{1}}
                \ldots
                \set{\nbit{n}{0},\nbit{n}{1}}
                \tiles
            }^\ast
        }^\ast
        \rowdelim
        \resetchar \ .
    \]
    That is, a sequence of delimited rows, each consisting of a sequence of alternating $n$-bit sequences and tiles;
    the beginning and end of the word is marked by $\resetchar$.
    Moreover, $L$ will insist that the first tile seen is $\inittile$ and that the final tile seen is $\fintile$, as required by the tiling problem.

\item
    $L_\hrel$ will enforce that any subsequence
    \[
        \tile \bit_1 \ldots \bit_n \tile'
    \]
    where for all $i$ we have
    $\bit_i \in \set{\nbit{i}{0}, \nbit{i}{1}}$
    in the value of $x_1$ is such that $\tup{\tile, \tile'} \in \hrel$.

\item
    The languages $L_1, \ldots L_n$ will together enforce the correct sequencing of bit values $\bit_1 \ldots \bit_n$.
    That is, between each $\rowdelim$ each sequence $\bit_1 \ldots \bit_n$ appears exactly once and in the correct order
    (i.e.\ %
    $\nbit{1}{0}\ldots\nbit{n-1}{0}\nbit{n}{0}$
    appears before
    $\nbit{1}{0}\ldots\nbit{n-1}{0}\nbit{n}{1}$
    and so on up to
    $\nbit{1}{1}\ldots\nbit{n-1}{1}\nbit{n}{1}$).
    To do this, each $L_i$ will check the following.
    \begin{enumerate}
    \item
        After each $\rowdelim$ the first instance of
        $\set{\nbit{i}{0},\nbit{i}{1}}$
        is $\nbit{i}{0}$.
    \item
        Before each $\rowdelim$ the last instance of
        $\set{\nbit{i}{0},\nbit{i}{1}}$
        is $\nbit{i}{1}$.
    \item
        For every $\nbit{i}{0}$ such that
        $\rowdelim$ does not appear before the next occurrence of
        $\bit_i \in \set{\nbit{i}{0},\nbit{i}{1}}$,
        if the immediately succeeding characters are
        $\nbit{i+1}{1}\ldots\nbit{n}{1}$
        then $\bit_i$ is $\nbit{i}{1}$,
        otherwise $\bit_i$ is $\nbit{i}{0}$.
    \item
        For every $\nbit{i}{1}$ such that
        $\rowdelim$ does not appear before the next occurrence of
        $\bit_i \in \set{\nbit{i}{0},\nbit{i}{1}}$,
        if the immediately succeeding characters are
        $\nbit{i+1}{1}\ldots\nbit{n}{1}$
        then $\bit_i$ is $\nbit{i}{0}$,
        otherwise $\bit_i$ is $\nbit{i}{1}$.
    \end{enumerate}
\end{enumerate}

Finally, we need to define $L_\vrel$.
As seen above, $L_\vrel$ will lead to an exponential number of languages
$L_{\bit_1\ldots\bit_n}$
inside which $x_1$ will need to be contained.
The role of
$L_{\bit_1\ldots\bit_n}$
will be to check that the
$\bit_1\ldots\bit_n$th
column of the tiling obeys the vertical matching relation.
Each
$L_{\bit_1\ldots\bit_n}$
can be seen to be representable by a polynomially sized automaton that stores in its states the last tile seen after the sequence
$\bit_1\ldots\bit_n$.
Then, when the sequence next occurs, the new tile can be compared with the previous one.
The automaton will proceed by storing the new tile and forgetting the old.

We will design an automaton for $L_\vrel$ which will lead to the generation of the correct automaton for each
$L_{\bit_1\ldots\bit_n}$.
We will use characters
$\repl{i}{\bit}{\bit'}$
as before to generate the required bit sequences in the transition labelling of the automaton representing
$L_{\bit_1\ldots\bit_n}$.

\begin{center}
    \begin{tikzpicture}[node distance=3cm and 2cm,on grid,auto,bend angle=75,shorten >= 1pt]
       \node[state] (q1)   {$\tup{q_1, \tile}$};
       \node[state] (q2) [right=of q1] {$\tup{q_2, \tile}$};
       \node[state] (q3) [right=of q2] {$\tup{q_3, \tile}$};
       \node        (qdots) [right=of q3] {$\cdots$};
       \node[state] (qn) [right=of qdots] {$\tup{q_{n+1}, \tile}$};
       \node[state] (qt1) [above right=of qn,yshift=-1.5cm] {$\tup{q_1, \tile_1}$};
       \node        (qtdots) [right=of qn] {$\vdots$};
       \node[state] (qtm) [below right=of qn,yshift=1.5cm] {$\tup{q_1, \tile_m}$};
       \node[state] (p2) [below=of q2] {$\tup{p_2, \tile}$};
       \node[state] (p3) [below=of q3] {$\tup{p_3, \tile}$};
       \node        (pdots) [below=of qdots] {$\cdots$};
       \node[state] (pn) [below=of qn] {$\tup{p_{n+1}, \tile}$};
       \path[->]
         (q1) edge node [above] {$\repl{1}{0}{0}$} (q2)
         (q1) edge node [below] {$\repl{1}{1}{1}$} (q2)
         (q1) edge node [above] {$\repl{1}{0}{1}$} (p2)
         (q1) edge node [below] {$\repl{1}{1}{0}$} (p2)
         (q2) edge node [above] {$\repl{2}{0}{0}$} (q3)
         (q2) edge node [below] {$\repl{2}{1}{1}$} (q3)
         (q2) edge node [above] {$\repl{2}{0}{1}$} (p3)
         (q2) edge node [below] {$\repl{2}{1}{0}$} (p3)
         (q3) edge node [above] {$\repl{3}{0}{0}$} (qdots)
         (q3) edge node [below] {$\repl{3}{1}{1}$} (qdots)
         (q3) edge node [above] {$\repl{3}{0}{1}$} (pdots)
         (q3) edge node [below] {$\repl{3}{1}{0}$} (pdots)
         (qdots) edge node [above] {$\repl{n}{0}{0}$} (qn)
         (qdots) edge node [below] {$\repl{n}{1}{1}$} (qn)
         (qdots) edge node [above] {$\repl{n}{0}{1}$} (pn)
         (qdots) edge node [below] {$\repl{n}{1}{0}$} (pn)
         (p2) edge node [above] {$\replall{2}$} (p3)
         (p3) edge node [above] {$\replall{3}$} (pdots)
         (pdots) edge node [above] {$\replall{n}$} (pn)
         (pn) edge [bend left] node [below] {$\tiles$} (q1)
         (qn) edge node [above] {$\tile_1$} (qt1)
         (qn) edge node [above] {$\tile_m$} (qtm);
    \end{tikzpicture}
    \vspace*{-1cm}
\end{center}

In addition, we also need to ensure that the elimination of the concatenations -- which leads to a non-deterministic change in the initial and final states of the automata -- does not disrupt the language accepted.
For this we will use the $\resetchar$ character, which marks the beginning and end of the value of $x_1$.
The automaton for $L_\vrel$ will treat $\resetchar$ as a kind of ``reset'' character, which takes the automaton back to a defined initial state, which will also be the accepting state.

Before giving the formal definition, we will give an extract of the automaton that will check for the sequence
$\bit_1\ldots\bit_n$
and check the tiling relation.
Each state is of the form $\tup{q, \tile}$ where $\tile$ is the previously saved tile.
In the diagram, $\replall{i}$ denotes the set of all characters
$\repl{i}{\bit}{\bit'}$
for all $\bit, \bit' \in \set{0,1}$.
I.e.\ the transition can read any of these characters.
The top row of the automaton shows the run whilst the correct sequence is being read (and the tile needs to be checked), whilst the bottom is for all bit sequences that diverge from
$\bit_1\ldots\bit_n$
(another column is being read).
From $\tup{q_n, \tile}$ the tiles
$t_1, \ldots, t_m$
are all tiles such that
$\tup{t, t_j} \in \vrel$
for all $j$.
The transitions from each
$\tup{q_1, \tile_j}$
will be analogous to those from
$\tup{q_1, \tile}$.

To see the pattern of the automaton extract, we show how it would be instantiated for $L_{0\ldots00}$ with all edges labelled by
$\repl{i}{\bit}{\bit'}$
removed.

\begin{center}
\small
    \begin{tikzpicture}[node distance=3cm and 2cm,on grid,auto,bend angle=75,shorten >= 1pt]
       \node[state] (q1)   {$\tup{q_1, \tile}$};
       \node[state] (q2) [right=of q1] {$\tup{q_2, \tile}$};
       \node[state] (q3) [right=of q2] {$\tup{q_3, \tile}$};
       \node        (qdots) [right=of q3] {$\cdots$};
       \node[state] (qn) [right=of qdots] {$\tup{q_{n+1}, \tile}$};
       \node[state] (qt1) [above right=of qn,yshift=-1.5cm] {$\tup{q_1, \tile_1}$};
       \node        (qtdots) [right=of qn] {$\vdots$};
       \node[state] (qtm) [below right=of qn,yshift=1.5cm] {$\tup{q_1, \tile_m}$};
       \node[state] (p2) [below=of q2] {$\tup{p_2, \tile}$};
       \node[state] (p3) [below=of q3] {$\tup{p_3, \tile}$};
       \node        (pdots) [below=of qdots] {$\cdots$};
       \node[state] (pn) [below=of qn] {$\tup{p_{n+1}, \tile}$};
       \path[->]
         (q1) edge node [above] {$\nbit{1}{0}$} (q2)
         (q1) edge node [below] {$\nbit{1}{1}$} (p2)
         (q2) edge node [above] {$\nbit{2}{0}$} (q3)
         (q2) edge node [above] {$\nbit{2}{1}$} (p3)
         (q3) edge node [above] {$\nbit{3}{0}$} (qdots)
         (q3) edge node [above] {$\nbit{3}{1}$} (pdots)
         (qdots) edge node [above] {$\nbit{n}{0}$} (qn)
         (qdots) edge node [above] {$\nbit{n}{1}$} (pn)
         (p2) edge node [above] {$\nbit{2}{0}$} (p3)
         (p2) edge node [below] {$\nbit{2}{1}$} (p3)
         (p3) edge node [above] {$\nbit{3}{0}$} (pdots)
         (p3) edge node [below] {$\nbit{3}{1}$} (pdots)
         (pdots) edge node [above] {$\nbit{n}{0}$} (pn)
         (pdots) edge node [below] {$\nbit{n}{1}$} (pn)
         (pn) edge [bend left] node [below] {$\tiles$} (q1)
         (qn) edge node [above] {$\tile_1$} (qt1)
         (qn) edge node [above] {$\tile_m$} (qtm);
    \end{tikzpicture}
\end{center}

Hence, we are ready to define an automaton $\cA_\vrel$ giving the language $L_\vrel$.
We give the definition first, and explanation below.
We define
\[
    \cA_\vrel = \tup{Q, \delta, q_0, F}
\]
where
\[
    \begin{array}{rcl}
        Q &=& \set{q_0, q_1, \ldots, q_n, q_{n+1}, p_2, \ldots, p_n, p_{n+1}} \ \cup \\
          & & \setcomp{\tup{q_i, \tile}}
                      {1 \leq i \leq n+1 \land \tile \in \tiles} \ \cup \\
          & & \setcomp{\tup{p_i, \tile}}
                      {2 \leq i \leq n+1 \land \tile \in \tiles} 
\end{array}
\]
\[
\begin{array}{rcl}
        \delta &=& \setcomp{q \xrightarrow{\resetchar} q_0}{q \in Q} \ \cup \\
               & & \set{q_0 \xrightarrow{\rowdelim} q_1}\ \cup \\
               & & \setcomp{\tup{q_1, \tile}
                            \xrightarrow{\rowdelim}
                            \tup{q_1, \tile}}
                           {\tile \in \tiles} \ \cup \\
               & & \setcomp{q_i
                            \xrightarrow{\repl{i}{\bit}{\bit}}
                            q_{i+1}}
                           {1 \leq i \le n \land \bit \in \set{0,1}} \ \cup \\
               & & \setcomp{\tup{q_i, \tile}
                            \xrightarrow{\repl{i}{\bit}{\bit}}
                            \tup{q_{i+1}, \tile}}
                           {1 \leq i \le n \land \bit \in \set{0,1}} \ \cup
\end{array}
\]
\[
\begin{array}{r c l}
               & & \setcomp{q_i
                            \xrightarrow{\repl{i}{\bit}{\bit'}}
                            p_{i+1}}
                           {1 \leq i \le n \land
                            \bit \neq \bit' \in \set{0,1}} \ \cup \\
               & & \setcomp{\tup{q_i, \tile}
                            \xrightarrow{\repl{i}{\bit}{\bit'}}
                            \tup{p_{i+1}, \tile}}
                           {1 \leq i \le n \land
                            \bit \neq \bit' \in \set{0,1}} \ \cup \\
               & & \setcomp{p_i
                            \xrightarrow{\repl{i}{\bit}{\bit'}}
                            p_{i+1}}
                           {2 \leq i \le n \land \bit, \bit' \in \set{0,1}} \ \cup \\
               & & \setcomp{\tup{p_i, \tile}
                            \xrightarrow{\repl{i}{\bit}{\bit'}}
                            \tup{p_{i+1}, \tile}}
                           {2 \leq i \le n \land \bit, \bit' \in \set{0,1}} \ \cup \\
               & & \setcomp{q_{n+1} \xrightarrow{\tile} \tup{q_1, \tile}}
                           {\tile \in \tiles} \ \cup \\
               & & \setcomp{\tup{q_{n+1}, \tile}
                            \xrightarrow{\tile'}
                            \tup{q_1, \tile'}}
                           {\tup{\tile, \tile'} \in \vrel} \ \cup \\
               & & \setcomp{\tup{p_{n+1}, \tile}
                            \xrightarrow{\tile'}
                            \tup{q_1, \tile}}
                           {\tile, \tile' \in \tiles} \\
        \\
        F &=& \set{q_0}
    \end{array}
\]
The first set of transitions in $\delta$ are the reset transitions.
From all states (including $q_0$) a $\resetchar$ will bring the automaton back to its initial state.
Since we enforce separately that the value of $x_0$ begins and ends with $\resetchar$, it will remain true through the concatenations and $\replaceall$ operations that the value of $x_i$ still begins and ends with $\resetchar$.
Thus, the only final state that can occur on a solution to the string constraints is $q_0$ and the initial state is immaterial.

Next, we define the transitions over $\rowdelim$ which simply track the beginning of a new row.
Since we enforce separately that the input is of the correct format, we can simply use a self-loop on $\tup{q_1, \tile}$ to pass over $\rowdelim$ -- it will only be used once.

The next transitions read
$\repl{i}{\bit}{\bit}$
characters.
This means that in
$L_{\bit_1 \ldots \bit_n}$
we are reading the value $\bit = \bit_i$ at position $i$.
Thus, this is part of the sequence of bits identifying the column
$L_{\bit_1 \ldots \bit_n}$
is checking.
Thus, we continue reading the sequence of bits in the $q$ states.

Conversely, we next deal with the case when
$\repl{i}{\bit}{\bit'}$
with $\bit \neq \bit'$ is being read.
That is, we are reading a column that is not the
$\bit_1\ldots\bit_n$th
and hence we simply skip over it by using the $p$ states.

Finally, we read the tiles.
If we are in state $q_n$ then we are reading the
$\bit_1\ldots\bit_n$th
column, but this is the first row (after a reset) and thus there is no matching relation to verify.
If we are in state $\tup{q_{n+1}, \tile}$ then we are reading the
$\bit_1\ldots\bit_n$th
column, but this is not the first row and we have to verify the matching relation.
Wherefrom, there is only a transition reading $\tile'$ if
$\tup{\tile, \tile'} \in \vrel$.
The tile $\tile'$ is then saved for comparison with the next tile.

Otherwise, if we are in a state $p_{n+1}$ or $\tup{p_{n+1}, \tile}$ then we are not in the correct column and there is nothing to verify.
In this case we simply skip over the tile and continue.

\subsection{Completing the Proof}

Given a tiling problem over a corridor of width $2^n$, we have given a (polynomially sized in $n$)
$\strline[\replaceall]$
constraint $\varphi$ using languages
$L$, $L_\vrel$, $L_\hrel$, $L_1$, \ldots, $L_n$
(representable by automata polynomially sized in $n$)
such that $\varphi$ is satisfiable iff the variable $x_1$ takes on a value
\[
    \resetchar
    \rowdelim \nbit{1}{0} \ldots \nbit{n}{0} \tile^1_1
              \nbit{1}{0} \ldots \nbit{n}{1} \tile^1_2
              \ldots
              \nbit{1}{1} \ldots \nbit{n}{1} \tile^1_\tilewidth
    \rowdelim \ldots
    \rowdelim \nbit{1}{0} \ldots \nbit{n}{0} \tile^\tileheight_1
              \ldots
              \nbit{1}{1} \ldots \nbit{n}{1} \tile^\tileheight_\tilewidth
    \rowdelim
    \resetchar
\]
such that
\[
    \rowdelim \tile^1_1 \ldots \tile^1_\tilewidth
    \rowdelim \tile^2_1 \ldots \tile^2_\tilewidth
    \rowdelim \ldots
    \rowdelim \tile^\tileheight_1 \ldots \tile^\tileheight_\tilewidth
    \rowdelim
\]
is a solution to the tiling problem.
Since such a tiling problem is $2^n$-SPACE-hard, we have shown \expspace-hardness of
the satisfiability problem for $\strline[\replaceall]$. Moreover, this holds even in the single-letter case.



\section{Proof of Theorem~\ref{thm:length}} \label{app:thmlength}
	
\begin{proof}
The proof follows the same line as that of Theorem~\ref{thm-ftconrev}. The decision procedure  consists of four steps, where the first two steps are the same as those in the proof of Theorem~\ref{thm-ftconrev}. 
		
The third step is adapted by applying the following additional transformation to length assertions: When splitting each variable $x$ into multiple fresh variables, say, $x_1 \ldots x_k$, replace all occurrences of $|x|$ (resp. $|x|_a$)  in length assertions by $\sum \limits_{1 \le i \le k} |x_i|$ (resp. $\sum  \limits_{1 \le i \le k}|x_i|_a$). 
		
The fourth step is adapted by applying the  following additional transformation to length assertions: For each variable $x \in X$, if $x^{(r)}$ occurs in the assignments of $S_4$, then replace each occurrence of $|x|$ (resp. $|x|_a$) with $|x^{(r)}|$ (resp. $|x^{(r)}|_a$), otherwise, replace each occurrence of $|x^{(r)}|$ (resp. $|x^{(r)}|_a$) with $|x|$ (resp. $|x|_a$).
		
It follows that we obtain a symbolic execution of polynomial size that contains only \FT{}s, regular constraints and length assertions. The \expspace{} upper bound follows from \cite[Theorem~12]{LB16}.
	\end{proof}

}

\end{document}